\RequirePackage{rotating}
\RequirePackage{amsmath}
\documentclass[smallextended]{svjour3}
\usepackage{amssymb}
\usepackage{stmaryrd}
\usepackage{listings}
\usepackage{rotating}
\usepackage{url}
\usepackage{framed}
\usepackage{array}
\usepackage{threeparttable}
\usepackage{verbatim}
\usepackage{tikz}
\usepackage{multirow}
\usepackage{tikz}
\usepackage{varioref}
\usepackage{hhline}
\usepackage{afterpage}
\usepackage{ifthen}
\usetikzlibrary{decorations,arrows,shapes}
\usepackage{thmtools}
\usepackage{thm-restate}

\declaretheorem[name=Theorem]{mytheorem}

\smartqed

\date{}
\journalname{}
\renewcommand{\makeheadbox}{}

\allowdisplaybreaks[1]

\newcommand*{\cfg}[2]{\mathop{\mathrm{cfg}_{#1}\llbracket #2 \rrbracket}\nolimits}
\newcommand*{\bcfg}[2]{\mathop{\mathrm{cfg}_{#1}\bigl\llbracket #2 \bigr\rrbracket}\nolimits}
\newcommand*{\bodycfg}[2]{\mathop{\mathrm{cfg}^\mathrm{b}_{#1}\llbracket #2 \rrbracket}\nolimits}
\newcommand*{\gcfg}[2]{\mathop{\mathrm{cfg}^1_{#1}\llbracket #2 \rrbracket}\nolimits}

\newcommand*{\tv}[1]{\mathop{\mathrm{tv}_{#1}}\nolimits}
\newcommand*{\true}{1}
\newcommand*{\false}{0}
\newcommand*{\maybe}{\mathrm{?}}
\newcommand*{\CFG}{\mathrm{CFG}}
\newcommand*{\Id}{\mathrm{Id}}
\newcommand*{\id}{\mathrm{id}}
\newcommand*{\df}{\mathrm{df}}
\newcommand*{\cs}{\mathrm{cs}}
\newcommand*{\tb}{t_\mathrm{b}}
\newcommand*{\tc}{t_\mathrm{c}}
\newcommand*{\Ms}{M^\mathrm{s}}
\newcommand*{\Mg}{M^\mathrm{g}}

\newcommand*{\Ft}{\mathsf{ft}}
\newcommand*{\Ftout}{\mathsf{ft}_\mathsf{out}}
\newcommand*{\St}{\mathsf{st}}
\newcommand*{\Gt}{\mathsf{gt}}
\newcommand*{\Gtout}{\mathsf{gt}_\mathsf{out}}
\newcommand*{\Bp}{\mathsf{bp}}
\newcommand*{\Cp}{\mathsf{cp}}
\newcommand*{\Rp}{\mathsf{rp}}

\newcommand*{\apn}{\alpha}

\newcommand*{\napn}{\tau}
\newcommand*{\bapn}{\beta}
\newcommand*{\capn}{\chi}
\newcommand*{\gapn}{\gamma}
\newcommand*{\rapn}{\rho}
\newcommand*{\stapn}{\sigma}
\newcommand*{\ftapn}{\nu}
\newcommand*{\pgapn}{\phi}

\newcommand*{\orig}{\mathrm{orig}}

\newcommand*{\dapn}{\delta}

\newcommand*{\NPATH}{\mathrm{NPATH}}
\newcommand*{\npe}{\mathrm{NP_E}}
\newcommand*{\nps}{\mathrm{NP_S}}

\newcommand*{\apc}[2]{\mathop{\mathrm{apc}_{#1}\llbracket #2 \rrbracket}\nolimits}
\newcommand*{\bapc}[2]{\mathop{\mathrm{apc}_{#1}\bigl\llbracket #2 \bigr\rrbracket}\nolimits}
\newcommand*{\bodyapc}[2]{\mathop{\mathrm{apc}^\mathrm{b}_{#1}\llbracket #2 \rrbracket}\nolimits}

\newcommand*{\tp}[1]{\mathop{\mathbf{t}_{#1}}\nolimits}
\newcommand*{\fp}[1]{\mathop{\mathbf{f}_{#1}}\nolimits}
\newcommand*{\pp}[1]{\mathop{\mathbf{p}_{#1}}\nolimits}
\newcommand*{\ttp}[1]{\mathop{\mathbf{tt}_{#1}}\nolimits}
\newcommand*{\tfp}[1]{\mathop{\mathbf{tf}_{#1}}\nolimits}
\newcommand*{\ffp}[1]{\mathop{\mathbf{ff}_{#1}}\nolimits}
\newcommand*{\ppp}[1]{\mathop{\mathbf{pp}_{#1}}\nolimits}

\newcommand*{\fund}[3]{\mathord{#1}\colon#2\to#3}
\newcommand*{\pard}[3]{\mathord{#1}\colon#2\rightarrowtail#3}

\providecommand*{\Nset}{\mathbb{N}}            

\newcommand*{\stm}{\mathrm{stm}}
\newcommand*{\Stm}{\mathrm{Stm}}
\newcommand*{\Exp}{\mathrm{Exp}}
\newcommand*{\Lab}{\mathrm{Lab}}
\newcommand*{\cand}{\mathbin{\mathtt{\&\&}}}
\newcommand*{\cor}{\mathbin{\mathtt{||}}}
\newcommand*{\cqmk}{\mathrel{\mathtt{?}}}
\newcommand*{\ccol}{\mathrel{\mathtt{:}}}
\newcommand*{\cqmkcol}{\mathbin{\mathord{\mathtt{?}}\mathord{\mathtt{:}}}}
\newcommand*{\cnot}{\mathop{\mathtt{!}}\nolimits}
\newcommand*{\uop}{\mathop{\mathrm{uop}}}
\newcommand*{\bop}{\mathbin{\mathrm{bop}}}

\newcommand*{\samebody}{\mathrm{sb}}

\newcommand*{\pg}{\mathrm{pg}}

\newcommand{\defeq}{\mathrel{\mathord{:}\mathord{=}}}

\newcommand{\summary}[1]{\textrm{\textbf{\textup{#1}}}}

\newcommand*{\Cplusplus}{{C\nolinebreak[4]\hspace{-.05em}\raisebox{.4ex}{\tiny\bf ++}}}

\renewcommand{\emptyset}{\varnothing}

\newcommand*{\sseq}{\subseteq}

\newcommand*{\union}{\cup}

\newcommand{\sset}[2]{{\renewcommand{\arraystretch}{1.2}
                      \left\{\,#1 \,\left|\,
                               \begin{array}{@{}l@{}}#2\end{array}
                      \right.   \,\right\}}}

\newcommand{\st}{\mathrel{.}}

\newcommand*{\bnf}{\mathrel{\mathord{:}\mathord{:}\mathord{=}}}
\newcommand*{\stick}{\mathrel{|}}

\newcommand{\proofsec}[1]{\par\textbf{#1}:}

\begin{document}

\title{The ACPATH Metric:
       Precise Estimation of the Number of Acyclic Paths in C-like Languages
      }
\author{Roberto Bagnara
\and Abramo Bagnara
\and Alessandro Benedetti
\and Patricia M.\ Hill
}
\authorrunning{R.~Bagnara,
               A.~Bagnara,
               A.~Benedetti,
               P.~M.~Hill
}
\titlerunning{The ACPATH Metric}

\institute{R.~Bagnara
    \at Department of Mathematical, Physical and Computer Sciences
    University of Parma,
    Italy, \\
    \email{{\it name.surname}@unipr.it}
  \and
    R.~Bagnara, A.~Bagnara, A.~Benedetti, and P.~M.~Hill
    \at BUGSENG srl,
    \url{http://bugseng.com},
    Italy, \\
    \email{{\it name.surname}@bugseng.com}
\and
    A.~Benedetti
    \at University of Parma,
    Italy
}

\maketitle

\begin{abstract}
  \emph{NPATH} is a metric introduced by Brian A.~Nejmeh in \cite{Nejmeh88}
  that is aimed at overcoming some important limitations of McCabe's
  cyclomatic complexity metric.
  Despite the fact that the declared NPATH objective is to
  count the number of acyclic execution paths through a function,
  the definition given for the C~language in \cite{Nejmeh88}
  fails to do so even for very simple programs.
  We show that counting the number of acyclic paths in CFG
  is unfeasible in general.  Then we define a new metric
  for C-like languages, called \emph{ACPATH}, that allows to quickly
  compute a very good estimation of the number of acyclic execution
  paths through the given function.  We show that, if the function body
  does not contain backward gotos and does not contain jumps into a loop
  from outside the loop, then such estimation is actually exact.
\end{abstract}

\section{Introduction}
\label{sec:introduction}

\emph{Software testing} is a process whereby software components or entire
systems are executed so as to gather information about their behavior.
Although a common expected outcome of software testing is the
identification of defects, testing and debugging are two quite different
processes \cite{Myers79}: while the latter is a development activity,
the former is one of the methodologies for software verification and
validation.

Despite the increasing adoption of formal methods and static verification
techniques, software testing is still the most used verification technique
in several industrial sectors costing as much as 50\% or
even 75\% of the total development costs \cite{HailpernS02}.

One of the problems of software testing is the need for adequate
\emph{test suites}, i.e., collections of so-called \emph{test cases},
each consisting of input and predicted output data.  Software testing
can only be considered to be acceptable as a verification methodology
when the available test suites exercise a significant portion of the
code and input space of the component under test.  Therefore, the
feasibility of meaningfully testing a system and its subsystems with a
test suite of manageable size and cost is a qualitative attribute,
called \emph{testability}, of the system/subsystems.

Testing at the unit level (a.k.a.\ \emph{unit testing}, i.e., testing
individual functions or small groups of functions) is of particular
importance as it often allows for the early detection of problems,
when the cost of fixing them is much lower than if they are found
during \emph{integration testing} (i.e., when the integration of
different units are tested as a whole to assess their ability to work
together).
Hence, an important reason for limiting the structural complexity of
software units is to facilitate unit testing,
i.e., to improve testability of the units by limiting the sizes
of the unit test suites and the intellectual effort of obtaining them.

The \emph{NPATH} metric was introduced by Brian A.~Nejmeh in \cite{Nejmeh88}
in order to automatically quantify the testability of individual procedures
or functions, yet addressing the shortcomings of McCabe \emph{cyclomatic
complexity}, another metric meant to quantify testability
\cite{Harrison84,McCabe76}.
According to \cite{Nejmeh88}, the shortcomings of cyclomatic complexity
are:\footnote{Cyclomatic complexity was criticized also by several
other authors:
see, e.g., \cite{GillK90,GillK91,Shepperd88}.}
\begin{itemize}
\item
  the number of acyclic paths  in the control-flow graph of a procedure
  varies from a linear to an exponential function of the cyclomatic complexity
  number \cite{Evangelist84b};  as the fraction of acyclic paths covered
  by a test suite is an important measure of adequacy of the test suite,
  it turns out that the cyclomatic complexity number has little correlation
  with the testing effort;
\item
  cyclomatic complexity does not distinguish between different kinds of
  control-flow structures (e.g., between conditional and iteration
  statements) whereas such distinction is important in the assessment
  of testability;
\item
  cyclomatic complexity does not take into account the way
  control-flow structures are possibly nested with one another
  (e.g., two disjoint while loops give rise to the same cyclomatic complexity
  number as two nested while loops) \cite{CurtisSMBL79};
  again, this distinction is relevant as far as testability is concerned.
\end{itemize}

While the declared intent of NPATH is to count the number of acyclic
paths through a function, the definition given for the C~language
in \cite{Nejmeh88} fails to do so, as shown by the following example:

\begin{example}
\label{ex:NPATH-buggy-original}
Consider the following C function:
\begin{verbatim}
  int f(int a, int b, int c, int d, int e) {
    if (a && b && c)
      return d ? 0 : 1;
    else
      return e ? 0 : 1;
  }
\end{verbatim}
The algorithm given in \cite{Nejmeh88} gives
$\mathrm{NPATH} = 2 + 2 + 2 = 6$,
but this is neither the number of possible paths within \verb+f()+
nor an upper bound to it.
In fact the number of possible paths is 8, corresponding to the following
combinations:
\begin{enumerate}
\item
\verb"a && b && c && d"
\item
\verb"a && b && c && !d"
\item
\verb"a && b && !c && e"
\item
\verb"a && b && !c && !e"
\item
\verb"a && !b && e"
\item
\verb"a && !b && !e"
\item
\verb"!a && e"
\item
\verb"!a && !e"
\end{enumerate}
\end{example}

Section~\ref{sec:ACPATH} presents a new metric for C-like
languages that demonstrably corresponds, under some conditions
that are often satisfied, to the number of acyclic paths through the function.

The plan of the paper is as follows:
Section~\ref{sec:preliminaries} introduces preliminary notions
and notations;
Section~\ref{sec:NPATH} recalls the NPATH metric for C-like
languages as defined in \cite{Nejmeh88} highlighting the difference
between what it is meant to measure and what it really measures;
Section~\ref{sec:ACPATH} presents the new ACPATH metric for C-like
languages;
Section~\ref{sec:implementation-experimental-evaluation}
presents the results of an experimental
evaluation that studies, on a number of real-world projects,
the relationship between ACPATH and NPATH;
Section~\ref{sec:conclusion} discusses
the contribution of the present paper,
some related work, and then concludes.

\section{Preliminaries}
\label{sec:preliminaries}

In this section we introduce the preliminary notions such as
\emph{control flow graphs} and \emph{acyclic paths}, notations
including the abstract C syntax used in the paper and some formal
definitions needed to prove the theoretical results of
Section~\ref{sec:ACPATH}.

\subsection{Control Flow Graphs}
\label{sec::control-flow-graphs}

A control flow graph (CFG) is an abstraction of the computation paths
of a procedure.
\begin{definition} \summary{(Control flow graph.)}
A \emph{control flow graph} $G$ is a triple $(N, A,s)$
where $(N, A)$ is a directed graph, 
hence $A \subseteq N \times N$,
and $s \in N$ is called the \emph{entry node of $G$}.
A node $n \in N$ such that $n$ has no successor in $A$
(i.e., for each $(x, y) \in A$ we have $x \neq n$)
is called an \emph{exit node of $G$}.
\end{definition}

A node in the graph represents either a basic block of code
(i.e., a sequence of statements where control flows from
the beginning to the end of the sequence) or a branch point
in the procedure.
The entry node represents the procedure's entry point and
each exit node represents an exit point of the procedure.
An arc represents possible flow control.

\subsection{Acyclic Paths in a CFG}
\label{sec::acyclic-paths}

An acyclic path in a CFG is a path from the entry node to a
target node that does not traverse an arc more than once.
\begin{definition} \summary{(Acyclic path.)}
An \emph{acyclic path} in a control flow graph $(N, A, s)$ is any
sequence of nodes of the form $n_0$, \dots,~$n_{k-1}$ such that
$n_0 = s$
and, if $M \defeq \bigl\{\, (n_{i-1}, n_i) \bigm| i = 1, \dots,~k \,\bigr\}$,
then $M \sseq A$ and $|M| = k$.
\end{definition}

Let $G = (N, A, s)$ be a CFG and $t \in N$ be a target node.
The number of acyclic paths in $G$ leading to $t$,
denoted by $\napn(G, t)$, can be computed as follows:
\begin{align}
  \napn(G, t)
    &\defeq
      \napn(s, A, t), \\
  \napn(n, A, t)
    &\defeq
      \begin{cases}
        1,
          &\text{if $n = t$;} \\
        \sum\limits_{(n, m) \in A} \napn\bigl(m, A \setminus \{ (n, m) \}, t\bigr),
          &\text{otherwise.}
      \end{cases}
\end{align}
If we denote by $e(G)$ the set of exit nodes of $G$,
the number of acyclic paths in $G$,
denoted by $\apn(G)$, is given by
\begin{equation}
  \apn(G)
    \defeq
  \sum\limits_{t \in e(G)} \napn(s, A, t).
\end{equation}

\subsection{C Abstract Syntax}
\label{sec::c-abstract-syntax}

The abstract syntax of expressions considered in this paper, which is inspired
by the one used in the \emph{Clang} C language family compiler front-end,
is approximated by the following grammar in BNF form:
\begin{align*}
  \Exp \ni
  E &\bnf x \stick k \stick \mathrm{ICE}
        \stick \cnot E_1 \stick \mathop{\mathtt{+}}E_1
        \stick \mathop{\mathtt{-}}E_1 \stick (E_1)
        \stick \mathop{(\mathrm{type})} E_1 \stick \uop E_1 \\
      &\stick E_1 \cand E_2
        \stick E_1 \cor E_2
        \stick E_1 , E_2
        \stick E_1 \cqmkcol E_2
        \stick E_1 \bop E_2
        \stick E_1 \cqmk E_2 \ccol E_3
\intertext{%
where $x$ is a variable, $k$ is an integer literal,
`$\mathrm{ICE}$' is any Integer Constant Expression, that is, a non-literal
expression that can be evaluated to a constant at compile time
(e.g., $3+4$),
`$\mathop{(\mathrm{type})}$' represents (implicit and explicit)
cast operators, `$\mathrm{uop}$' and `$\mathrm{bop}$' are any unary or binary
operators except those already considered.
The abstract syntax of commands is approximated by the following grammar:
}
  \Stm \ni
    S &\bnf E; \stick S_1 \; S_2 \stick \mathtt{return}
       \stick \mathtt{return} \; E
       \stick \mathtt{if} \; (E) \; S_1 \; \mathtt{else} \; S_2 \\
      &\stick \mathtt{if} \; (E) \; S_1
       \stick \mathtt{switch} \; (E) \; S
       \stick \mathtt{while} \; (E) \; S
       \stick \mathtt{do} \; S \; \mathtt{while} \; (E) \\
      &\stick \mathtt{for} \; (E_1; E_2; E_3) \; S
       \stick \mathtt{break}
       \stick \mathtt{continue}
       \stick \mathtt{goto} \; \id
       \stick L : S \stick \{S\}
       \stick \mathrm{stm}
\intertext{%
where $\mathrm{stm}$ generates any command except those already considered
and $L : S$ is a labeled statement:
}
    \Lab \ni
      L &\bnf \mathtt{} \; \mathtt{case} \; z \stick \mathtt{default} \stick \id
\end{align*}
where $L$ is a \emph{label}, $z$ is an ICE and $\id$ is a C identifier.

\subsection{From C Abstract Syntax To Control Flow Graphs}
\label{sec::from-c-abstract-syntax-to-control-flow-graphs}

For any procedure and hence any abstract command $\mathrm{Stm}$ that
represents it, the actual control flow and the corresponding CFG will
depend on the compiler and on the selected optimization level and
capabilities.
For the purposes of this paper, it suffices to define a notion of
``reference CFG''
and to restrict the possible optimization levels to three:
\begin{enumerate}
\label{enum:optimization}
\item[0]
\label{enum:optimization-0}
no optimization at all,
\item[1]
\label{enum:optimization-1}
branch removal via Boolean interpretation of each integer
literal, and
\item[2]
\label{enum:optimization-2}
branch removal via the Boolean interpretation of each ICE.
\end{enumerate}
Note that, if the metrics we are after are meant to measure
testability only, there is no difference between a constant literal and
an ICE.  An alternative point of view is that NPATH's purpose is to
evaluate also readability and maintainability and, in this case, an ICE
can be considered as an ordinary compound expression.  In this paper, we
wish to support both views, whence the parametrization
on the optimization level.

The following definition gives, by structural induction on the
abstract C syntax of Section~\ref{sec::c-abstract-syntax}, the
reference CFG for any function body. The definition can be skipped
unless the reader wishes to check the proofs of the theorems.
Appendix~\vref{app:example-reference-cfgs} provides several examples
showing, by means of figures, the definition at work.

\begin{definition} \summary{(Reference CFG for the C language.)}
\label{def:reference-cfg}
Let $\CFG$ denote the set of all CFGs where the nodes are natural numbers
and let $i \in \{ 0, 1, 2 \}$ denote the three optimization levels.

To define the reference CFG for expressions we first define the function
$\fund{\tv{i}}{\Exp}{\{ 0, 1, \maybe  \}}$ that returns, for each $E \in \Exp$,
a three-valued Boolean defined as follows, where $E \mapsto b$ means
``$E$ evaluates to $b$'':
\begin{equation}
  \tv{i}(E)
    \defeq
      \begin{cases}
        b,
          &\text{if $E \mapsto b$ and $i = 1$ and $E$ is an integer literal;} \\
        b,
          &\text{if $E \mapsto b$ and $i = 2$ and $E$ is an ICE;} \\
        \maybe,
          &\text{otherwise.}
      \end{cases}
\end{equation}

The function
\[
\fund{\cfg{i}{\cdot}}{\Exp \times \Nset^3}{\CFG \times \Nset}
\]
 is defined as follows: whenever
$\cfg{i}{E}(t, f, m) = \langle G, m' \rangle$, then
 $G = (N, A, s) \in\CFG$, where the nodes are
 $N \sseq \{ t, f \} \union [m, m'-1]$ and
$t$ (resp., $f$) is reached from $s$ if $E$ evaluates to
 true (resp., false).

\paragraph{Variables:}
\begin{equation}
\label{eq:cfg_i-variables}
  \cfg{i}{x}(t, f, m)
    \defeq
      \bigl\langle \bigl(\{m, t, f\}, \{ (m,t), (m,f) \}, m\bigr), m+1 \bigr\rangle.
\end{equation}

\paragraph{Constants:}
if $E = k$ or $E = \mathrm{ICE}$,
\begin{equation}
\label{eq:cfg_i-constants}
  \cfg{i}{E}(t, f, m)
    \defeq
      \begin{cases}
        \bigl\langle \bigl(\{t\}, \emptyset, t\bigr), m \bigr\rangle,
          &\text{if $\tv{i}(E) = 1$;} \\
        \bigl\langle \bigl(\{f\}, \emptyset, f\bigr), m \bigr\rangle,
          &\text{if $\tv{i}(E) = 0$;} \\
        \bigl\langle \bigl(\{m, t, f\}, \{ (m,t), (m,f) \}, m\bigr), m+1 \bigr\rangle,
          &\text{otherwise.}
      \end{cases}
\end{equation}

\paragraph{Logical negation:}
\begin{equation}
\label{eq:cfg_i-negExpr}
  \cfg{i}{\cnot E_1}(t, f, m)
    \defeq
      \cfg{i}{E_1}(f, t, m),
\end{equation}
where $E = \cnot E_1$.

\paragraph{Unary plus and minus, parentheses and cast expressions:}
\begin{equation}
\label{eq:cfg_i-specUnaryExpr}
\cfg{i}{E}(t, f, m)
    \defeq
      \cfg{i}{E_1}(t, f, m),
\end{equation}
where
$E \in \bigl\{ +E_1, -E_1, (E_1), (type)E_1 \bigr\}$.

\paragraph{Other unary operators:}
\begin{equation}
\label{eq:cfg_i-unaryExpr}
  \cfg{i}{\uop E_1}(t, f, m)
    \defeq
      \bigl\langle (N, A, s_1), m_1 \bigr\rangle,
\end{equation}
where $\uop$ is a unary operator not already considered,
$N \defeq N_1 \union \{m, t, f\}$,
$A \defeq A_1 \union \bigl\{(m,t), (m,f)\bigr\}$, and
$\cfg{i}{E_1}(m, m, m+1) = \bigl\langle (N_1, A_1, s_1), m_1 \bigr\rangle$.

\paragraph{Logical conjunction:}
\begin{equation}
\label{eq:cfg_i-conjExpr}
  \cfg{i}{E_1 \cand E_2}(t, f, m)
    \defeq
      \begin{cases}
        \bigl\langle \bigl(\{f\}, \emptyset, f\bigr), m \bigr\rangle,
          &\text{if $\tv{i}(E_1) = \false$,} \\
        \cfg{i}{E_2}(t, f, m),
          &\text{if $\tv{i}(E_1) = \true$,} \\
        \bigl\langle (N, A, s_1), m_2 \bigr\rangle,
          &\text{otherwise,}
      \end{cases}
\end{equation}
where
$N \defeq N_1 \union N_2$, $A \defeq A_1 \union A_2$,
$\cfg{i}{E_2}(t, f, m) = \bigl\langle (N_2, A_2, s_2), m_1 \bigr\rangle$, and
$\cfg{i}{E_1}(s_2, f, m_1) =  \bigl\langle (N_1, A_1, s_1), m_2 \bigr\rangle$.

\paragraph{Logical disjunction:}
\begin{equation}
\label{eq:cfg_i-disjExpr}
  \cfg{i}{E_1 \, \cor \, E_2}(t, f, m)
    \defeq
      \begin{cases}
        \bigl\langle \bigl(\{t\}, \emptyset, t\bigr), m \bigr\rangle,
          &\text{if $\tv{i}(E_1) = \true$,} \\
        \cfg{i}{E_2}(t, f, m),
          &\text{if $\tv{i}(E_1) = \false$,} \\
        \bigl\langle (N, A, s_1), m_2 \bigr\rangle,
          &\text{otherwise,}
      \end{cases}
\end{equation}
where
$N \defeq N_1 \union N_2$, $A \defeq A_1 \union A_2$,
$\cfg{i}{E_2}(t, f, m) = \bigl\langle (N_2, A_2, s_2), m_1 \bigr\rangle$ and
$\cfg{i}{E_1}(t, s_2, m_1) = \bigl\langle (N_1, A_1, s_1), m_2 \bigr\rangle$.

\paragraph{Comma operator:}
\begin{equation}
\label{eq:cfg_i-commaExpr}
  \cfg{i}{E_1, E_2}(t, f, m)
    \defeq
      \bigl\langle (N, A, s_1), m_2 \bigr\rangle,
\end{equation}
where
$N \defeq N_1 \union N_2$,
$A \defeq A_1 \union A_2$,
$\cfg{i}{E_2}(t, f, m) = \bigl\langle (N_2, A_2, s_2), m_1 \bigr\rangle$ and
$\cfg{i}{E_1}(s_2, s_2, m_1) = \bigl\langle (N_1, A_1, s_1), m_2 \bigr\rangle$.

\paragraph{Binary conditional operator:}
\begin{equation}
\label{eq:cfg_i-binCondExpr}
  \cfg{i}{E_1 \cqmkcol \; E_2}(t, f, m)
    \defeq
      \begin{cases}
        \bigl\langle \bigl(\{t\}, \emptyset, t\bigr), m \bigr\rangle,
          &\text{if $\tv{i}(E_1) = \true$,} \\
        \cfg{i}{E_2}(t, f, m),
          &\text{if $\tv{i}(E_1) = \false$,} \\
        \bigl\langle (N, A, s_1), m_2 \bigr\rangle,
          &\text{otherwise,}
      \end{cases}
\end{equation}
where
$N \defeq N_1 \union N_2$,
$A \defeq A_1 \union A_2$,
$\cfg{i}{E_2}(t, f, m) = \bigl\langle (N_2, A_2, s_2), m_1 \bigr\rangle$ and
$\cfg{i}{E_1}(t, s_2,m_1) = \bigl\langle (N_1, A_1, s_1), m_2 \bigr\rangle$.

\paragraph{Other binary operators:}
\begin{equation}
\label{eq:cfg_i-binExpr}
  \cfg{i}{E_1 \bop E_2}(t, f, m)
    \defeq
  \bigl\langle (N, A, s_1), m_2 \bigr\rangle,
\end{equation}
where
$N \defeq N_1 \union N_2 \union \{m, t, f\}$,
$A \defeq A_1 \union A_2 \union \bigl\{(m,t), (m,f)\bigr\}$, and
\begin{align*}
  \cfg{i}{E_2}(m, m, m+1) &= \bigl\langle (N_2, A_2, s_2), m_1 \bigr\rangle, \\
  \cfg{i}{E_1}(s_2, s_2, m_1) &= \bigl\langle (N_1, A_1, s_1), m_2 \bigr\rangle,
\end{align*}

\paragraph{Conditional operator:}
\begin{equation}
\label{eq:cfg_i-condExpr}
  \cfg{i}{E_1 \cqmk \, E_2 \ccol E_3}(t, f, m)
    \defeq
      \begin{cases}
        \cfg{i}{E_2}(t, f, m),
          &\text{if $\tv{i}(E_1) = \true$,} \\
        \cfg{i}{E_3}(t, f, m),
          &\text{if $\tv{i}(E_1) = \false$,} \\
        \bigl\langle (N, A, s_1), m_3 \bigr\rangle,
          &\text{otherwise,}
      \end{cases}
\end{equation}
where
$N \defeq N_1 \union N_2 \union N_3$,
$A \defeq A_1 \union A_2 \union A_3$ and
\begin{align*}
  \cfg{i}{E_2}(t, f, m) &= \bigl\langle (N_2, A_2, s_2), m_1 \bigr\rangle, \\
  \cfg{i}{E_3}(t, f, m_1) &= \bigl\langle (N_3, A_3, s_3), m_2 \bigr\rangle, \\
  \cfg{i}{E_1}(s_2, s_3,m_2) &= \bigl\langle (N_1, A_1, s_1), m_3 \bigr\rangle.
\end{align*}

Before defining the reference CFG for statements, we define a special form
that will be used for labeled statements:
\[
  \fund{\cfg{i}{\cdot}}
       {\Lab \times \Nset^2}
       {\CFG
         \times \wp\bigl((\Id \union \{\cs, \df\}) \times \Nset\bigr)
          \times \Nset}
\]
where $\Nset_\bot \defeq \Nset \union \{\bot\}$ and
$\Id$ denotes the set of identifiers in the C language.
If $\cfg{i}{L}(t, m) = \langle G, \Ms,  m' \rangle$,
then:
\begin{itemize}
\item
  $G = (N, A, s) \in \CFG$ and
  $N \sseq \{ t \} \union [m, m'-1]$;
\item
  $t \in N$ is reached if/when the execution of $L$ terminates;
\item
  $\Ms$ is a multimap associating elements of $\Id \union \{\cs, \df\}$
  (where $\cs$ incorporates all $\mathtt{case} \; z$, where $z$ is an ICE, and $\df$ stands for $\mathtt{default}$,
  respectively) to nodes, such that at most
  one occurrence of $\df$ is allowed:
  if $(\id, n) \in \Ms$, then $n$ is the node in $G$
  corresponding to a statement labeled with $\id$;
  if $(\cs, n) \in \Ms$, then $n$ is a node in $G$
  corresponding to a $\mathtt{case}$-labeled statement;
  if $(\df, n) \in \Ms$, then $n$ is the node in $G$
  corresponding to a $\mathtt{default}$-labeled statement;
\item
  $m,m' \in \Nset$ are, respectively, the lower and the upper bound of the nodes' labels introduced by $L$
\end{itemize}

\paragraph{Labels:}
\begin{equation}
\label{eq:cfg_i-labels}
  \bcfg{i}{L}(t, m)
    \defeq
      \begin{cases}
        \bigl\langle \bigl(\{m, t\},\bigl\{(m, t)\bigr\}, m\bigr), \bigl\{(\cs, m)\bigr\}, m+1 \bigl\rangle,
          &\text{if $L = \mathtt{case} \; z$,} \\
        \bigl\langle \bigl(\{m, t\},\bigl\{(m, t)\bigr\}, m\bigr), \bigl\{(\df, m)\bigr\}, m+1 \bigl\rangle,
          &\text{if $L = \mathtt{default}$,} \\
        \bigl\langle \bigl(\{m, t\},\bigl\{(m, t)\bigr\}, m\bigr), \bigl\{(\id, m)\bigr\}, m+1 \bigl\rangle,
          &\text{if $L = \id$.}
      \end{cases}
\end{equation}

We can now define the reference CFG for statements
\[
  \cfg{i}{\cdot} \colon
    \Stm \times \Nset \times \Nset_\bot^2 \times \Nset
    \rightarrow
      \CFG
      \times \wp\bigl((\Id \union \{\cs, \df\}) \times \Nset\bigr)
      \times \wp(\Id \times \Nset)
      \times \Nset.
\]
If $\cfg{i}{S}(t, \tb, \tc, m) = \langle G, \Ms, \Mg, m' \rangle$,
then:
\begin{itemize}
\item
  $G = (N, A, s) \in \CFG$ and
  $N \sseq \{ t , \tb, \tc\} \union [m, m'-1]$;
\item
  $t \in N$ is reached if/when the execution of $S$ terminates;
\item
  $\tb \in N$ is reached if/when the execution of $S$ terminates
  because a $\mathtt{break}$ has been executed;
\item
  $\tc \in N$ is reached if/when the execution of $S$ terminates
  because a $\mathtt{continue}$ has been executed;
\item
  $\Ms$ is a multimap for the labels in $S$, as defined earlier;
\item
  $\Mg$ is a map associating the identifiers of $\mathtt{goto}$ statements
  in $S$ to their target node in $G$;
\item
  $m,m' \in \Nset$ are, respectively, the lower and the upper bound
  of the nodes' labels introduced by $S$.
\end{itemize}

\paragraph{Expression statement:}
\begin{equation}
\label{eq:cfg_i-exprStat}
  \cfg{i}{E;}(t, \tb, \tc, m)
    \defeq
  \bigl\langle (N, A, s), \emptyset, \emptyset, m_1 \bigr\rangle, \\
\end{equation}
where
$\cfg{i}{E}(t, t, m) = \bigl\langle (N, A, s), m_1 \bigr\rangle$.

\paragraph{Sequential composition:}
\begin{equation}
\label{eq:cfg_i-seqStat}
  \cfg{i}{S_1 \; S_2}(t, \tb, \tc, m)
    \defeq
      \bigl\langle (N_1 \union N_2, A_1 \union A_2, s),  \Ms, \Mg, m_2 \bigr\rangle, \\
\end{equation}
where
$\Ms \defeq \Ms_1 \union \Ms_2$,
$\Mg \defeq \Mg_1 \union \Mg_2$,
and
\begin{align*}
  \cfg{i}{S_2}(t, \tb, \tc, m)
    &= \bigl\langle (N_2, A_2, s_2), \Ms_2, \Mg_2, m_1\bigr\rangle, \\
  \cfg{i}{S_1}(s_2, \tb, \tc, m_1)
    &= \bigl\langle (N_1, A_1, s), \Ms_1, \Mg_2, m_2 \bigr\rangle.
\end{align*}

\paragraph{Return statement:}
\begin{equation}
\label{eq:cfg_i-returnStat}
  \bcfg{i}{\mathtt{return}}(t, \tb, \tc, m)
    \defeq
   \bigl\langle \bigl(\{m\}, \emptyset, m\bigl), \emptyset, \emptyset, m+1  \bigr\rangle,
\end{equation}

\paragraph{Return with expression statement:}
\begin{equation}
\label{eq:cfg_i-returnExprStat}
  \bcfg{i}{\mathtt{return} \; E}(t, \tb, \tc, m)
    \defeq
   \bigl\langle (N, A, s), \emptyset, \emptyset, m_1  \bigr\rangle,
\end{equation}
where
$N \defeq N_E \union \{m\}$ and
$\cfg{i}{E}(m, m, m+1) = \bigl\langle (N_E, A, s), m_1 \bigr\rangle$.

\paragraph{Conditional statement:}
\begin{multline}
\label{eq:cfg_i-condStat}
  \bcfg{i}{\mathtt{if} \; (E) \; S_1 \; \mathtt{else} \; S_2}(t, \tb, \tc, m) \\
    \defeq
      \begin{cases}
        \cfg{i}{S_1}(t, \tb, \tc, m),
          &\text{if $\tv{i}(E) = \true \land \Ms_2 = \emptyset$,} \\
        \cfg{i}{S_2}(t, \tb, \tc, m),
          &\text{if $\tv{i}(E) = \false \land \Ms_1 = \emptyset$,} \\
        \bigl\langle (N, A, s), \Ms, \Mg, m_3 \bigr\rangle,
          &\text{otherwise,}
      \end{cases}
\end{multline}
where
$N \defeq N_E \union N_1 \union N_2 \union \bigl\{ m, m_1, t \bigr\}$,
$A \defeq A_E \union A_1 \union A_2 \union \bigl\{ (m, t), (m_1, t) \bigr\}$,
$\Ms \defeq \Ms_1 \union \Ms_2$,
$\Mg \defeq \Mg_1 \union \Mg_2$, and
\begin{align*}
  \cfg{i}{S_2}(m, \tb, \tc, m+1) &= \bigl\langle (N_2, A_2, s_2), \Ms_2, \Mg_2, m_1\bigr\rangle, \\
  \cfg{i}{S_1}(m_1, \tb, \tc, m_1+1) &= \bigl\langle (N_1, A_1, s_1), \Ms_1, \Mg_1, m_2 \bigr\rangle, \\
  \cfg{i}{E}(s_1, s_2, m_2) &= \bigl\langle (N_E, A_E, s), m_3 \bigr\rangle.
\end{align*}

\paragraph{One-armed conditional statement:}
\begin{multline}
\label{eq:cfg_i-one-armedCondStat}
  \bcfg{i}{\mathtt{if} \; (E) \; S_1}(t, \tb, \tc, m) \\
    \defeq
      \begin{cases}
        \cfg{i}{S_1}(t, \tb, \tc, m),
          &\text{if $\tv{i}(E) = \true$,} \\
        \bigl\langle (\{ t \}, \emptyset, t), \emptyset, \emptyset, m \bigr\rangle,
          &\text{if $\tv{i}(E) = \false \land \Ms = \emptyset$,} \\
        \bigl\langle (N, A, s), \Ms, \Mg, m_2 \bigr\rangle,
          &\text{otherwise,}
      \end{cases}
\end{multline}
where
$N \defeq N_E \union N_1  \union \bigl\{ m, t \bigr\}$,
$A \defeq A_E \union A_1  \union \bigl\{ (m, t) \bigr\}$, and
\begin{align*}
  \cfg{i}{S_1}(m, \tb, \tc, m+1)
    &= \bigl\langle (N_1, A_1, s_1), \Ms, \Mg, m_1 \bigr\rangle, \\
  \cfg{i}{E}(s_1, t, m_1)
    &= \bigl\langle (N_E, A_E, s), m_2 \bigr\rangle.
\end{align*}

\paragraph{Switch statement:}
\begin{multline}
\label{eq:cfg_i-switchStat}
  \bcfg{i}{\mathtt{switch} \; (E) \; S}(t, \tb, \tc, m) \\
    \defeq
      \begin{cases}
        \bigl\langle (N, A_1, s), \Ms, \Mg, m_2 \bigr\rangle,
          &\text{if $(\df, n) \in \Ms_1$,} \\
        \bigl\langle (N, A_2, s), \Ms, \Mg, m_2 \bigr\rangle,
          &\text{otherwise,}
      \end{cases}
\end{multline}
where
$N \defeq N_E \union N_S \union \{m, m_1\}$ and
\begin{align*}
  A_1
    &\defeq
      A_E \union A_S \union \{(m, t)\} \\
        &\union \bigl\{\,
                 (m_1, n)
               \bigm|
                 \exists l \in \{\cs, \df\} \st (l, n) \in \Ms_1
               \,\bigr\}, \\
  A_2 &\defeq A_1 \union \bigl\{(m_1, m)\bigr\}, \\
  \Ms
    &\defeq
      \Ms_1 \setminus \bigl\{\, (l, n) \bigm| l \in \{\cs, \df\} \,\bigr\}, \\
  \cfg{i}{S}(m, m, \tc, m+1)
    &= \bigl\langle (N_S, A_S, s_S), \Ms_1, \Mg, m_1\bigr\rangle, \\
  \cfg{i}{E}(m_1, m_1, m_1+1)
    &= \bigl\langle (N_E, A_E, s), m_2 \bigr\rangle.
\end{align*}

\paragraph{While statement:}
\begin{equation}
\label{eq:cfg_i-whileStat}
  \bcfg{i}{\mathtt{while} \; (E) \; S}(t, \tb, \tc, m)
    \defeq
      \bigl\langle (N, A, s_E), \Ms, \Mg, m_2 \bigr\rangle,
\end{equation}
where
$N \defeq N_E \union N_S \union \bigl\{ m, m_1 \}$,
$A \defeq A_E \union A_S \union \bigl\{ (m, s_E), (m_1, s_S) \bigr\}$,
and
\begin{align*}
  \cfg{i}{S}(m, t, s_E, m+1)
    &= \bigl\langle (N_S, A_S, s_S), \Ms, \Mg, m_1 \bigr\rangle, \\
  \cfg{i}{E}(m_1, t, m_1+1)
    &= \bigl\langle (N_E, A_E, s_E), m_2 \bigr\rangle.
\end{align*}

\paragraph{Do-while statement:}
\begin{equation}
\label{eq:cfg_i-doWhileStat}
  \bcfg{i}{\mathtt{do} \; S \; \mathtt{while} \; (E) }(t, \tb, \tc, m)
    \defeq
  \bigl\langle (N, A, s_S), \Ms, \Mg, m_2+1 \bigr\rangle
\end{equation}
where
$N \defeq N_E \union N_S \union \{ m_1, m_2 \}$,
$A \defeq A_E \union A_S \union \bigl\{ (m_1, s_E), (m_2, s_S) \bigr\}$,
and
\begin{align*}
  \cfg{i}{E}(m_2, t, m) &= \bigl\langle (N_E, A_E, s_E), m_1 \bigr\rangle, \\
  \cfg{i}{S}(m_1, t, s_E, m_1+1) &= \bigl\langle (N_S, A_S, s_S), \Ms, \Mg, m_2 \bigr\rangle.
\end{align*}

\paragraph{For statement:}
\begin{multline}
\label{eq:cfg_i-forStat}
  \bcfg{i}{\mathtt{for} \; ( E_1; E_2; E_3) \; S}(t, \tb, \tc, m) \\
    \defeq
  \bcfg{i}{E_1; \; \mathtt{while} \; (E_2) \; \{S \; E_3;\}}(t, \tb, \tc, m).
\end{multline}

\paragraph{Break statement:}
assuming $\tb \neq \bot$,
\begin{equation}
\label{eq:cfg_i-breakStat}
  \bcfg{i}{\mathtt{break}}(t, \tb, \tc, m)
    \defeq
      \bigl\langle
        (\{m, \tb\},\bigl\{(m,\tb)\bigr\}, m), \emptyset, \emptyset, m+1
      \bigr\rangle.
\end{equation}

\paragraph{Continue statement:}
assuming $\tc \neq \bot$,
\begin{equation}
\label{eq:cfg_i-continueStat}
  \bcfg{i}{\mathtt{continue}}(t, \tb, \tc, m)
    \defeq
      \bigl\langle
        (\{m, \tc\},\bigl\{(m,\tc)\bigr\}, m), \emptyset, \emptyset, m+1
      \bigr\rangle.
\end{equation}

\paragraph{Goto statement:}
\begin{equation}
\label{eq:cfg_i-gotoStat}
  \bcfg{i}{\mathtt{goto} \; \id}(t, \tb, \tc, m)
    \defeq
      \bigl\langle
        (\{m, t\},\bigl\{(m, t)\bigr\}, m), \emptyset, \bigl\{(\id, m)\bigr\}, m+1
      \bigr\rangle.
\end{equation}

\paragraph{Labeled statement:}
\begin{equation}
\label{eq:cfg_i-labeledStat}
  \bcfg{i}{L \ccol S}(t, \tb, \tc, m)
    \defeq
  \bigl\langle (N, A, m), \Ms, \Mg, m_2 \bigr\rangle,
\end{equation}
where
$N \defeq N_L \union N_S$,
$A \defeq A_L \union A_S$,
$\Ms = \Ms_S \union \Ms_L$, and
\begin{align*}
  \cfg{i}{S}(t, \tb, \tc, m)
    &= \bigl\langle (N_S, A_S, s_S), \Ms_S, \Mg, m_1 \bigr\rangle, \\
  \cfg{i}{L}(s_S, m_1)
    &= \bigl\langle (N_L, A_L, s_L), \Ms_L, m_2 \bigr\rangle.
\end{align*}

\paragraph{Compound statement:}
\begin{equation}
\label{eq:cfg_i-compundStat}
  \bcfg{i}{\{S\}}(t, \tb, \tc, m)
    \defeq
   \cfg{i}{S}(t, \tb, \tc, m).
\end{equation}

\paragraph{Other statements:}
\begin{equation}
\label{eq:cfg_i-otherStat}
  \bcfg{i}{\{\mathrm{stm}\}}(t, \tb, \tc, m)
    \defeq
  \bigl\langle (\{t\},\emptyset, t), \emptyset, \emptyset, m \bigr\rangle.
\end{equation}

Finally, let $B \in \Stm$ be a full C function body:
the CFG constructed for $B$ with respect to optimization level $i$,
denoted by $\bodycfg{i}{B}$, is given by
\begin{equation}
  \bodycfg{i}{B} \defeq (N, A, s),
\end{equation}
where
\(
  \cfg{i}{B}(0, \bot, \bot, 1)
    = \bigl\langle (N, A_B, s), \Ms, \Mg, m_1 \bigr\rangle
\)
and $A$ is obtained from $A_S$ by adding the arcs corresponding
to \texttt{goto} statements, namely:
\begin{multline*}
  A \defeq A_B \\
    \union
      \bigl\{\,
        (n_1, n_2)
      \bigm|
        \exists s, l
          \st
            s = (\id_s, n_1) \in \Mg
          \land
            l = (\id_l, n_2) \in \Ms
          \land
            \id_s = \id_l
      \,\bigr\}.
\end{multline*}
\end{definition}

In the sequel, we will refer to the overloaded function
$\gcfg{i}{\cdot} \defeq \pi_1 \circ \cfg{i}{\cdot}$ where
$\pi_1 \defeq \lambda x_1, \dots x_n \st x_1$ is the first projection
of a variable number $n \geq 1$ of arguments.
In other words, for a program phrase $P \in \Exp \union \Stm$,
$\gcfg{i}{P}(\ldots)$ denotes the graph component
computed by $\cfg{i}{P}(\ldots)$.

\section{The NPATH Metric}
\label{sec:NPATH}

The definition of the NPATH metric for the C language, extracted from
\cite{Nejmeh88} and adapted to the grammar given in
Section~\ref{sec::c-abstract-syntax}, is given in
Tables~\ref{tab:NPATH-expressions} (for expressions)
and~\ref{tab:NPATH-statements} (for statements).
\begin{table}[ht]
\caption{Inductive definition of function $\npe$}
\label{tab:NPATH-expressions}
\centering
\begin{tabular}{l||c}
$E$ & $\npe(E)$ \\
\hhline{=::=}
$x$ & \multirow{2}*{$0$} \\
\hhline{-||~}
$c$ \\
\hhline{-||-}
$\cnot E_1$ & \multirow{6}*{$\npe(E_1)$} \\
\hhline{-||~}
$+E_1$ \\
\hhline{-||~}
$-E_1$ \\
\hhline{-||~}
$(E_1)$ \\
\hhline{-||~}
$(type)E_1$ \\
\hhline{-||~}
$\uop E_1$ \\
\hhline{-||-}
$E_1 \cand E_2$ & \multirow{2}*{$\npe(E_1) + \npe(E_2) + 1$} \\
\hhline{-||~}
$E_1 \cor E_2$ \\
\hhline{-||-}
$E_1, \; E_2$ & \multirow{2}*{$\npe(E_1) + \npe(E_2)$} \\
\hhline{-||~}
$E_1 \cqmkcol E_2$ \\
\hhline{-||-}
$E_1 \bop E_2$ & $\npe(E_1) + \npe(E_2)$ \\
\hhline{-||-}
$E_1 \cqmk E_2 \ccol E_3$ & $\npe(E_1) + \npe(E_2) + \npe(E_3) + 2$
\end{tabular}
\end{table}
\begin{table}[ht]
\caption{Inductive definition of function $\nps$}
\label{tab:NPATH-statements}
\centering
\begin{tabular}{l||c}
$S$ & $\nps(S)$ \\
\hhline{=::=}
$E;$ & $\npe(E)$ \\
\hhline{-||-}
$S_1 \; S_2$ & $\nps(S_1) \nps(S_2)$ \\
\hhline{-||-}
$\mathtt{return}$ & $1$ \\
\hhline{-||-}
$\mathtt{return} \; E$ & $max\bigl(1, \npe(E)\bigr)$ \\
\hhline{-||-}
$\mathtt{if} \; (E) \; S_1 \; \mathtt{else} \; S_2$ & $\npe(E) + \nps(S_1) + \nps(S_2)$ \\
\hhline{-||-}
$\mathtt{if} \; (E) \; S_1$ & $\npe(E) + \nps(S_1) + 1$ \\
\hhline{-||-}
$\mathtt{switch} \; (E) \; S_B$ &
$\npe(E) + \sum_{i = 1}^{k}\nps(S_i) + \nps(S_d)$ \\
\hhline{-||-}
$\mathtt{while} \; (E) \; S_1$ & \multirow{2}*{$\npe(E) + \nps(S_1) + 1$} \\
\hhline{-||~}
$\mathtt{do} \; S_1 \; \mathtt{while} \; (E)$  \\
\hhline{-||-}
$\mathtt{for} \; (E_1; E_2; E_3) \; S$ & $\npe(E_1) + \npe(E_2) + 
\npe(E_3) + \nps(S) + 1$ \\
\hhline{-||-}
$\mathtt{break}$ & \multirow{2}*{$1$} \\
\hhline{-||~}
$\mathtt{continue}$ \\
\hhline{-||-}
$\mathtt{goto} \; \mathrm{id}$ & $1$ \\
\hhline{-||-}
$\mathtt{L \ccol  S_1}$ & \multirow{2}*{$\nps(S_1)$} \\
\hhline{-||~}
$\{S_1\}$ \\
\hhline{-||-}
$\mathrm{stm}$ & $1$
\end{tabular}
\end{table}
Note that the syntax used in \cite{Nejmeh88} imposes strong
limits on the structure
of \texttt{switch} statements, hence the definition
given in Table~\ref{tab:NPATH-statements} is only valid
if $S_B$ has the form
\(
     \mathtt{case} \; n_1 \ccol S_1
  \; \mathtt{case} \; n_2 \ccol S_2
  \cdots
    \mathtt{case} \; n_k \ccol S_k
  \; \mathtt{default} \ccol S_d
\).\footnote{The ACPATH metric that we will define in
Section~\ref{sec:ACPATH} has no such limitation.}

The introduction of NPATH in~\cite{Nejmeh88} is motivated
by a convincing argument about the advantages
of counting the number of acyclic paths in order to estimate
the path complexity of a function.
One would assume that the definition of NPATH given in~\cite{Nejmeh88}
would provide a way of counting the number of acyclic paths but,
as we have already seen in Example~\vref{ex:NPATH-buggy-original},
this is not the case.

One of the main problems of NPATH is that,
as shown by Example~\ref{ex:NPATH-buggy-original},
in the conditional (resp., loop) statements, the number of acyclic paths
in the controlling expressions and in the construct's branches (resp., the body)
compound in a multiplicative, not in an additive way.
For the conditional, each acyclic path in one branch can be combined with
each acyclic path in the controlling expression that directs control flow
into that branch.

We now provide further examples where NPATH either underestimates
or overestimates the number of acyclic paths in the CFG of a C function.

\begin{example}
\label{ex:NPATH-buggy-while}
Consider the C function
\begin{verbatim}
  int f(int a, int b, int c, int d) {
    while(a || (b && c && d) ) {
      ... /* no branching statements here */
    }
  }
\end{verbatim}
We have $\NPATH = 3$, but the possible acyclic paths are $6$,
corresponding to the following combinations, where the ellipsis separates
the values of $\mathtt{a}$, $\mathtt{b}$, $\mathtt{c}$ and $\mathtt{d}$
before and after the first execution of the while body:
\begin{enumerate}
\item
\verb"a ... !a && b && c && !d"
\item
\verb"a ... !a && b && !c"
\item
\verb"a ... !a && !b"
\item
\verb"!a && b && c && !d"
\item
\verb"!a && b && !c"
\item
\verb"!a && !b"
\end{enumerate}
The problem shown by this example is that NPATH does not consider
the backward jump caused by while statement at the end of the execution
of the body.
In order to correctly compute the number of acyclic paths,
in addition to the paths that do not execute the while body,
we must consider the paths that first evaluate the guard to true,
whereby the body is executed, and then evaluate to false.
\end{example}

\begin{example}
\label{ex:NPATH-buggy-switch}
Consider the C function
\begin{verbatim}
  int f(int a, int b, int c) {
    switch (a) {
      case 1: b ? 0 : 1;
      default: return c ? 0 : 1;
    }
  }
\end{verbatim}
We have $\NPATH = 2 + 2 = 4$,
but the possible acyclic paths are 6,
corresponding to the following combinations:
\begin{enumerate}
\item
\verb"a = 1 && b && c"
\item
\verb"a = 1  && b && !c"
\item
\verb"a = 1 && !b && c"
\item
\verb"a = 1 && !b && !c"
\item
\verb"a <> 1 && c"
\item
\verb"a <> 1 && !c"
\end{enumerate}
Here the problem is that NPATH does not correctly capture
the syntax and semantics of the C \texttt{switch} statement.
In the function above, if $\mathtt{a}$ is equal to $1$,
control passed to the \verb|case 1| branch and,
after the execution of its range, since it does not contain
any \texttt{break} statement, the default range is executed.
In other words, NPATH does not account for so-called \emph{fall-through}
in C \texttt{switch} statements.
\end{example}

\begin{example}
\label{ex:NPATH-buggy-return-break-continue}
Consider the following C functions:
\begin{verbatim}
  void f(int a, int b, int c, int d, int e) {
    do {
      if (a)
        break/continue/return;
      if(b)
        ... /* no branching statements here */
      else
        ... /* no branching statements here */
    } while (c)
  }
\end{verbatim}
We have $\NPATH = 4 + 1 = 5$,
but there are only $3$ acyclic paths are instead only $3$ paths
for the \texttt{break} and the \texttt{return} cases,
corresponding to the following:
\begin{enumerate}
\item
\verb"a"
\item
\verb"!a, b, !c"
\item
\verb"!a, !b, !c"
\end{enumerate}
And there are only $2$ acyclic paths for the \texttt{continue} case,
corresponding to the following:
\begin{enumerate}
\item
\verb"!a, b, !c"
\item
\verb"!a, !b, !c"
\end{enumerate}
In these examples NPATH overstates the number of acyclic paths
because it does not distinguish \texttt{return}, \texttt{continue}
and \texttt{break} statements
from statements that do not affect control flow:
in all three cases the while body execution is abandoned
if $\mathtt{a}$ evaluates to true.
\end{example}

\begin{example}
\label{ex:NPATH-buggy-do-while}
Independently from the considerations illustrated by the previous
example, NPATH can overstate the number of acyclic paths for
$\mathtt{do-while}$ loops.  The simplest example is the idiomatic
\begin{verbatim}
  do { S } while (0)
\end{verbatim}
which is commonly used as a macro body so that macro calls can be
terminated with a semicolon without introducing a null statement,
while embedding $\mathtt{S}$ into a compound statement.
If $\mathtt{S}$ is a single basic block, we have $\NPATH \geq 2$
but there is only $1$ acyclic path.
\end{example}

Given that NPATH does not count acyclic paths one might think:
let us make the compiler build the CFG, and then let us count
how many acyclic path it contains from the entry node $s$ to any
exit node.
Unfortunately, this is unfeasible for general graphs.\footnote{We
are grateful to Charles Colbourn for indicating the following reduction
to us.}

\begin{mytheorem}
\label{thm:counting-s-T-acyclic-paths-is-sharp-P-complete}
Consider a directed graph $G = (N, A)$ with entry node $s$
and exit nodes in set $T$.
Counting $s-T$ acyclic paths in $G$
is $\mathsf{\sharp P}$-complete.
\end{mytheorem}
\begin{proof}
First, we can assume that $G$ only has one exit node $t$.
If it has more, then introduce a new node $t$ and place a directed arc
from each exit node to $t$;
this does not change the number of paths to an exit node.
Form a new directed graph $G' = (N', A')$ with
\begin{align*}
  N'
    &\defeq
      \bigl\{\, x_\mathrm{in}, x_\mathrm{out} \bigm| x \in N \,\bigr\}, \\
  A'
    &\defeq \bigl\{\, (x_\mathrm{in}, x_\mathrm{out}) \bigm| x \in N \,\bigr\}
    \union \bigl\{\, (x_\mathrm{out}, y_\mathrm{in}) \bigm| (x, y) \in A \,\bigr\}.
\end{align*}
A path from $s$ to $t$ in $G$ that repeats no nodes corresponds to a path
in $G'$ from $s_\mathrm{in}$ to $t_\mathrm{out}$ that repeats no arcs.
Indeed, every path in $G'$ alternates
arcs of the form $(x_\mathrm{out}, y_\mathrm{in})$
with arcs of the form $(y_\mathrm{in}, y_\mathrm{out})$.
Using node $y$ at most once in $G$ is the same as using arc
$(y_\mathrm{in}, y_\mathrm{out})$ at most once in $G'$.

This reduces the problem of counting paths with no repeated nodes
to that of counting ones with no repeated arcs,
and the former is $\mathsf{\sharp P}$-complete \cite{Valiant79}.
Then the latter is also $\mathsf{\sharp P}$-complete.
\qed
\end{proof}
Note that the $\mathsf{\sharp P}$-complete problems are at least as difficult
as the $\mathsf{NP}$-complete problems.
Indeed, the existence of a polynomial-time algorithm for solving
a $\mathsf{\sharp P}$-complete problem would imply
$\mathsf{P} = \mathsf{NP}$.

\section{The ACPATH Metric}
\label{sec:ACPATH}

In this section we present \emph{ACPATH}. This is a new metric for
C-like languages that, in contrast to the NPATH metric,
corresponds to the exact number of acyclic paths through any function
with no backjumps\footnote{A \emph{backjump}
  is a \texttt{goto} statement that jumps to a labeled statement
  that precedes it.}
and no jumps into some early-terminating loops.
Note that, as most coding standards disallow
such \texttt{goto} statements, in practice,
backjumps are rarely used in critical code.
For instance,
MISRA~C, the most influential C~coding standard~\cite{MISRA-C-2023},
has an \emph{advisory} rule forbidding
all \texttt{goto} statements and a \emph{required} rule forbidding
backjumps; so, while a forward \texttt{goto} can be used without
justification, a backjump requires a formal deviation.
Other \emph{required} MISRA guidelines prevent jumping or switching
into blocks from the outside of them, whether such blocks are loop
bodies or not, thereby preventing the occurrence of the second condition.

In Section~\ref{sec::ACPATH-expressions}, we present
an algorithms that counts the acyclic paths through expressions.
In Section~\ref{sec::ACPATH-statements}, we deal with
the more complex task of counting paths through statements.
All the algorithms presented in Sections~\ref{sec::ACPATH-expressions}
and~\ref{sec::ACPATH-statements}, are parametric
with respect to an optimization level.  As formally stated at the end of
the section, all the algorithms are correct for each optimization
level.

\subsection{Execution Paths Through Expressions}
\label{sec::ACPATH-expressions}

To deal with expressions, we introduce three functions:
$\tp{i}$, $\fp{i}$ and $\pp{i}$.
For each optimization level $i \in \{ 0, 1, 2\}$ and each
$E \in \Exp$ that is evaluated at optimization level $i$:
\begin{itemize}
\item
  $\tp{i}(E)$ counts the number of execution paths
  through $E$ that may evaluate to true;
\item
  $\fp{i}(E)$ counts the number of execution paths
  through $E$ that may evaluate to false;
\item
  $\pp{i}(E)$ counts the total number of possible execution paths through $E$.
\end{itemize}
It is important to stress that here we are dealing with path counting
with respect to a reference CFG and without any semantic inference apart
from those encoded in the optimization level.
Hence, when we say that a path through $E$ \emph{may} evaluate to true,
we mean that a path exists in the reference CFG for the considered
optimization level, and that the same optimization level does not
allow concluding that the path evaluates to false.

\begin{definition} \summary{($\tp{i}$, $\fp{i}$, $\pp{i}$.)}
  \label{def:t-f-p}
The functions
$\fund{\tp{i}}{\Exp}{\Nset}$,
$\fund{\fp{i}}{\Exp}{\Nset}$ and
$\fund{\pp{i}}{\Exp}{\Nset}$
are inductively defined, for each $i \in \{ 0, 1, 2\}$ and $E \in \Exp$,
as per \textup{Table~\ref{tab:algorithm-expression-cases}}.
\end{definition}

\begin{sidewaystable}
\caption{Inductive definition of
         $\fund{\tp{i}}{\Exp}{\Nset}$,
         $\fund{\fp{i}}{\Exp}{\Nset}$ and
         $\fund{\pp{i}}{\Exp}{\Nset}$}
\label{tab:algorithm-expression-cases}

\centering
\begin{tabular}{l|c|c|c}
$E$ & $\tp{i}(E)$ & $\fp{i}(E)$ & $\pp{i}(E)$ \\
\hline
$x$ & $1$ & $1$ & $1$ \\
\hline
$c$ if $\tv{i}(c)=\maybe$ & $1$ & $1$ & $1$ \\
$c$ if $\tv{i}(c)=\true$ & $1$ & $0$ & $1$ \\
$c$ if $\tv{i}(c)=\false$ & $0$ & $1$ & $1$ \\
\hline
$\cnot E_1$
  & $\fp{i}(E_1)$
  & $\tp{i}(E_1)$
  & $\pp{i}(E_1)$ \\
\hline
$+E_1$ & \multirow{4}*{$\tp{i}(E_1)$} & \multirow{4}*{$\fp{i}(E_1)$} & \multirow{4}*{$\pp{i}(E_1)$} \\
\cline{1-1}
$-E_1$ \\
\cline{1-1}
$(E_1)$ \\
\cline{1-1}
$(type)E_1$ \\
\hline
$\uop E_1$ & $\pp{i}(E_1)$ & $\pp{i}(E_1)$ & $\pp{i}(E_1)$ \\
\hline
$E_1 \cand E_2$
  & $\tp{i}(E_1) \tp{i}(E_2)$
  & $\fp{i}(E_1) + \tp{i}(E_1) \fp{i}(E_2)$
  & $\fp{i}(E_1) + \tp{i}(E_1) \pp{i}(E_2)$ \\
\hline
$E_1 \cor E_2$
  & $\tp{i}(E_1) + \fp{i}(E_1) \tp{i}(E_2)$
  & $\fp{i}(E_1) \fp{i}(E_2)$
  & $\tp{i}(E_1) + \fp{i}(E_1) \pp{i}(E_2)$ \\
\hline
$E_1 , E_2$
  & $\pp{i}(E_1) \tp{i}(E_2)$
  & $\pp{i}(E_1) \fp{i}(E_2)$
  & $\pp{i}(E_1) \pp{i}(E_2)$ \\
\hline
$E_1 \cqmkcol E_2$
  & $\tp{i}(E_1) + \fp{i}(E_1) \tp{i}(E_2)$
  & $\fp{i}(E_1) \fp{i}(E_2)$
  & $\tp{i}(E_1) + \fp{i}(E_1) \pp{i}(E_2)$ \\
\hline
$E_1 \bop E_2$
  & $\pp{i}(E_1) \pp{i}(E_2)$
  & $\pp{i}(E_1) \pp{i}(E_2)$
  & $\pp{i}(E_1) \pp{i}(E_2)$ \\
\hline
$E_1 \cqmk E_2 \ccol E_3$
  & $\tp{i}(E_1) \tp{i}(E_2)+\fp{i}(E_1) \tp{i}(E_3)$
  & $\tp{i}(E_1) \fp{i}(E_2)+\fp{i}(E_1) \fp{i}(E_3)$
  & $\tp{i}(E_1) \pp{i}(E_2)+\fp{i}(E_1) \pp{i}(E_3)$ \\
\end{tabular}
\end{sidewaystable}

In order to deal with acyclic paths induced by while loops,
we also need functions
$\ttp{i}$, $\tfp{i}$, $\ffp{i}$ and $\ppp{i}$;
for each $i \in \{ 0, 1, 2\}$ and each $E \in \Exp$:
\begin{itemize}
\item
  $\ttp{i}(E)$ (resp., $\ffp{i}(E)$) counts the number of ways
  in which the expression $E$ can be traversed twice at optimization
  level $i$, where both evaluation paths may lead to true (resp., false)
  and they do not share any arc;
\item
  $\tfp{i}(E)$ counts the number of ways in which the expression $E$
  can be traversed twice at optimization level $i$, where the two
  evaluations may lead to different Boolean values (i.e., one to true
  and the other to false), and the two traversals do not share any arc;
\item
  $\ppp{i}(E)$ counts the total number of possible ways in which the
  expression $E$ can be traversed twice at optimization level $i$,
  where the two traversals do not share any arc.
\end{itemize}

\begin{definition} \summary{($\ttp{i}$, $\tfp{i}$, $\ffp{i}$, $\ppp{i}$.)}
  \label{def:tt-ff-pp}
The functions
$\fund{\ttp{i}}{\Exp}{\Nset}$,
$\fund{\tfp{i}}{\Exp}{\Nset}$,
$\fund{\ffp{i}}{\Exp}{\Nset}$ and
$\fund{\ppp{i}}{\Exp}{\Nset}$
are inductively defined, for each $i \in \{ 0, 1, 2\}$ and $E \in \Exp$,
as per
\textup{Tables~\ref{tab:algorithm-expression-cases-tt}--\ref{tab:algorithm-expression-cases-pp}}.
\end{definition}

\begin{table}
\centering
\caption{Inductive definition of $\fund{\ttp{i}}{\Exp}{\Nset}$}
\label{tab:algorithm-expression-cases-tt}

\begin{tabular}{l|c}
$E$ & $\ttp{i}$ \\
\hline
$x$ & $0$ \\
\hline
$c$ if $\tv{i}(c)=\maybe$ & $0$ \\
$c$ if $\tv{i}(c)=\true$ & $1$ \\
$c$ if $\tv{i}(c)=\false$ & $0$ \\
\hline
$\cnot E_1$
  & $\ffp{i}(E_1)$ \\
\hline
$+E_1$ & \multirow{4}*{$\ttp{i}(E_1)$} \\
\cline{1-1}
$-E_1$ \\
\cline{1-1}
$(E_1)$ \\
\cline{1-1}
$(type)E_1$ \\
\hline
$\uop E_1$ & $0$ \\
\hline
$E_1 \cand E_2$
  & $\ttp{i}(E_1) \ttp{i}(E_2)$ \\
\hline
$E_1 \cor E_2$
  & $\ttp{i}(E_1) + 2 \tfp{i}(E_1) \tp{i}(E_2) + \ffp{i}(E_1) \ttp{i}(E_2)$ \\
\hline
$E_1 , E_2$
  & $\ppp{i}(E_1) \ttp{i}(E_2)$ \\
\hline
$E_1 \cqmkcol E_2$
  & $\ttp{i}(E_1) + 2 \tfp{i}(E_1) \tp{i}(E_2) + \ffp{i}(E_1) \ttp{i}(E_2)$ \\
\hline
$E_1 \bop E_2$
  & $0$ \\
\hline
$E_1 \cqmk E_2 \ccol E_3$
  & $\ttp{i}(E_1) \ttp{i}(E_2) + 2 \tfp{i}(E_1) \tp{i}(E_2) \tp{i}(E_3) + \ffp{i}(E_1) \ttp{i}(E_3)$ \\
\end{tabular}
\end{table}

\begin{table}
\centering
\caption{Inductive definition of $\fund{\tfp{i}}{\Exp}{\Nset}$}
\label{tab:algorithm-expression-cases-tf}

\begin{tabular}{l|c}
$E$ & $\tfp{i}(E)$ \\
\hline
$x$ & $1$ \\
\hline
$c$ if $\tv{i}(c)=\maybe$ & $1$ \\
$c$ if $\tv{i}(c)=\true$ & $0$ \\
$c$ if $\tv{i}(c)=\false$ & $0$ \\
\hline
$\cnot E_1$
  & $\tfp{i}(E_1)$ \\
\hline
$+E_1$ & \multirow{4}*{$\tfp{i}(E_1)$} \\
\cline{1-1}
$-E_1$ \\
\cline{1-1}
$(E_1)$ \\
\cline{1-1}
$(type)E_1$ \\
\hline
$\uop E_1$ & $\ppp{i}(E_1)$ \\
\hline
$E_1 \cand E_2$
  & $\tfp{i}(E_1) \tp{i}(E_2) + \ttp{i}(E_1) \tfp{i}(E_2)$ \\
\hline
$E_1 \cor E_2$
  & $\tfp{i}(E_1) \fp{i}(E_2) + \ffp{i}(E_1) \tfp{i}(E_2)$ \\
\hline
$E_1 , E_2$
  & $\ppp{i}(E_1) \tfp{i}(E_2)$ \\
\hline
$E_1 \cqmkcol E_2$
  & $\tfp{i}(E_1) \fp{i}(E_2) + \ffp{i}(E_1) \tfp{i}(E_2)$ \\
\hline
$E_1 \bop E_2$
  & $\ppp{i}(E_1) \ppp{i}(E_2)$ \\
\hline
$E_1 \cqmk E_2 \ccol E_3$
  & \(
      \begin{aligned}
      &\ttp{i}(E_1) \tfp{i}(E_2) + \ffp{i}(E_1) \tfp{i}(E_3) \\
    + &\tfp{i}(E_1) \bigl(\tp{i}(E_2) \fp{i}(E_3) + \fp{i}(E_2) \tp{i}(E_3)\bigr)
      \end{aligned}
    \)
\end{tabular}
\end{table}

\begin{table}
\centering
\caption{Inductive definition of function $\fund{\ffp{i}}{\Exp}{\Nset}$}
\label{tab:algorithm-expression-cases-ff}

\begin{tabular}{l|c}
$E$ & $\ffp{i}(E)$ \\
\hline
$x$ & $0$ \\
\hline
$c$ if $\tv{i}(c)=\maybe$ & $0 $ \\
$c$ if $\tv{i}(c)=\true$ & $0$ \\
$c$ if $\tv{i}(c)=\false$ & $1$\\
\hline
$\cnot E_1$ & $\ttp{i}(E_1)$ \\
\hline
$+E_1$ & \multirow{4}*{$\ffp{i}(E_1)$}  \\
\cline{1-1}
$-E_1$ \\
\cline{1-1}
$(E_1)$ \\
\cline{1-1}
$(type)E_1$ \\
\hline
$\uop E_1$ & $0$ \\
\hline
$E_1 \cand E_2$
  & $\ffp{i}(E_1) + 2 \tfp{i}(E_1) \fp{i}(E_2) + \ttp{i}(E_1) \ffp{i}(E_2)$ \\
\hline
$E_1 \cor E_2$
  & $\ffp{i}(E_1) \ffp{i}(E_2)$ \\
\hline
$E_1 , E_2$
  & $\ppp{i}(E_1) \ffp{i}(E_2)$ \\
\hline
$E_1 \cqmkcol E_2$
  & $\ffp{i}(E_1) \ffp{i}(E_2)$ \\
\hline
$E_1 \bop E_2$
  & $0$ \\
\hline
$E_1 \cqmk E_2 \ccol E_3$
  & $\ttp{i}(E_1) \ffp{i}(E_2) + 2 \tfp{i}(E_1) \fp{i}(E_2) \fp{i}(E_3) + \ffp{i}(E_1) \ffp{i}(E_3)$ \\
\hline
\end{tabular}
\end{table}

\begin{table}
\centering
\caption{Inductive definition of function $\fund{\ppp{i}}{\Exp}{\Nset}$}
\label{tab:algorithm-expression-cases-pp}

\begin{tabular}{l|c}
$E$ & $\ppp{i}(E)$ \\
\hline
$x$ & $0$ \\
\hline
$c$ if $\tv{i}(c)=\maybe$ & $0$ \\
$c$ if $\tv{i}(c)=\true$ & $1$ \\
$c$ if $\tv{i}(c)=\false$ & $1$ \\
\hline
$\cnot E_1$
  & $\ppp{i}(E_1)$ \\
\hline
$+E_1$ & \multirow{4}*{$\ppp{i}(E_1)$} \\
\cline{1-1}
$-E_1$ \\
\cline{1-1}
$(E_1)$ \\
\cline{1-1}
$(type)E_1$ \\
\hline
$\uop E_1$ & $0$ \\
\hline
$E_1 \cand E_2$
  & $\ffp{i}(E_1) + 2 \tfp{i}(E_1) \pp{i}(E_2) + \ttp{i}(E_1) \ppp{i}(E_2)$ \\
\hline
$E_1 \cor E_2$
  & $\ttp{i}(E_1) + 2 \tfp{i}(E_1) \pp{i}(E_2) + \ffp{i}(E_1) \ppp{i}(E_2)$ \\
\hline
$E_1 , E_2$
  & $\ppp{i}(E_1) \ppp{i}(E_2)$ \\
\hline
$E_1 \cqmkcol E_2$
  & $\ttp{i}(E_1) + 2 \tfp{i}(E_1) \pp{i}(E_2) + \ffp{i}(E_1) \ppp{i}(E_2)$ \\
\hline
$E_1 \bop E_2$
  & $0$ \\
\hline
$E_1 \cqmk E_2 \ccol E_3$
  & $\ttp{i}(E_1) \ppp{i}(E_2) + 2 \tfp{i}(E_1) \pp{i}(E_2) \pp{i}(E_3) + \ffp{i}(E_1) \ppp{i}(E_3)$ \\
\end{tabular}
\end{table}

\subsection{Execution Paths Through Statements}
\label{sec::ACPATH-statements}

We first consider a labeled statement $L \ccol S$: the paths that
reach $S$ are those that ``fall to $L$ from above'' plus those
that ``switch to $L$'' if $L$ is a case or default label, or ``go to
$L$'' if $L$ is an identifier label in a \texttt{goto} statement. The (possibly
decorated) symbols $\Ft$, $\St$ and $\Gt$ will be used as mnemonics
for the number of paths that \emph{fall through}, \emph{switch to} and
\emph{go to} $L$, respectively.  In the sequel, if $\Id$ is the set of
identifier labels in a function, $\Gt$ will be a partial function
$\pard{\Gt}{\Id}{\Nset}$, mapping any label identifier $\id \in \Id$
to the cumulative number of paths that reach all the $\mathtt{goto} \;
\id$ statements in the function that occur before the labeled
statement, $\id \ccol S$.

\begin{definition} \summary{($\fund{\apc{i}{}}{\Lab \times \Nset \times \Nset \times (\Id \rightarrowtail \Nset)}{\Nset}$.)}
Let $L$ be the label for a labeled statement,
and $\Ft$, $\St$ and $\Gt$ be as defined above.
Then $\apc{i}{L}(\Ft, \St, \Gt)$ is defined as follows:

\proofsec{Case label}
\begin{equation}
\label{eq:apc_i-caseLabel}
  \apc{i}{\mathtt{case} \; n}(\Ft, \St, \Gt)
    \defeq
      \Ft + \St.
\end{equation}

\proofsec{Default label}
\begin{equation}
\label{eq:apc_i-defaultLabel}
  \apc{i}{\mathtt{default}}(\Ft, \St, \Gt)
    \defeq
      \Ft + \St.
\end{equation}

\proofsec{Identifier label}
\begin{equation}
\label{eq:apc_i-idLabel}
  \apc{i}{\id}(\Ft, \St, \Gt)
    \defeq
      \Ft + \Gt(\id).
\end{equation}
\end{definition}

We assume that a function terminates with an empty statement $\epsilon$ so that
each non-empty statement $S$ in a function has a successor $S_1$.
In order to count the total number of paths that reach $S_1$,
we introduce the (overloaded) function $\apc{i}{}$.
If $\Ft$, $\St$, $\Gt$ are described as above, then
$\apc{i}{S}(\Ft, \St, \Gt)$ computes:
\begin{itemize}
\item $\Ft_{out}$: the number of acyclic paths that that
  ``fall through $S_1$ from above'';
\item $\Bp$: the cumulative sum of the acyclic paths that lead to
  \texttt{break} nodes that terminate the execution of $S$,
  i.e., \texttt{break} nodes
  that are not in \texttt{switch} or loop statements in $S$;
\item $\Cp$: the cumulative sum of the acyclic paths that lead to
  \texttt{continue} nodes that terminate the execution of $S$, i.e.,
  \texttt{continue} nodes that are not in loop statements in $S$;
\item $\Rp$: the cumulative sum of the acyclic paths that lead to
  \texttt{return} nodes in $S$;
\item
  $\Gt_{out}$: a partial function
  $\pard{\Gt}{\Id}{\Nset}$, mapping any label identifier $\id$
  to the cumulative number of paths that reach
  all the $\mathtt{goto} \; \id$ statements that occur
  before $S_1$.
\end{itemize}

Let $\pard{\Gt,\Gt_1,\Gt_2}{\Id}{\Nset}$.
The function $\pard{\Gt[n/\id]}{\Id}{\Nset}$
is given, for each $x \in \Id$, by
\begin{align*}
  \Gt[n/\id](x)
    &\defeq
      \begin{cases}
        n,      &\text{if $x = \id$,} \\
        \Gt(x), &\text{otherwise;}
      \end{cases}
\intertext{%
in addition, $\pard{(\Gt_1 + \Gt_2)}{\Id}{\Nset}$ is given,
for each $x \in \Id$, by
}
  (\Gt_1 + \Gt_2)(x) &\defeq \Gt_1(x) + \Gt_2(x).
\end{align*}

\begin{definition} \summary{($\fund{\apc{i}{}}{\Stm \times \Nset \times \wp(\Id \times \Nset) }{\Nset^4 \times \wp(\Id \times \Nset)}$.)}
\label{def:apc_i-Stat}
We define the function
\[
  \fund{\apc{i}{}}{\Stm \times \Nset \times \wp(\Id \times \Nset) }{\Nset^4 \times \wp(\Id \times \Nset)}
\]
as follows:

\proofsec{Expression statement}
\begin{equation}
\label{eq:apc_i-exprStat}
\apc{i}{E;}(\Ft, \St, \Gt)
\defeq
(\pp{i}(E)\Ft, 0, 0, 0, \Gt).
\end{equation}

\proofsec{Sequential composition}
\begin{equation}
\label{eq:apc_i-seqStat}
\apc{i}{S_1 \; S_2}(\Ft, \St, \Gt)
\defeq
(\Ft_2, \Bp, \Cp, \Rp, \Gt_2),
\end{equation}
where
$\Bp = \Bp_1 + \Bp_2$,
$\Cp = \Cp_1 + \Cp_2$,
$\Rp = \Rp_1 + \Rp_2$,
\begin{align*}
  \apc{i}{S_1}(\Ft, \St, \Gt) &= (\Ft_1, \Bp_1, \Cp_1, \Rp_1, \Gt_1), \\
  \apc{i}{S_2}(\Ft_1, \St, \Gt_1) &= (\Ft_2, \Bp_2, \Cp_2, \Rp_2, \Gt_2).
\end{align*}

\proofsec{Return statement}
\begin{equation}
\label{eq:apc_i-returnStat}
\apc{i}{\mathtt{return}}(\Ft, \St, \Gt)
\defeq
(0, 0, 0, \Ft, \Gt).
\end{equation}

\proofsec{Return with expression statement}
\begin{equation}
\label{eq:apc_i-returnExprStat}
\apc{i}{\mathtt{return} \; E}(\Ft, \St, \Gt)
\defeq
(0, 0, 0, \pp{i}(E) \Ft, \Gt).
\end{equation}

\proofsec{Conditional statement}
\begin{equation}
\label{eq:apc_i-condStat}
\bapc{i}{\mathtt{if} \; (E) \; S_1 \; \mathtt{else} \; S_2}(\Ft, \St, \Gt)
\defeq
(\Ft_{out}, \Bp, \Cp, \Rp, \Gt_2),
\end{equation}
where we have
$\Ft_{out} = \Ft_1 + \Ft_2$,
$\Bp = \Bp_1 + \Bp_2$,
$\Cp = \Cp_1 + \Cp_2$,
$\Rp = \Rp_1 + \Rp_2$,
\begin{align*}
  \apc{i}{S_1}(\tp{i}(E) \Ft, \St, \Gt) &= (\Ft_1, \Bp_1, \Cp_1, \Rp_1, \Gt_1), \\
  \apc{i}{S_2}(\fp{i}(E) \Ft, \St, \Gt_1) &= (\Ft_2, \Bp_2, \Cp_2, \Rp_2, \Gt_2).
\end{align*}

\proofsec{One-armed conditional statement}
\begin{equation}
\label{eq:apc_i-one-armedCondStat}
\bapc{i}{\mathtt{if} \; (E) \; S_1}(\Ft, \St, \Gt)
\defeq
(\Ft_{out}, \Bp_1, \Cp_1, \Rp_1, \Gt_1),
\end{equation}
where we have
\begin{align*}
\Ft_{out} &= \Ft_1 + \fp{i}(E), \\
\apc{i}{S_1}(\tp{i}(E) \Ft, \St, \Gt) &= (\Ft_1, \Bp_1, \Cp_1, \Rp_1, \Gt_1).
\end{align*}

\proofsec{Switch statement}
\begin{equation}
\label{eq:apc_i-switchStat}
\bapc{i}{\mathtt{switch} \; (E) \; S}(\Ft, \St, \Gt)
   \defeq
     \begin{cases}
        (\Ft_1, 0, \Cp_S, \Rp_S, \Gt_S),
          &\text{if $d(S)$,} \\
        (\Ft_2, 0, \Cp_S, \Rp_S, \Gt_S),
          &\text{otherwise,}
      \end{cases}
\end{equation}
where we have
\begin{align*}
\Ft_1 &= \Ft_S + \Bp_S, \\
\Ft_2 &= \Ft_S + \Bp_S + \pp{i}(E) \Ft, \\
\apc{i}{S}(0, \pp{i}(E) \Ft, \Gt) &= (\Ft_S, \Bp_S, \Cp_S, \Rp_S, \Gt_S),
\end{align*}
and $d(S)$ is true if and only if $S$ contains a default label out
of all inner \texttt{switch}.

\proofsec{While statement}
\begin{equation}
\label{eq:apc_i-whileStat}
\bapc{i}{\mathtt{while} \; (E) \; S}(\Ft, \St, \Gt)
\defeq (\Ftout, 0, 0, \Rp_S, \Gt_S)
\end{equation}
where we have
\begin{align*}
\Ftout &= \fp{i}(E)\Ft + \Bp_S\tp{i}(E) + \bigl(\Ft_S + \Cp_S\bigr) \tfp{i}(E), \\
\apc{i}{S}(\Ft, \St, \Gt) &= (\Ft_S, \Bp_S, \Cp_S, \Rp_S, \Gt_S).
\end{align*}

\proofsec{Do-while statement}
\begin{equation}
\label{eq:apc_i-doWhileStat}
\bapc{i}{\mathtt{do} \; S \; \mathtt{while} \; (E)}(\Ft, \St, \Gt)
\defeq
(\Ft_{out}, 0, 0, \Rp_S, \Gt_S),
\end{equation}
where we have
\begin{align*}
\Ft_{out} &= \fp{i}(E) \Ft_S + \Bp_S, \\
\apc{i}{S}(\Ft, \St, \Gt) &= (\Ft_S, \Bp_S, \Cp_S, \Rp_S, \Gt_S).
\end{align*}

\proofsec{For statement}
\begin{equation}
\label{eq:apc_i-forStat}
\bapc{i}{\mathtt{for} \; (E_1; E_2; E_3) \; S}(\Ft, \St, \Gt)
\defeq
\bapc{i}{E_1 \; \mathtt{while} \; (E_2) \; \{S \; E_3;\}}(\Ft, \St, \Gt).
\end{equation}

\proofsec{Break statement}
\begin{equation}
\label{eq:apc_i-breakStat}
\apc{i}{\mathtt{break}}(\Ft, \St, \Gt)
\defeq
(0, \Ft, 0, 0, \Gt).
\end{equation}

\proofsec{Continue statement}
\begin{equation}
\label{eq:apc_i-continueStat}
\apc{i}{\mathtt{continue}}(\Ft, \St, \Gt)
\defeq
(0, 0, \Ft, 0, \Gt).
\end{equation}

\proofsec{Goto statement}
\begin{equation}
\label{eq:apc_i-gotoStat}
  \apc{i}{\mathtt{goto} \; \id}(\Ft, \St, \Gt)
    \defeq
      \bigl(0, 0, 0, 0, \Gt\bigl[\bigl(\Gt(\id) + \Ft\bigr)/\id\bigr]\bigr).
\end{equation}

\proofsec{Labeled statement}
\begin{equation}
\label{eq:apc_i-labelStat}
\apc{i}{L \ccol  S}(\Ft, \St, \Gt)
\defeq
(\Ft_S, \Bp, \Cp, \Rp, \Gt_{out}),
\end{equation}
where we have
\begin{align*}
\apc{i}{L}(\Ft, \St, \Gt) &= \Ft_L, \\
\apc{i}{S}(\Ft_L, \St, \Gt) &= (\Ft_S, \Bp, \Cp, \Rp, \Gt_{out}).
\end{align*}

\proofsec{Compound statement}
\begin{equation}
\label{eq:apc_i-compStat}
\bapc{i}{\{S\}}(\Ft, \St, \Gt)
  \defeq
\apc{i}{S}(\Ft, \St, \Gt).
\end{equation}

\proofsec{Other statements}
\begin{equation}
\label{eq:apc_i-otherStat}
\apc{i}{\mathrm{stm}}(\Ft, \St, \Gt)
  \defeq
(\Ft, 0, 0, 0, \Gt).
\end{equation}
\end{definition}

It is clear from its definition that, for each function body $B \in \Stm$,
$\bodyapc{i}{B}$ can be computed with a single traversal of $B$.
It is also easy to prove the following:
\begin{proposition}
\label{prop:any-graph-in-C}
Let $G = (N, A, s)$ be a directed graph with entry node $s$ and exit nodes
in set $T$.  Then there exists a C function body $B \in \Stm$
of size $O(A)$
such that, for each $i \in \{ 0, 1, 2 \}$,
$\apn\bigl(\bodycfg{i}{B}\bigr) = \apn(G)$.
\end{proposition}
\begin{proof}
Assume $s = n_1$ and $T = \{ t_1, \dots, t_h \}$.
Let $\bigl\{ \{ n_1, \dots, n_k \}, \{ t_1, \dots t_h \} \bigr\}$
be a partition of $N$, and let
$I = \{ \id_{n_1}, \dots, \id_{n_k}, \id_{t_1}, \dots \id_{t_h} \}$
be a set of C identifiers in one-to-one correspondence with $N$.
Then define
\[
  B
    \defeq
      \id_{n_1}: S_{n_1}; \dots; \id_{n_k}: S_{n_k};
      \id_{t_1}: S_{t_1}; \dots; \id_{t_h}: S_{t_h}
\]
where:
\begin{itemize}
\item
for each $i = 1$, \dots,~$k$,
\begin{align*}
  S_{n_i} &\defeq \mathtt{switch} \; (x_{n_i}) \; S'_{n_i} \\
\intertext{%
where $x_{n_i}$ is a variable and,
if $\bigl\{ (n_i, n_{i,1}), \dots, (n_i, n_{i,p_i}) \bigr\}$
is the subset of $A$ containing all arcs leaving $n_i$,
then
}
  S'_{n_i}
    &\defeq
     \mathtt{case} \; 1 \ccol \mathtt{goto} \; \id_{n_{i,1}};
     \;
     \cdots
     \; ; \;
     \mathtt{case} \; p_i \ccol
     \mathtt{default} \ccol \mathtt{goto} \; \id_{n_{i,p_i}};
\end{align*}
\item
for each $j = 1$, \dots,~$h$, $S_{t_j} = \mathtt{return}$.
\end{itemize}
Checking that $\bodycfg{i}{B}$ has the same number of acyclic paths
as $G$ is straightforward.
\qed
\end{proof}

Theorem~\ref{thm:counting-s-T-acyclic-paths-is-sharp-P-complete},
together with Proposition~\ref{prop:any-graph-in-C} and the fact that
$\bodyapc{i}{B}$ computes the number of acyclic paths in
$\bodycfg{i}{B}$ in time linearly proportional to the size of $B$,
implies that $\bodyapc{i}{B}$ cannot be exact for all
function bodies $B \in \Stm$.
However, it is correct for a very large class of function bodies,
informally characterized as follows (a formal definition is given
in Appendix~\ref{app:proofs}).

\begin{definition} \summary{(Controlled function body (informal).)}
\label{defn:controlled-function-body-informal}
Let $B \in \Stm$ be a full C function body.
We call $B$ a \emph{controlled function body} if it satisfies
the following properties:
\begin{itemize}
\item
  it does not contain any backjump;
\item
  if a loop in $B$ can terminate its execution by means of
  $\mathtt{goto}$, $\mathtt{break}$ or $\mathtt{return}$ statements, then no $\mathtt{goto}$ or
  $\mathtt{switch}$ statement in $B$ will
  jump or \texttt{switch} into the loop from outside.
\end{itemize}
\end{definition}

The following examples show minimal function bodies that are not
controlled.
\begin{example}
Here we jump into the loop from outside
via $\mathtt{goto}$ and the loop exits via $\mathtt{break}$:
\begin{verbatim}
  void f(int x) {
    goto l1;
    while (x) {
      break;
     l1: ... /* no branching statements here */
    }
}
\end{verbatim}
The value for ACPATH is $1$, but there are $2$ acyclic paths:
the path that jumps into the loop and then evaluates the guard to
false, and the path that jumps into the loop, then evaluates
the guard to true and exits via the $\mathtt{break}$ statement.

Here we jump into the loop from outside
via $\mathtt{switch}$ and the loop exits via $\mathtt{return}$:
\begin{verbatim}
  void g(int x, int y) {
    switch (x) {
        do {
          return;
      case 0: ... /* no branching statements here */
        } while (y)
    }
  }
\end{verbatim}
The value of ACPATH is $2$, but there are $3$ acyclic paths:
the path that switches to the end of the function,
the path that switches in the loop and evaluates the loop guard false,
and the path that switches into the loop, then evaluates the loop
guard to true, then exits via $\mathtt{return}$ statement.
\end{example}

\begin{definition} \summary{($\fund{\bodyapc{i}{}}{\Stm}{\Nset}$.)}
\label{def:apc_i-Fullbody}
Let $B \in \Stm$ be a full C function body, and $l(B) \in \wp(\Id)$
the set of labels in $B$; then the number of acyclic paths through $B$ with
respect to optimization level $i$, denoted by $\bodyapc{i}{B}$, is
given by
\begin{equation}
\label{eq:apc_i-Fullbody-eq}
  \bodyapc{i}{B} \defeq \Ft_{out} + \Rp,
\end{equation}
where
$\Gt = \bigl\{\, (\id, 0) \bigm| \id \in l(B) \,\bigr\}$ and
$\apc{i}{B}(1, 0, \Gt) = (\Ft_{out}, \Bp, \Cp, \Rp, \Gt_{out})$.
\end{definition}

\begin{restatable}{mytheorem}{apniscorrect}
\label{thm:apn-correctness}
Let $B \in \Stm$ be a controlled function body.
Then
\[
  \bodyapc{i}{B} = \apn\bigl(\bodycfg{i}{B}\bigr).
\]
\end{restatable}

The proof of Theorem~\ref{thm:apn-correctness} is in Appendix A.

\section{Implementation and Experimental Evaluation}
\label{sec:implementation-experimental-evaluation}

This section reports on a study of the relationship between the metric
introduced in this paper, ACPATH, and NPATH \cite{Nejmeh88}.
As we have already seen, they are not equivalent from the theoretical point
of view: we will now show that they are not equivalent also from the
point of view of their practical application.

\subsection{Implementation}
\label{sec::implementation}

The ACPATH and NPATH metrics (and many others) have been implemented
in ECLAIR, a powerful platform for the automatic analysis, verification,
testing and transformation of C, \Cplusplus{} and Java source code,
as well as Java bytecode.%
\footnote{\url{http://bugseng.com/products/eclair}}
In particular, for assessing the complexity of software, ECLAIR
provides comprehensive code metrics that can be accumulated
over a single function, translation unit, program or even the whole
project.
The ACPATH algorithm has been implemented in ECLAIR as a metric over
a complete function body, the optimization level $i \in \{ 0, 1, 2 \}$
being a parameter of the analysis.
Although the ACPATH metric is only fully specified and verified here for C
code, the implementation is designed to handle both C and \Cplusplus{} user
code, i.e., fully instantiated preprocessed code.
The same holds true for the implementation of NPATH.
Apart from the generalization to \Cplusplus{}, the implementation of both
metrics closely follows the definitions given in this paper.
The implementation language is a very high-level logical description
language that is automatically translated to executable code.
The implementation of NPATH is around 300 lines long whereas 550 lines
are sufficient to implement ACPATH.

\subsection{Sampled Functions}
\label{sec::sampled-functions}

The experimental evaluation was conducted on $61$ C projects,
for a total of $35284$ functions: the majority of such projects
involve safety- or mission-critical functionality, mainly from the
automotive sector, with projects from other domains (aerospace, railway and
medical appliances), some operating system kernels and some, non-critical
open-source projects.

The condition about the absence of backjumps is largely satisfied:
only $19$ C functions ($0.05\%$) have one or more backjumps.
We are currently instrumenting ECLAIR in order to count the number
of functions that do not satisfy the other conditions
of Definition~\ref{defn:controlled-function-body-informal}.
However, we expect the number of functions/methods to which
Theorem~\ref{thm:apn-correctness} does not apply to be very small,
if not negligible.

\subsection{Statistical Analysis}
\label{sec::statistical-analysis}

Many studies on software metrics place reliance upon
Pearson linear correlation coefficient $r$,
which, however, assumes the variables are approximately normally
distributed \cite{Shepperd88}.
This is definitely not our case, as sample skewness of our
data is around~$200$.
A common methodology for skew reduction is transformation:
in our case sample skew drops to around $50$ after taking the logarithm
of the metric values.  A further drop to around $1$ is obtained by
adding~$1$ and taking the logarithm again, that is, by transforming
the data with the function $\log\bigl(1 + \log(x)\bigr)$:
Figure~\ref{fig:scatter-plot} shows the scatter plot after the transformations.
\begin{figure}[ht]
\centering
\includegraphics[width=10cm]{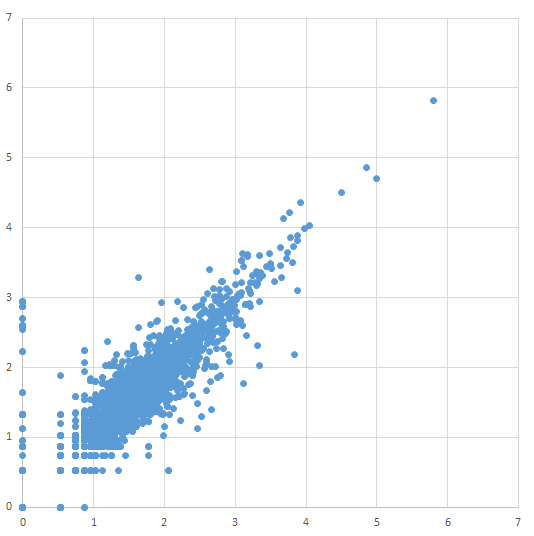}
\caption{Scatter plot:
         $\log\bigl(1 + \log(\mathrm{ACPATH})\bigr)$, $x$ axis, vs
         $\log\bigl(1 + \log(\mathrm{NPATH})\bigr)$, $y$ axis}
\label{fig:scatter-plot}
\end{figure}
The figure shows that there is good, but not absolute correlation,
with Pearson correlation coefficient $r \simeq 0.98$.
The average error is $\mu \simeq 0$, i.e., NPATH overestimations
approximately balance the underestimations,
but the standard deviation of the error, $\sigma \simeq 0.12$, confirms,
taking into account the $\log\bigl(\log(\cdot)\bigr)$ transformation,
that estimation errors can be quite large.
To give an idea what this means, let us report this statistics
back from the $\log\bigl(1 + \log(x)\bigr)$ to the $x$ scale
and let us focus on values of ACPATH in the range $[1, 26]$.
In Figure~\ref{fig:line-graph} the red, blue and green lines
represent $\mu$, $\mu - \sigma$ and $\mu + \sigma$, respectively.
In Figure~\ref{fig:bar-graph} we report, for each value of ACPATH
in the range $[1, 26]$, the computed values for NPATH, where we discarded
the tails, $16\%$ on each side, so that each bar represents $68\%$
of the samples (this matches the range $[\mu - \sigma, \mu + \sigma]$
of Figure~\ref{fig:line-graph}).
It can be seen that, while the error committed by NPATH is low
for functions with number of acyclic paths below $10$, the error
can become rather large even for slightly more complex functions,
and their distribution is quite faithfully described by the values
of $\mu$ and $\sigma$ obtained as described above.

\begin{figure}[p]
\centering
\includegraphics[width=10cm]{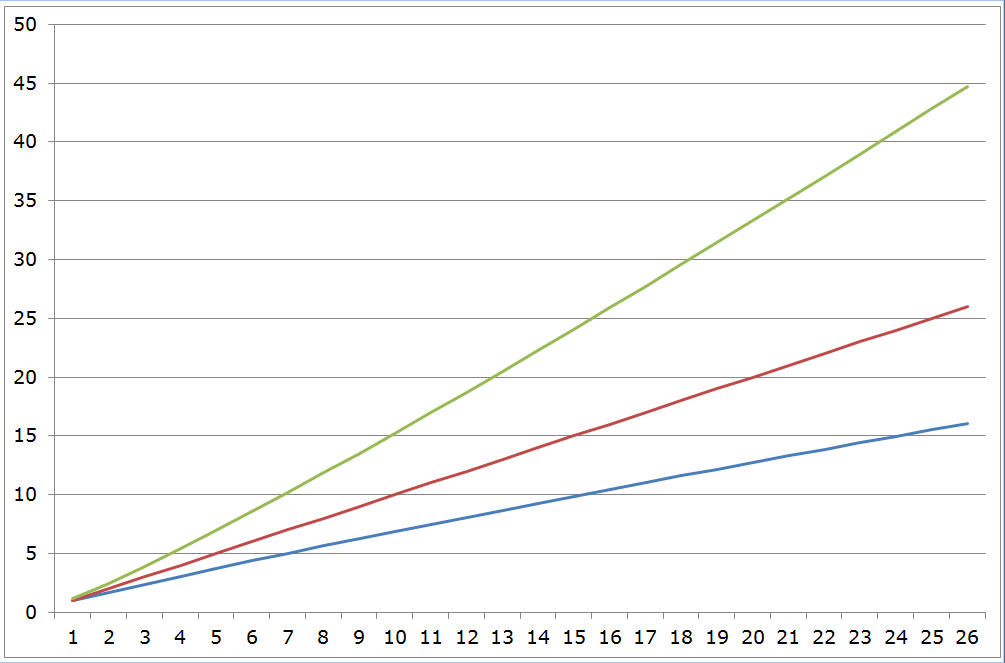}
\caption{ACPATH value, $x$ axis, vs NPATH predicted value distribution,
         $y$ axis}
\label{fig:line-graph}
\end{figure}
\begin{figure}[p]
\centering
\includegraphics[width=10cm]{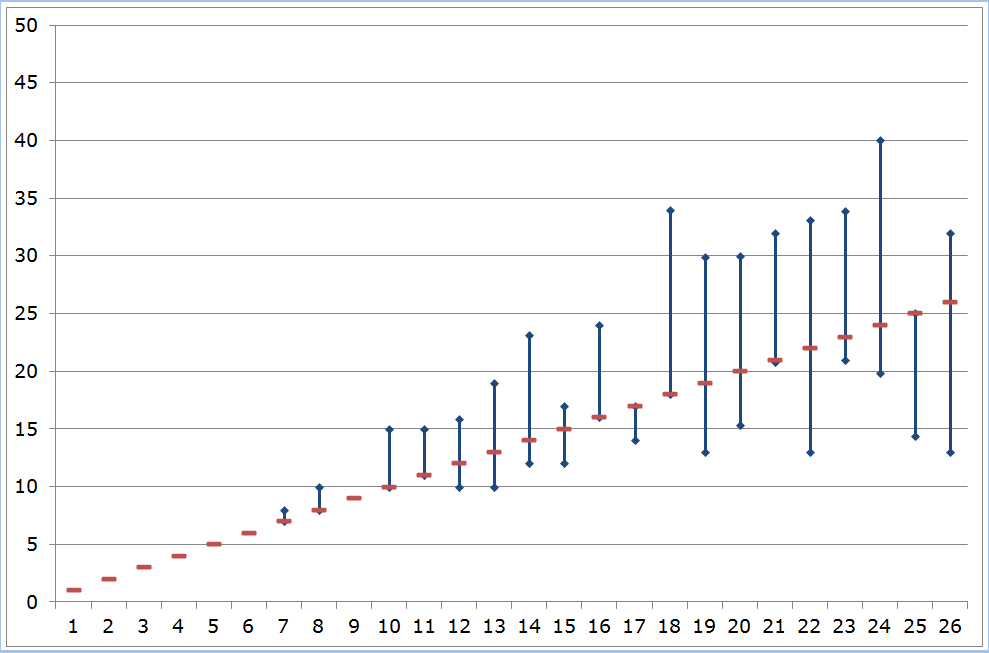}
\caption{ACPATH value, $x$ axis, vs NPATH, $y$ axis:
         each bar contains the central $68\%$ of the samples}
\label{fig:bar-graph}
\end{figure}

\subsection{Analysis with Respect to Commonly Used Thresholds}
\label{sec::analysis-with-thresholds}

The statistical analysis of the previous section does not take into
account the fact that organizations that enforce the adoption of sound
software engineering principles place strong limitations on the
maximum number of acyclic paths that a function may have.
Other elements for the comparison can thus be given with
reference to the recommended thresholds, under the assumption that
such thresholds were actually meant to apply to the number of acyclic
paths and not the measure NPATH, that is neither a lower bound nor an
upper bound for them even for trivial programs.

In \cite{Nejmeh88} a threshold value of $200$ is recommended
based on studies conducted at AT\&T Bell Laboratories.
The HIS\footnote{HIS,
Herstellerinitiative Software, is an interest group set up by
Audi, BMW, Daimler, Porsche, and Volkswagen in order to join forces
on methods of software design and quality assurance for microprocessor
based control units.}
Software Test Working Group,
in the document defining the \emph{HIS metrics} \cite{HIS-SCM-2008},
which are widespread in the automotive and other industrial
sectors, is stricter and recommends not to exceed $80$.
In our experiment, we found that the ACPATH threshold of $80$
is respected by $95\%$ of the C functions,
whereas the threshold of $200$ is respected by $97\%$ of the C functions.
The number of functions that would be miscategorized with respect to
violating or complying with the thresholds is rather small:
\begin{align*}
\mathrm{ACPATH} &>    80  & \mathrm{NPATH} &\leq 80  & \#\mathrm{f} &= 195 \\
                &\leq 80  &                &>    80  &              &= 281 \\
                &>    200 &                &\leq 200 &              &= 152 \\
                &\leq 200 &                &>    200 &              &= 182 \\
\end{align*}
which also shows that, for ``borderline'' functions,
NPATH errs on the ``non-compliant'' side more often than it errs
on the ``compliant'' side.
The differences can be very large though, here are the worst cases
we have found in our experimentation:
\begin{align*}
  \mathrm{ACPATH} &= 1,          &\mathrm{NPATH} &= 67108864, \\
  \mathrm{ACPATH} &= 597781,     &\mathrm{NPATH} &= 21, \\
  \mathrm{ACPATH} &= 2329612972, &\mathrm{NPATH} &= 130.
\end{align*}
The first such worst-case result concerns a function
that contains $26$ macro invocations that expand to
as many $\mathtt{do} \; S \; \mathtt{while} \, (0)$ statements,
where $S$ is a basic block; so while ACPATH gives $1$,
i.e., the correct number of acyclic paths, NPATH results
in $2^{26} = 67108864$.

\afterpage{\clearpage}

\section{Conclusion}
\label{sec:conclusion}

Path complexity is a program complexity measure that can be
used to assess the testability of units, that is, functions
or methods, depending on the programming language.
It is recognized that this is a much better estimator of testability than
cyclomatic complexity \cite{Harrison84,McCabe76}, which was
found to perform no better that LSLOCs (Logical Source Lines of Code)
in this regard \cite{Nejmeh88,Shepperd88}.

As the number of paths in a unit can be unbounded, in \cite{Nejmeh88}
it was proposed to use the number of acyclic paths as a proxy.
In the same paper, the NPATH metric for C programs was proposed
with a claim that it would count the number of acyclic
paths for programs without \texttt{goto} statements \cite[page~192]{Nejmeh88}:
\begin{quote}
NPATH is a measure that is more closely related to the
number of acyclic execution paths through a function.
In particular, the NPATH measure differs from the actual
number of acyclic execution paths by the number
of acyclic execution paths resulting from \texttt{goto} statements.
\end{quote}
In reality, the syntax for the $\mathtt{switch}$ statement
is significantly more restricted than what is actually
permitted by the definition of the C programming language
so, while not acknowledged explicitly in \cite[page~192]{Nejmeh88},
this is another restriction of the NPATH metric.

The starting point for this paper was the experimental discovery that,
in fact, the number of acyclic execution paths through a function
can differ from NPATH enormously.
Indeed, the difference can be seen even for very simple examples
and these small variations can compound in a multiplicative way to
very large differences.

We then asked whether or not the number of acyclic paths could be
computed exactly by working directly on the control flow graph of the function
rather than on the syntax of its body.
Unfortunately, we discovered that the problem of counting the
acyclic paths in a graph is $\mathsf{\sharp P}$-complete,
which leaves little hope for an efficient algorithm.
If one wants simply to know whether the number of acyclic paths
is below or above a certain (small) threshold, then a possibility
is to enumerate all the acyclic paths exhaustively or until the
threshold is attained.  For a CFG $G = (N, A, s)$ this can be done in time
$O\bigl(|A| \cdot \Delta(G)\bigr)$ per enumerated acyclic path,
where $\Delta(G)$ is the maximum degree
of $G$,\footnote{The \emph{degree} of a node is
the number of arcs incident to the node.
In the CFG of a C function, the maximum degree is essentially given
by the maximum number of cases in $\mathtt{switch}$ statements.}
using Johnson's algorithm \cite{Johnson75} on the
line graph of $G$.\footnote{The \emph{line graph}
of a directed graph $G = (N, A)$ is the directed graph $L(G)$ whose node
set is $A$ and arcs the set of adjacent arcs in $A$.}
This is probably still not efficient enough and, moreover, such an
approach does not allow to differentiate between functions that
violate the threshold.

In this paper, we defined a new metric, called ACPATH, for C-like languages:
even though we defined it formally only for the C programming language,
we have extended it to \Cplusplus{}.
The metric can be computed very efficiently with a single traversal
of the abstract syntax tree of the function.
Moreover, we have proved that ACPATH does correspond to the number
of acyclic paths under two conditions: (1) absence of backjumps;
(2) absence of jumps into a loop whenever there are jumps out of the same loop,
terminating the loop in ways unrelated to the evaluation of its guard.
We proved that, if condition (1) is removed, then the existence
of an efficient algorithm to compute the number of acyclic paths
would imply $\mathsf{P} = \mathsf{NP}$.
The study of what can be done if condition (2) is removed or weakened
is one direction for future work.

We thus proposed ACPATH as a natural successor and replacement for NPATH:
the former solves most of the problems of the latter while retaining
the existence of a very efficient algorithm to compute it,
with only a minor increase in the complexity of the definition.
Moreover, ACPATH has been proved correct on counting the number
of acyclic paths for most programs that are written in practice,
and the cases where ACPATH is not exact are easy to detect.
Exactness of ACPATH on many functions allowed us to conduct
experiments, using the ECLAIR software verification platform,
on the adequacy of NPATH, both from a statistical
and from a more pragmatic point of view.

A different approach to the estimation of execution path complexity has been
proposed in \cite{BangAB15}, called \emph{asymptotic path complexity}.
Using techniques from algebraic graph theory and linear algebra, the
authors show how to obtain a closed form upper bound to
$\mathrm{path}(n)$ ---the number of (possibly cyclic) paths in the
unit of length at most $n$--- of the form $n$, $n^2$, $n^3$ and so on,
or $b^n$ for some exponential base $b$.
The bound is only valid to the limit, i.e., as $n$ goes to infinity.

For future work, we plan to formalize the extension to \Cplusplus{}
and to other languages, Java source code and bytecode to start with.
A \Cplusplus{} extension is already implemented in ECLAIR but its correctness
has not yet been formally proved.
The only language feature that requires special care in the generalization
to \Cplusplus{} is structured exception handling.
We believe that an extension for Java can then easily be derived
from the one to \Cplusplus: additionally, as Java has no
$\mathtt{goto}$ statements and the syntax of $\mathtt{switch}$ is more
restrictive than in \Cplusplus{} so that it is not possible to
jump into loops, we conjecture ACPATH for Java will always be exact.
The generalization to Java bytecode is more problematic, as it depends
on the ability to reconstruct loops, which is nontrivial especially
if obfuscation techniques have been used.

\begin{acknowledgements}
We are grateful to Charles Colbourn (Arizona State University, U.S.A.)
for indicating us the reduction in the proof of
Theorem~\ref{thm:counting-s-T-acyclic-paths-is-sharp-P-complete},
and to Silvano Martello (University of Bologna, Italy) and
Yves Crama  (University of Li\`{e}ge, Belgium) for helping us finding
what we were looking for.
Thanks to Alessio Conte and Andrea Marino (University of Pisa, Italy),
for the discussions we had on algorithms for the enumeration
of acyclic paths in generic directed graphs.
Our gratitude goes also to Francesco Morandin (University of Parma, Italy),
for his invaluable help with the statistical analysis of the experimental data.
\end{acknowledgements}

\clearpage
\appendix

\ESM{Appendices}

\section{Technical Proofs}
\label{app:proofs}

We now show that, for each optimization level~$i$
and each expression $E$, the following hold:
\begin{itemize}
\item
  $\tp{i}(E)$ (resp., $\fp{i}(E)$) is the number of acyclic paths in $\gcfg{i}{E}(t, f, m)$ leading to $t$ (resp., $f$);
\item
  $\pp{i}(E)$ is the total number of acyclic paths in $\gcfg{i}{E}(t, f, m)$.
\end{itemize}

  Note that in all the
  proofs, a ``path to~$s$'' means a path leading to~$s$., ``pair of paths'' always means a pair of acyclic paths that do
  not share an arc
\begin{lemma}
\label{lm:expr-count}
Let $i \in \{ 0, 1, 2 \}$ be an optimization level,
$E \in \Exp$ be an expression, and let
$m, t, f \in \Nset$ be such that $m > \max\{t, f\}$.
Then:
\begin{align}
  \tp{i}(E) &= \napn\bigl(\gcfg{i}{E}(t, f, m), t\bigr), \\
  \fp{i}(E) &= \napn\bigl(\gcfg{i}{E}(t, f, m), f\bigr), \\
  \pp{i}(E) &= \apn\bigl(\gcfg{i}{E}(t, t, m)\bigr).
\end{align}
\end{lemma}
\begin{proof}
The proof is by structural induction on $E$.

\proofsec{Variables}
By~Definition~\ref{def:t-f-p},
$\tp{i}(E) = \fp{i}(E) = \pp{i}(E) = 1$.
By~\eqref{eq:cfg_i-variables}, letting $s \in \{t,f\}$,  we have
\begin{align*}
  \napn\bigl(\gcfg{i}{E}(t, f, m), s\bigr)
  &= \napn\Bigl(\big(\{ m,t,f \}, \{ (m,t),(m,f) \}, m\bigr), s\Bigr) \\
  &= \napn\bigl(m, \{ (m,t), (m,f) \}, s\bigr) \\
  &= \napn\bigl(s, \{ (m, t), (m,f) \} \setminus \{(m,s)\}, s\bigr) = 1, \\
  \apn\bigl(\gcfg{i}{E}(t, t, m)\bigr)
  &= \apn\Bigl(\big(\{ m,t \}, \{ (m,t) \}, m\bigr)\Bigr) \\
  &= \napn\bigl(m, \{ (m,t), (m,f) \}, t\bigr) = 1.
\end{align*}

\proofsec{Constants}
There are three cases:
\begin{description}
\item[$\tv{i}(E) = \true$:]
  by~Definition~\ref{def:t-f-p},
  $\tp{i}(E) = 1$, $\fp{i}(E) = 0$ and $\pp{i}(E) = 1$;
  hence, by~\eqref{eq:cfg_i-constants},
  \begin{align*}
    \dapn\bigl(\gcfg{i}{E}(t, f, m), t\bigr)
    &= \dapn\Bigl(\big(\{ t \}, \{ \}, t\bigr), t\Bigr)
    = \dapn\bigl(t, \{ \}, t\bigr) = 1, \\
    \dapn\bigl(\gcfg{i}{E}(t, f, m), f\bigr)
    &= \dapn\Bigl(\big(\{ t \}, \{ \}, t\bigr), f\Bigr)
    = \dapn\bigl(t, \{ \}, f\bigr) = 0, \\
    \apn\bigl(\gcfg{i}{E}(t, t, m)\bigr)
    &= \apn\Bigl(\big(\{ t \}, \{ \}, t\bigr)\Bigr) = 1.
  \end{align*}
\item[$\tv{i}(E) = \false$:]
  the proof is similar to the previous case;
\item[$\tv{i}(E) = \maybe$:]
  the proof is similar to the case of variables.
\end{description}

\proofsec{Logical negation}
By~Definition~\ref{def:t-f-p},
$\tp{i}(!E) = \fp{i}(E)$ and $\pp{i}(!E) = \pp{i}(E)$ and,
by~\eqref{eq:cfg_i-negExpr},
$\gcfg{i}{E}(t, f, m) = \gcfg{i}{!E_1}(f, t, m)$. Hence, by the
inductive hypothesis,
\begin{alignat*}{3}
  \fp{i}(\cnot E_1) &=\tp{i}(E_1)
           &&= \napn\bigl(\gcfg{i}{E_1}(f, t, m), f\bigr)
           &&= \napn\bigl(\gcfg{i}{!E_1}(t, f, m), f\bigr), \\
  \tp{i}(\cnot E_1) &= \fp{i}(E_1)
           &&= \napn\bigl(\gcfg{i}{E_1}(f, t, m), t\bigr)
           &&= \napn\bigl(\gcfg{i}{!E}(t, f, m), t\bigr), \\
  \pp{i}(\cnot E_1) &= \pp{i}(E_1)
           &&= \apn\bigl(\gcfg{i}{E_1}(t, t, m)\bigr)
           &&= \apn\bigl(\gcfg{i}{!E_1}(t, t, m)\bigr).
\end{alignat*}

\proofsec{Unary plus, unary minus, parenthesis and cast operators}
By~Definition~\ref{def:t-f-p},
\[
\tp{i}(\uop E_1) = \tp{i}(E_1),\quad \fp{i}(\uop E_1) = \fp{i}(E_1),\quad
\pp{i}(\uop E_1) = \pp{i}(E_1)
\]
and,
by~\eqref{eq:cfg_i-specUnaryExpr},
\[
\gcfg{i}{\uop E_1}(t, f, m) = \gcfg{i}{E_1}(t, f, m).
\]
Hence the proof follows by induction.

\proofsec{Other unary operators}
By~Definition~\ref{def:t-f-p},
\begin{equation*}
  \tp{i}(\uop E_1) = \fp{i}(\uop E_1) = \pp{i}(\uop E_1) = \pp{i}(E_1).
\end{equation*}
Moreover, for each of the paths to~$m$ in
$\gcfg{i}{E_1}(m, m, m+1)$, there is one path to~$t$ passing through the arc
$(m, t)$ and one path to~$f$ passing through $(m, f)$ in
$\gcfg{i}{\uop E_1}(t, f, m)$, and one path to an exit node in
$\gcfg{i}{E_1}(t, t, m)$.
Hence by~\eqref{eq:cfg_i-unaryExpr},
\begin{align*}
 \apn\bigl(\gcfg{i}{E_1}(t, t, m)\bigr)
                  &= \napn\bigl(\gcfg{i}{\uop E_1}(t, f, m), t\bigr), \\
 \apn\bigl(\gcfg{i}{E_1}(t, t, m)\bigr)
                  &= \napn\bigl(\gcfg{i}{\uop E_1}(t, f, m), f\bigr), \\
  \apn\bigl(\gcfg{i}{E_1}(t, t, m)\bigr)
                  &= \apn\bigl(\gcfg{i}{\uop E_1}(t, t, m)\bigr).
\end{align*}
Hence the proof follows by induction.

\proofsec{Logical conjunction}
By~Definition~\ref{def:t-f-p},
\begin{equation}
  \label{eq:lm:expr-count:lc-t-f-p}
  \begin{aligned}
  \tp{i}(E_1 \cand E_2) &= \tp{i}(E_1) \tp{i}(E_2), \\
  \fp{i}(E_1 \cand E_2) &= \fp{i}(E_1) + \tp{i}(E_1) \fp{i}(E_2), \\
  \pp{i}(E_1 \cand E_2) &= \fp{i}(E_1) + \tp{i}(E_1) \pp{i}(E_2).
  \end{aligned}
\end{equation}
There are three cases:
\begin{description}
\item [$\tv{i}(E_1) = \false$:]
  then by~\eqref{eq:cfg_i-constants}, $\tp{i}(E_1) = 0$ and $\fp{i}(E_1) = 1$;
  also by~\eqref{eq:cfg_i-conjExpr},
  $\gcfg{i}{E_1 \cand E_2}(t, f, m) = \bigl(\{f\}, \emptyset, f\bigr)$,
  so there are no paths to~$t$ and one path to~$f$.
  Moreover, in control-flow graph
  $\gcfg{i}{E_1 \cand E_2}(t, t, m) = \bigl(\{t\}, \emptyset, t\bigr)$,
  there is one path to an exit node.
  Hence, by~\eqref{eq:lm:expr-count:lc-t-f-p},
  \begin{align*}
    \napn\bigl(\gcfg{i}{E_1 \cand E_2}(t, f, m), t\bigr) &= 0
    = \tp{i}(E_1 \cand E_2), \\
    \napn\bigl(\gcfg{i}{E_1 \cand E_2}(t, f, m), f\bigr) &= 1
    = \fp{i}(E_1 \cand E_2), \\
    \apn\bigl(\gcfg{i}{E_1 \cand E_2}(t, t, m)\bigr) &= 1
    = \pp{i}(E_1 \cand E_2).
\end{align*}

\item [$\tv{i}(E_1) = \true$:]
  then,
  by~\eqref{eq:cfg_i-constants}, $\tp{i}(E_1) = 1$ and $\fp{i}(E_1) = 0$;
  also, by~\eqref{eq:cfg_i-conjExpr},
  $\gcfg{i}{E_1 \cand E_2}(t, f, m) = \gcfg{i}{E_2}(t, f, m)$.
  Hence, by \eqref{eq:lm:expr-count:lc-t-f-p} and the inductive hypothesis:
\begin{align*}
  \tp{i}(E_1 \cand E_2)
    &= \tp{i}(E_2) \\
    &= \napn\bigl(\gcfg{i}{E_2}(t, f, m), t\bigr) \\
    &= \napn\bigl(\gcfg{i}{E_1 \cand E_2}(t, f, m), t\bigr), \\
  \fp{i}(E_1 \cand E_2)
    &= \fp{i}(E_2) \\
    &= \napn\bigl(\gcfg{i}{E_2}(t, f, m), f\bigr) \\
    &= \napn\bigl(\gcfg{i}{E_1 \cand E_2}(t, f, m), f\bigr), \\
  \pp{i}(E_1 \cand E_2)
    &= \pp{i}(E_2) \\
    &= \apn\bigl(\gcfg{i}{E_2}(t, t, m)\bigr) \\
    &= \apn\bigl(\gcfg{i}{E_1 \cand E_2}(t, t, m)\bigr).
\end{align*}

\item [$\tv{i}(E_1) = \maybe$:]
  using the notation in~\eqref{eq:cfg_i-conjExpr},
  since the exit node for true evaluation for
  graph $\gcfg{i}{E_1}(s_2, f, m_1)$
  is the entry node for graph $\gcfg{i}{E_2}(t, f, m)$,
  for each path to~$s_2$ in $\gcfg{i}{E_1}(s_2, f, m_1)$,
  there are $\napn\bigl(\gcfg{i}{E_2}(t, f, m), t\bigr)$ paths to
  $t$ and $\napn\bigl(\gcfg{i}{E_2}(t, f, m), f\bigr)$ paths
  to~$f$ in $\gcfg{i}{E_1 \cand E_2}(t, f, m)$,
  as well as $\apn\bigl(\gcfg{i}{E_2}(t, t, m)\bigr)$
  paths to an exit node in
  $\gcfg{i}{E_1 \cand E_2}(t, t, m)$. In addition, control-flow graph
  $\gcfg{i}{E_1}(s_2, f, m_1)$ has $\napn\bigl(\gcfg{i}{E_1}(s_2, f, m), f\bigr)$
  paths to~$f$.
  Hence, by~\eqref{eq:lm:expr-count:lc-t-f-p},~\eqref{eq:cfg_i-conjExpr} and the inductive hypothesis:
\begin{align*}
  \tp{i}(E_1 \cand E_2) &= \tp{i}(E_1) \tp{i}(E_2) \\
    &= \napn\bigl(\gcfg{i}{E_1}(s_2, f, m), s_2\bigr) \napn\bigl(\gcfg{i}{E_2}(t, f, m), t\bigr) \\
    &= \napn\bigl(\gcfg{i}{E_1 \cand E_2}(t, f, m), t\bigr), \\
  \fp{i}(E_1 \cand E_2) &= \fp{i}(E_1) + \tp{i}(E_1) \fp{i}(E_2) \\
    &= \napn\bigl(\gcfg{i}{E_1}(s_2, f, m), f\bigr) \\
    &\quad \mathord{} + \napn\bigl(\gcfg{i}{E_1}(s_2, f, m), s_2\bigr) \napn\bigl(\gcfg{i}{E_2}(t, f, m), f\bigr) \\
    &= \napn\bigl(\gcfg{i}{E_1 \cand E_2}(t, f, m), f\bigr), \\
  \pp{i}(E_1 \cand E_2) &= \fp{i}(E_1) + \tp{i}(E_1) \pp{i}(E_2) \\
    &= \napn\bigl(\gcfg{i}{E_1}(s_2, f, m), f\bigr) \\
    &\quad \mathord{} + \napn\bigl(\gcfg{i}{E_1}(s_2, f, m), s_2\bigr) \apn\bigl(\gcfg{i}{E_2}(t, t, m)\bigr) \\
    &= \apn\bigl(\gcfg{i}{E_1 \cand E_2}(t, t, m)\bigr).
\end{align*}
\end{description}

\proofsec{Logical disjunction}
Dual to the case of logical conjunction.

\proofsec{Comma operator}
Using the notation in~\eqref{eq:cfg_i-commaExpr},
as $\gcfg{i}{E_1}(s_2, s_2, m_1)$
has both the exit nodes equal
to the entry node of $\gcfg{i}{E_2}(t, f, m)$,
for each path leading to~$s_2$
in $\gcfg{i}{E_1}(s_2, s_2, m_1)$
there are $\napn\bigl(\gcfg{i}{E_2}(t, f, m), t\bigr)$ paths to~$t$
and $\napn\bigl(\gcfg{i}{E_2}(t, f, m), f\bigr)$ paths
to~$f$ in $\gcfg{i}{E_1, E_2}(t, f, m)$,
as well as $\apn\bigl(\gcfg{i}{E_2}(t, t, m)\bigr)$ paths
to an exit node in $\gcfg{i}{E_1, E_2}(t, t, m)$.
Hence, by Definition~\ref{def:t-f-p},~\eqref{eq:cfg_i-commaExpr}
and the inductive hypothesis:
\begin{align*}
  \tp{i}(E_1, E_2)
    &= \pp{i}(E_1) \tp{i}(E_2) \\
    &= \apn\bigl(\gcfg{i}{E_1}(s_2, s_2, m_1)\bigr) \napn\bigl(\gcfg{i}{E_2}(t, f, m), t\bigr) \\
    &= \napn\bigl(\gcfg{i}{E_1, E_2}(t, f, m), t\bigr), \\
  \fp{i}(E_1, E_2) &= \pp{i}(E_1) \fp{i}(E_2) \\
    &= \apn\bigl(\gcfg{i}{E_1}(s_2, s_2, m_1)\bigr) \napn\bigl(\gcfg{i}{E_2}(t, f, m), f\bigr) \\
    &= \napn\bigl(\gcfg{i}{E_1, E_2}(t, f, m), f\bigr), \\
  \pp{i}(E_1, E_2) &= \pp{i}(E_1) \pp{i}(E_2) \\
    &= \apn\bigl(\gcfg{i}{E_1}(s_2, s_2, m_1)\bigr) \apn\bigl(\gcfg{i}{E_2}(t, t, m)\bigr) \\
    &= \apn\bigl(\gcfg{i}{E_1, E_2}(t, t, m)\bigr).
\end{align*}

\proofsec{Binary conditional operator}
By Definition~\ref{def:t-f-p},
the functions $\tp{i}$, $\fp{i}$ and $\pp{i}$ are defined the same
for logical disjunction and the binary conditional operator.
Moreover, the definitions for logical disjunction in~\eqref{eq:cfg_i-disjExpr}
and the binary conditional operator in~\eqref{eq:cfg_i-binCondExpr}
are the same.
Hence the proof of this case is identical to the case of logical disjunction.

\proofsec{Other binary operators}
Using the notation in~\eqref{eq:cfg_i-binExpr},
as graph $\gcfg{i}{E_1}(s_2, s_2, m_1)$
has both exit nodes equal to the entry node
of graph $\gcfg{i}{E_2}(m, m, m+1)$,
for each path
to~$s_2$ in $\gcfg{i}{E_1}(s_2, s_2, m_1)$
there are $\apn\bigl(\gcfg{i}{E_2}(m, m, m+1)\bigr)$ paths to~$m$
in control-flow graph $\gcfg{i}{E_1 \bop E_2})(t, f, m)$.
Moreover, since $(m, t), (m, f) \in A$, each path to~$m$ can be extended to
paths to~$t$ and~$f$ in $\gcfg{i}{E_1 \bop E_2})(t, f, m)$
and an exit node in $\gcfg{i}{E_1 \bop E_2})(t, t, m)$.
Hence, by Definition~\ref{def:t-f-p},~\eqref{eq:cfg_i-binExpr}
and the inductive hypothesis:
\begin{align*}
  \tp{i}(E_1 \bop E_2) &= \pp{i}(E_1) \pp{i}(E_2) \\
    &= \apn\bigl(\gcfg{i}{E_1}(s_2, s_2, m_1)\bigr) \apn\bigl(\gcfg{i}{E_2}(m, m, m+1)\bigr) \\
    &= \napn\bigl(\gcfg{i}{E_1 \bop E_2}(t, f, m), t\bigr), \\
  \fp{i}(E_1 \bop E_2) &= \pp{i}(E_1) \pp{i}(E_2) \\
    &= \apn\bigl(\gcfg{i}{E_1}(s_2, s_2, m_1)\bigr) \apn\bigl(\gcfg{i}{E_2}(m, m, m+1)\bigr) \\
    &= \napn\bigl(\gcfg{i}{E_1 \bop E_2}(t, f, m), f\bigr), \\
  \pp{i}(E_1 \bop E_2) &= \pp{i}(E_1) \pp{i}(E_2) \\
    &= \apn\bigl(\gcfg{i}{E_1}(s_2, s_2, m_1)\bigr) \apn\bigl(\gcfg{i}{E_2}(m, m, m+1)\bigr) \\
    &= \apn\bigl(\gcfg{i}{E_1 \bop E_2}(t, t, m)\bigr).
\end{align*}

\proofsec{Conditional operator}
By Definition~\ref{def:t-f-p},
\begin{equation}
  \label{eq:lm:expr-count:co-t-f-p}
  \begin{aligned}
    \tp{i}(E_1 \cqmk E_2 \ccol E_3)
    &= \tp{i}(E_1) \tp{i}(E_2)+\fp{i}(E_1) \tp{i}(E_3) \\
    \fp{i}(E_1 \cqmk E_2 \ccol E_3)
    &= \fp{i}(E_1) \fp{i}(E_2)+\fp{i}(E_1) \fp{i}(E_3) \\
    \pp{i}(E_1 \cqmk E_2 \ccol E_3)
    &= \tp{i}(E_1) \pp{i}(E_2)+\fp{i}(E_1) \pp{i}(E_3).
  \end{aligned}
\end{equation}
There are three cases:
\begin{description}
\item[$\tv{i}(E_1) = \true$:]
  by~\eqref{eq:cfg_i-condExpr},
  \[
  \gcfg{i}{E_1 \cqmk E_2 \ccol E_3}(t, f, m) = \gcfg{i}{E_2}(t, f, m),
  \]
  and by~\eqref{eq:cfg_i-constants}, $\tp{i}(E_1) = 1$, and $\fp{i}(E_1) = 0$.
  Hence by~\eqref{eq:lm:expr-count:co-t-f-p}
  and the inductive hypothesis,
  \begin{align*}
  \tp{i}(E_1 \cqmk E_2 \ccol E_3)
    &= \tp{i}(E_2) \\
    &= \napn\bigl(\gcfg{i}{E_2}(t, f, m), t\bigr) \\
    &= \napn\bigl(\gcfg{i}{E_1 \cqmk E_2 \ccol E_3}(t, f, m), t\bigr), \\
  \fp{i}(E_1 \cqmk E_2 \ccol E_3)
    &= \fp{i}(E_2) \\
    &= \napn\bigl(\gcfg{i}{E_2}(t, f, m), f\bigr) \\
    &= \napn\bigl(\gcfg{i}{E_1 \cor E_2}(t, f, m), f\bigr), \\
  \pp{i}(E_1 \cqmk E_2 \ccol E_3)
    &= \pp{i}(E_2) \\
    &=\apn\bigl(\gcfg{i}{E_2}(t, t, m)\bigr) \\
  &= \apn\bigl(\gcfg{i}{E_1 \cqmk E_2 \ccol E_3}(t, t, m)\bigr).
\end{align*}
\item [$\tv{i}(E_1) = 0$:] symmetric to the previous case.
\item [$\tv{i}(E_1) = \maybe$:]
  Using the notation in~\eqref{eq:cfg_i-condExpr},
  as $\gcfg{i}{E_1}(s_2, s_3, m_2)$ has the exit node for true evaluation
  equal to the entry node of $\gcfg{i}{E_2}(t, f, m)$,
  for each path leading to~$s_2$
  in $\gcfg{i}{E_1}(s_2, s_3, m_2)$
  there are
  $\napn\bigl(\gcfg{i}{E_2}(t, f, m), t\bigr)$ paths to~$t$ and
  $\napn\bigl(\gcfg{i}{E_2}(t, f, m), f\bigr)$ paths to $f$
  in $\gcfg{i}{E_1 \cqmk E_2 \ccol E_3}(t, f, m)$,
  and
  $\apn\bigl(\gcfg{i}{E_2}(t, t, m)\bigr)$ paths to an exit node
  in $\gcfg{i}{E_1 \cqmk E_2 \ccol E_3}(t, t, m)$.

  Similarly, as the exit node
  for false evaluation in graph $\gcfg{i}{E_1}(s_2, s_3, m_2)$
  is the entry node of $\gcfg{i}{E_3}(t, f, m_1)$,
  for each path to~$s_3$ in $\gcfg{i}{E_1}(s_2, s_3, m_2)$,
  there are $\napn\bigl(\gcfg{i}{E_3}(t, f, m), t\bigr)$ paths to~$t$
  and $\napn\bigl(\gcfg{i}{E_3}(t, f, m), t\bigr)$ paths to~$f$
  in control-flow graph $\gcfg{i}{E_1 \cqmk E_2 \ccol E_3}(t, f, m)$;
  furthermore, there are $\apn\bigl(\gcfg{i}{E_3}(t, t, m)\bigr)$ paths
  to an exit node in $\gcfg{i}{E_1 \cqmk E_2 \ccol E_3}(t, t, m)$.
  Hence, by~\eqref{eq:lm:expr-count:co-t-f-p},~\eqref{eq:cfg_i-condExpr}
  and the inductive hypothesis:
\begin{align*}
  \tp{i}(E_1 \cqmk E_2 \ccol E_3)
    &= \tp{i}(E_1) \tp{i}(E_2) + \fp{i}(E_1) \tp{i}(E_3) \\
    &= \napn\bigl(\gcfg{i}{E_1}(s_2, s_3, m_2), s_2\bigr) \napn\bigl(\gcfg{i}{E_2}(t, f, m), t\bigr) \\
    & \quad \mathord{}
      + \napn\bigl(\gcfg{i}{E_1}(s_2, s_3, m_2), s_3\bigr) \napn\bigl(\gcfg{i}{E_3}(t, f, m), t\bigr) \\
    &= \napn\bigl(\gcfg{i}{E_1 \cqmk E_2 \ccol E_3}(t, f, m), t\bigr), \\
 \fp{i}(E_1 \cqmk E_2 \ccol E_3)
    &= \tp{i}(E_1) \fp{i}(E_2) + \fp{i}(E_1) \fp{i}(E_3) \\
    &= \napn\bigl(\gcfg{i}{E_1}(s_2, s_3, m_2), s_2\bigr) \napn\bigl(\gcfg{i}{E_2}(t, f, m), f\bigr) \\
    & \quad \mathord{}
      + \napn\bigl(\gcfg{i}{E_1}(s_2, s_3, m_2), s_3\bigr) \napn\bigl(\gcfg{i}{E_3}(t, f, m), f\bigr) \\
    &= \napn\bigl(\gcfg{i}{E_1 \cqmk E_2 \ccol E_3}(t, f, m), f\bigr), \\
  \pp{i}(E_1 \cqmk E_2 \ccol E_3)
    &= \tp{i}(E_1) \pp{i}(E_2) + \fp{i}(E_1) \pp{i}(E_3) \\
    &= \napn\bigl(\gcfg{i}{E_1}(s_2, s_3, m_2), s_2\bigr) \apn\bigl(\gcfg{i}{E_2}(t, t, m)\bigr) \\
    & \quad \mathord{}
      + \napn\bigl(\gcfg{i}{E_1}(s_2, s_3, m_2), s_3\bigr) \apn\bigl(\gcfg{i}{E_3}(t, t, m)\bigr) \\
      &= \apn\bigl(\gcfg{i}{E_1 \cqmk E_2 \ccol E_3}(t, t, m)\bigr). \qed
\end{align*}
\end{description}
\end{proof}

In order to capture the number of acyclic paths in CFGs containing
\emph{while} loops we will need another path-counting function.
\emph{While} loops are special in that the subgraphs corresponding to their
controlling expressions can be traversed twice: once on a path leading
to execution of the body and one leading to exit from the loop;
the two traversals should of course not share any arc.

\begin{definition} \summary{($\fund{\dapn}{\CFG\times\Nset^2}{\Nset}$.)}
\label{def:dapn}
Let $G = (N, A, s)$ be a CFG and $s_1, s_2 \in N$.
Then the number of pairs of acyclic paths in $G$ to~$s_2$ passing
by $s_1$, denoted by $\dapn(G, s_1, s_2)$, is given by:
\begin{align}
\dapn(G, s_1, s_2)
    &\defeq
      \dapn(s, s, A, s_1, s_2), \\
  \dapn(s, n, A, s_1, s_2)
    &\defeq
      \begin{cases}
        \napn(s, A, s_2),
          &\text{if $n = s_1$;} \\
        \sum\limits_{(n, m) \in A}
        \dapn\bigl(s, m, A \setminus \{ (n, m) \}, s_1, s_2\bigr),
          &\text{otherwise.}
      \end{cases}
\end{align}
\end{definition}

We now show, using Definitions~\ref{def:tt-ff-pp} and~\ref{def:dapn},
that for each optimization level~$i$ and each expression $E$, the
following hold:
\begin{itemize}
\item
  $\ttp{i}(E)$ is the number of pairs of acyclic paths in $G$
  to true passing by true;
\item
  $\tfp{i}(E)$ is the number of pairs of acyclic paths in $G$
  to false passing by true;
\item
  $\ffp{i}(E)$ is the number of pairs of acyclic paths in $G$
  to false passing by false;
\item
  $\ppp{i}(E)$ is the number of pairs of acyclic paths in $G$ that do
  not share an arc.
\end{itemize}
Note that in the following proofs,
a ``pair of paths'' will always means a pair of acyclic paths that do
not share an arc.

\begin{lemma}
\label{lm:expr-double-count}
Let $i \in \{ 0, 1, 2 \}$ be an optimization level,
$E \in \Exp$ be an expression, and let
$m, t, f \in \Nset$ be such that $m > \max\{t, f\}$.
Then:
\begin{align}
   \ttp{i}(E) &= \dapn\bigl(\gcfg{i}{E}(t, f, m), t, t\bigr), \\
   \tfp{i}(E) &= \dapn\bigl(\gcfg{i}{E}(t, f, m), t, f\bigr) = \dapn\bigl(\gcfg{i}{E}(t, f, m), f, t\bigr), \\
   \ffp{i}(E) &= \dapn\bigl(\gcfg{i}{E}(t, f, m), f, f\bigr), \\
   \ppp{i}(E) &= \dapn\bigl(\gcfg{i}{E}(t, t, m), t, t\bigr).
\end{align}
\end{lemma}

\begin{proof}
  The proof is by structural induction on $E$.

\proofsec{Variables}
By Definition~\ref{def:tt-ff-pp},
$\ttp{i}(E) = 0$,
$\tfp{i}(E) = 1$,
$\ffp{i}(E) = 0$,
$\ppp{i}(E) = 0$.
By~\eqref{eq:cfg_i-variables} and Definition~\ref{def:dapn}, we have
\begin{align*}
  \dapn\bigl(\gcfg{i}{E}(t, f, m), s_1, s_2\bigr)
  &= \dapn\Bigl(\big(\{ m,t,f \}, \{ (m,t),(m,f) \}, m\bigr), s_1, s_2\Bigr) \\
  &= \dapn\bigl(m, m, \{ (m,t), (m,f) \}, s_1, s_2\bigr) \\
  &= \dapn\bigl(m, s_1, \{ (m,t), (m,f) \} \setminus \{ (m, s_1) \}, s_1, s_2\bigr) \\
  &= \begin{cases}
      \dapn\bigl(m, t, \{ (m,f) \}, t, s_2\bigr)
        = \napn\bigl(m, \{ (m,f) \}, s_2\bigr),
          &\quad\textrm{if $s_1 = t$,} \\
      \dapn\bigl(m, f, \{ (m,t) \}, f, s_2\bigr)
        = \napn\bigl(m, \{ (m,t) \}, s_2\bigr),
          &\quad\textrm{if $s_1 = f$,}
    \end{cases}
\end{align*}
which evaluates to~$0$ if $s_1 = s_2$ and $1$ otherwise.
Similarly,
\begin{align*}
  \dapn\bigl(\gcfg{i}{E}(t, t, m), t, t\bigr)
  &= \dapn\Bigl(\big(\{ m,t \}, \{ (m,t) \}, m\bigr), t, t\Bigr) \\
  &= \dapn\bigl(m, \{ (m,t) \}, t, t\bigr) \\
  &= \napn(m, \emptyset, t) = 0.
\end{align*}

\proofsec{Constants}
There are three cases:
\begin{description}
\item[$\tv{i}(E) = \true$:]
By Definition~\ref{def:tt-ff-pp},
$\ttp{i}(E) = 1$,
$\tfp{i}(E) = 0$,
$\ffp{i}(E) = 0$,
$\ppp{i}(E) = 1$.
By~\eqref{eq:cfg_i-constants},
  $\gcfg{i}{E}(t, f, m) = \bigl(\{t\}, \emptyset, t\bigr)$ so that
\begin{align*}
  \dapn\bigl(\gcfg{i}{E}(t, f, m), s_1, s_2\bigr)
  &= \dapn\Bigl(\bigl(\{t\}, \emptyset, t\bigr), s_1, s_2\Bigr) \\
  &= \napn\bigl(t, \emptyset, s_2\bigr),
\end{align*}
which evaluates to~$1$ if $s_2 = t$ and $0$ otherwise.
Similarly,
\begin{align*}
  \dapn\bigl(\gcfg{i}{E}(t, t, m), t, t\bigr)
  &= \dapn\Bigl(\bigl(\{t\}, \emptyset, t\bigr), t, t\Bigr) \\
  &= \napn\bigl(t, \emptyset, t\bigr) = 1.
\end{align*}
\item[$\tv{i}(E) = \false$:]
  the proof is similar to the previous case;
\item[$\tv{i}(E) = \maybe$:]
  the proof is similar to the case $\tv{i}(E) = \maybe$ for variables.
\end{description}

\proofsec{Logical negation} By~\eqref{eq:cfg_i-negExpr},
$\gcfg{i}{\cnot E}(t, f, m) = \gcfg{i}{E}(f, t, m)$; hence, by
Definition~\ref{def:tt-ff-pp} and the inductive hypothesis,
\begin{align*}
  \ffp{i}(\cnot E) &= \ttp{i}(E) \\
            &= \dapn\bigl(\gcfg{i}{E}(f, t, m), f, f\bigr) \\
            &= \dapn\bigl(\gcfg{i}{!E}(t, f, m), f, f\bigr), \\
  \tfp{i}(\cnot E) &= \tfp{i}(E)  \\
            &= \dapn\bigl(\gcfg{i}{E}(f, t, m), t, f\bigr) \\
            &= \dapn\bigl(\gcfg{i}{!E}(t, f, m), t, f\bigr), \\
  \tfp{i}(\cnot E) &= \tfp{i}(E) \\
            &= \dapn\bigl(\gcfg{i}{E}(f, t, m), f, t\bigr) \\
            &= \dapn\bigl(\gcfg{i}{!E}(t, f, m), f, t\bigr), \\
  \ttp{i}(\cnot E) &= \ffp{i}(E) \\
            &= \dapn\bigl(\gcfg{i}{E}(f, t, m), t, t\bigr) \\
            &= \dapn\bigl(\gcfg{i}{!E}(t, f, m), t, t\bigr), \\
  \ppp{i}(\cnot E) &= \ppp{i}(E) \\
            &= \dapn\bigl(\gcfg{i}{E}(t, t, m), t, t\bigr) \\
            &= \dapn\bigl(\gcfg{i}{!E}(t, t, m), t, t\bigr).
\end{align*}

\proofsec{Unary plus, unary minus, parentheses and cast operators}
By~\eqref{eq:cfg_i-specUnaryExpr}, we have
$\gcfg{i}{E}(t, f, m) = \gcfg{i}{E_1}(t, f, m)$; hence, by
Definition~\ref{def:tt-ff-pp} and the inductive hypothesis,
\begin{alignat*}{3}
  \ttp{i}(E) &= \ttp{i}(E_1)
            &&= \dapn\bigl(\gcfg{i}{E_1}(t, f, m), t, t\bigr)
            &&= \dapn\bigl(\gcfg{i}{E}(t, f, m), t, t\bigr), \\
  \tfp{i}(E) &= \tfp{i}(E_1)
            &&= \dapn\bigl(\gcfg{i}{E_1}(t, f, m), t, f\bigr)
            &&= \dapn\bigl(\gcfg{i}{E}(t, f, m), t, f\bigr), \\
  \tfp{i}(E) &= \tfp{i}(E_1)
            &&= \dapn\bigl(\gcfg{i}{E_1}(t, f, m), f, t\bigr)
            &&= \dapn\bigl(\gcfg{i}{E}(t, f, m), f, t\bigr), \\
  \ffp{i}(E) &= \ffp{i}(E_1)
            &&= \dapn\bigl(\gcfg{i}{E_1}(t, f, m), f, f\bigr)
            &&= \dapn\bigl(\gcfg{i}{E}(t, f, m), f, f\bigr), \\
  \ppp{i}(E) &= \ppp{i}(E_1)
            &&= \dapn\bigl(\gcfg{i}{E_1}(t, t, m), t, t\bigr)
            &&= \dapn\bigl(\gcfg{i}{E}(t, t, m), t, t\bigr).
\end{alignat*}

\proofsec{Other unary operators}
Let $\gcfg{i}{\uop E_1}(t, f, m)$ be as defined
in~\eqref{eq:cfg_i-specUnaryExpr}.
For each pair of paths to~$m$ in $\gcfg{i}{E_1}(m, m, m+1)$
there is $1$ pair of paths, one to~$t$ and the other to
$f$, and $0$ pairs of paths both
to~$t$ (resp., $f$)
in $\gcfg{i}{\uop E_1}(t, f, m)$,
and also $0$ pairs of paths both to an exit node
in $\gcfg{i}{E_1}(t, t, m)$.
Hence, by
Definition~\ref{def:tt-ff-pp} and the inductive hypothesis, we have:
\begin{alignat*}{3}
  \ttp{i}(\uop E_1) &= 0 \\
                    &= \dapn\bigl(\gcfg{i}{\uop E_1}(t, f, m), t, t\bigr), \\
  \tfp{i}(\uop E_1) &= \ppp{i}(E_1) \\
                    &= \dapn\bigl(\gcfg{i}{E_1}(t, t, m), t, t\bigr) \\
                    &= \dapn\bigl(\gcfg{i}{\uop E_1}(t, f, m), t, f\bigr), \\
  \tfp{i}(\uop E_1) &= \ppp{i}(E_1) \\
                    &= \dapn\bigl(\gcfg{i}{E_1}(t, t, m), t, t\bigr) \\
                    &= \dapn\bigl(\gcfg{i}{\uop E_1}(t, f, m), f, t\bigr), \\
  \ffp{i}(\uop E_1) &= 0 \\
                    &= \dapn\bigl(\gcfg{i}{\uop E_1}(t, f, m), f, f\bigr), \\
  \ppp{i}(\uop E_1) &= 0 \\
                    &= \dapn\bigl(\gcfg{i}{\uop E_1}(t, t, m), t, t\bigr).
\end{alignat*}

\proofsec{Logical conjunction}
There are three cases:
\begin{description}
\item [$\tv{i}(E_1) = \false$:]
  then, by~\eqref{eq:cfg_i-conjExpr},
  $\gcfg{i}{E_1 \cand E_2}(t, f, m) = \bigl(\{f\}, \emptyset, f\bigr)$,
  so $\dapn$ in this case is similar to that of the false
  constant case.
  By~\eqref{eq:cfg_i-constants},
  $\ffp{i}(E_1) = 1$ and $\ttp{i}(E_1) = \tfp{i}(E_1) = 0$.
  Hence, by Definition~\ref{def:tt-ff-pp},
\begin{align*}
  \ttp{i}(E_1 \cand E_2)
    &= \ttp{i}(E_1) \ttp{i}(E_2) \\
    &= 0 \\
    &= \dapn\bigl(\gcfg{i}{E_1 \cand E_2}(t, f, m), t, t\bigr), \\
  \tfp{i}(E_1 \cand E_2)
    &= \tfp{i}(E_1) \tp{i}(E_2) + \ttp{i}(E_1) \tfp{i}(E_2) \\
    &= 0 \\
    &= \dapn\bigl(\gcfg{i}{E_1 \cand E_2}(t, f, m), t, f\bigr), \\
  \tfp{i}(E_1 \cand E_2)
    &= \tfp{i}(E_1) \tp{i}(E_2) + \ttp{i}(E_1) \tfp{i}(E_2) \\
    &= 0 \\
    &= \dapn\bigl(\gcfg{i}{E_1 \cand E_2}(t, f, m), f, t\bigr), \\
  \ffp{i}(E_1 \cand E_2)
    &= \ffp{i}(E_1) + 2 \tfp{i}(E_1) \fp{i}(E_2) + \ttp{i}(E_1) \ffp{i}(E_2) \\
    &= 1 + 0 + 0 \\
    &= \dapn\bigl(\gcfg{i}{E_1 \cand E_2}(t, f, m), f, f\bigr), \\
  \ppp{i}(E_1 \cand E_2)
    &= \ffp{i}(E_1) + 2 \tfp{i}(E_1) \pp{i}(E_2) + \ttp{i}(E_1) \ppp{i}(E_2) \\
    &= 1 + 0 + 0 \\
    &= \dapn\bigl(\gcfg{i}{E_1 \cand E_2}(t, t, m), t, t\bigr).
\end{align*}
\item [$\tv{i}(E_1) = \true$:]
  then $\gcfg{i}{E_1 \cand E_2}(t, f, m) = \gcfg{i}{E_2}(t, f, m)$ and,
  by~\eqref{eq:cfg_i-constants}, we have
  $\ttp{i}(E_1) = 1$ and $\tfp{i}(E_1) = \ffp{i}(E_1) = 0$.
  Hence, by Definition~\ref{def:tt-ff-pp} and the inductive hypothesis:
\begin{align*}
  \ttp{i}(E_1 \cand E_2) &= \ttp{i}(E_1) \ttp{i}(E_2)  \\
                        &= \ttp{i}(E_2) \\
                        &= \dapn\bigl(\gcfg{i}{E_2}(t, f, m), t, t\bigr), \\
                        &= \dapn\bigl(\gcfg{i}{E_1 \cand E_2}(t, f, m), t, t\bigr), \\
  \tfp{i}(E_1 \cand E_2) &= \tfp{i}(E_1) \tp{i}(E_2) +
                            \ttp{i}(E_1) \tfp{i}(E_2) \\
                        &= \tfp{i}(E_2) \\
                        &= \dapn\bigl(\gcfg{i}{E_2}(t, f, m), t, f\bigr), \\
                        &= \dapn\bigl(\gcfg{i}{E_1 \cand E_2}(t, f, m), t, f\bigr), \\
  \tfp{i}(E_1 \cand E_2) &= \tfp{i}(E_1) \tp{i}(E_2) +
                            \ttp{i}(E_1) \tfp{i}(E_2) \\
                        &= \tfp{i}(E_2) \\
                        &= \dapn\bigl(\gcfg{i}{E_2}(t, f, m), f, t\bigr), \\
                        &= \dapn\bigl(\gcfg{i}{E_1 \cand E_2}(t, f, m), f, t\bigr), \\
  \ffp{i}(E_1 \cand E_2) &= \ffp{i}(E_1) + 2 \tfp{i}(E_1) \fp{i}(E_2) +
                            \ttp{i}(E_1) \ffp{i}(E_2) \\
                        &= \ffp{i}(E_2) \\
                        &= \dapn\bigl(\gcfg{i}{E_2}(t, f, m), f, f\bigr), \\
                        &= \dapn\bigl(\gcfg{i}{E_1 \cand E_2}(t, f, m), f, f\bigr), \\
  \ppp{i}(E_1 \cand E_2) &= \ffp{i}(E_1) + 2 \tfp{i}(E_1) \pp{i}(E_2) +
                            \ttp{i}(E_1) \ppp{i}(E_2) \\
                        &= \ppp{i}(E_2) \\
                        &= \dapn\bigl(\gcfg{i}{E_2}(t, t, m), t, t\bigr), \\
                        &= \dapn\bigl(\gcfg{i}{E_1 \cand E_2}(t, t, m), t, t\bigr). \\
\end{align*}
\item [$\tv{i}(E_1) = \maybe$:]
  let $\gcfg{i}{E_1 \cand E_2}(t, f, m)$ be defined as
  in~\eqref{eq:cfg_i-conjExpr}.

  Since $\gcfg{i}{E_1}(s_2, f, m_1)$
  has the exit node for true evaluation equal to the entry node of
  $\gcfg{i}{E_2}(t, f, m)$, for each
  pair of paths in $\gcfg{i}{E_1}(s_2, f, m_1)$ both to~$s_2$,
  there are in $\gcfg{i}{E_1 \cand E_2}(t, f, m)$:
  \begin{itemize}
    \item
      $\dapn\bigl(\gcfg{i}{E_2}(t, f, m), t, t\bigr)$ pairs of paths
      both to~$t$,
    \item
      $\dapn\bigl(\gcfg{i}{E_2}(t, f, m), t, f\bigr)$ pairs of paths
      the first to~$t$ and the second to~$f$,
    \item
      $\dapn\bigl(\gcfg{i}{E_2}(t, f, m), f, t\bigr)$ pairs of paths
      the first to~$f$ and the second to~$t$,
    \item
      $\dapn\bigl(\gcfg{i}{E_2}(t, f, m), f, f\bigr)$ pairs of paths
      both to~$f$,
  \end{itemize}
  and in $\gcfg{i}{E_1 \cand E_2}(t, t, m)$,
  \begin{itemize}
    \item
      $\dapn\bigl(\gcfg{i}{E_2}(t, t, m), t, t\bigr)$ pairs of
      paths both to an exit node.
  \end{itemize}

  For each pair of paths, in $\gcfg{i}{E_1}(s_2, f, m_1)$, where the
  first is to~$s_2$ and the second one to~$f$:
  there are in $\gcfg{i}{E_1 \cand E_2}(t, f, m)$:
  \begin{itemize}
    \item
      $\napn\bigl(\gcfg{i}{E_2}(t, f, m), t\bigr)$ pairs of paths,
      the first to~$t$ and the second to~$f$,
    \item
      $\napn\bigl(\gcfg{i}{E_2}(t, f, m), f\bigr)$ pairs of paths
      both to~$f$ in $\gcfg{i}{E_1 \cand E_2}(t, f, m)$,
  \end{itemize}
  and in $\gcfg{i}{E_1 \cand E_2}(t, t, m)$,
  \begin{itemize}
    \item
      $\apn\bigl(\gcfg{i}{E_2}(t, t, m)\bigr)$ pairs of paths both
      to an exit node.
  \end{itemize}

  For each pair of paths in $\gcfg{i}{E_1}(s_2, f, m_1)$, where the
  first is to~$f$ and the second to~$s_2$,
  there are in $\gcfg{i}{E_1 \cand E_2}(t, f, m)$:
  \begin{itemize}
    \item
      $\napn\bigl(\gcfg{i}{E_2}(t, f, m), t\bigr)$ pairs of paths,
      the first to~$f$ and the second to~$t$,
    \item
      $\napn\bigl(\gcfg{i}{E_2}(t, f, m), f\bigr)$ pairs of paths
      both to~$f$ in $\gcfg{i}{E_1 \cand E_2}(t, f, m)$,
  \end{itemize}
  and in $\gcfg{i}{E_1 \cand E_2}(t, t, m)$,
  \begin{itemize}
    \item
      $\apn\bigl(\gcfg{i}{E_2}(t, t, m)\bigr)$ pairs of paths
      both to an exit node.
  \end{itemize}

  In addition, there are $\dapn\bigl(\gcfg{i}{E_1}(s_2, f, m), f, f\bigr)$
  pairs of paths in $\gcfg{i}{E_1}(s_2, f, m_1)$ both to~$f$.

  Hence, by Definition~\ref{def:tt-ff-pp} and the
  inductive hypothesis:
\begin{align*}
  \ttp{i}(E_1 \cand E_2) &= \ttp{i}(E_1) \ttp{i}(E_2)  \\
                         &= \dapn\bigl(\gcfg{i}{E_1}(s_2, f, m), s_2, s_2\bigr)
                           \dapn\bigl(\gcfg{i}{E_2}(t, f, m), t, t\bigr) \\
                         &= \dapn\bigl(\gcfg{i}{E_1 \cand E_2}(t, f, m), t, t\bigr), \\
  \tfp{i}(E_1 \cand E_2) &= \tfp{i}(E_1) \tp{i}(E_2) +
                            \ttp{i}(E_1) \tfp{i}(E_2) \\
                         &= \dapn\bigl(\gcfg{i}{E_1}(s_2, f, m), s_2, f\bigr)
                           \napn\bigl(\gcfg{i}{E_2}(t, f, m), t\bigr) \\
                         & \quad \mathord{}
                           + \dapn\bigl(\gcfg{i}{E_1}(s_2, f, m), s_2, s_2\bigr)
                           \dapn\bigl(\gcfg{i}{E_2}(t, f, m), t, f\bigr) \\
                         &= \dapn\bigl(\gcfg{i}{E_1 \cand E_2}(t, f, m), t, f\bigr), \\
  \tfp{i}(E_1 \cand E_2) &= \tfp{i}(E_1) \tp{i}(E_2) +
                            \ttp{i}(E_1) \tfp{i}(E_2) \\
                         &= \dapn\bigl(\gcfg{i}{E_1}(s_2, f, m), s_2, f\bigr)
                           \napn\bigl(\gcfg{i}{E_2}(t, f, m), t\bigr) \\
                         & \quad \mathord{}
                           + \dapn\bigl(\gcfg{i}{E_1}(s_2, f, m), s_2, s_2\bigr)
                           \dapn\bigl(\gcfg{i}{E_2}(t, f, m), f, t\bigr) \\
                         &= \dapn\bigl(\gcfg{i}{E_1 \cand E_2}(t, f, m), f, t\bigr), \\
  \ffp{i}(E_1 \cand E_2) &= \ffp{i}(E_1) + 2 \tfp{i}(E_1) \fp{i}(E_2) +
                            \ttp{i}(E_1) \ffp{i}(E_2) \\
                         &= \dapn\bigl(\gcfg{i}{E_1}(s_2, f, m), f, f\bigr) \\
                         & \quad \mathord{}
                           + \dapn\bigl(\gcfg{i}{E_1}(s_2, f, m), s_2, f\bigr)
                           \napn\bigl(\gcfg{i}{E_2}(t, f, m), t\bigr) \\
                         & \quad \mathord{}
                           + \dapn\bigl(\gcfg{i}{E_1}(s_2, f, m), f, s_2\bigr)
                           \napn\bigl(\gcfg{i}{E_2}(t, f, m), f\bigr) \\
                         & \quad \mathord{}
                           + \dapn\bigl(\gcfg{i}{E_1}(s_2, f, m), s_2, s_2\bigr)
                            \dapn\bigl(\gcfg{i}{E_2}(t, f, m), f,  f\bigr) \\
                         &= \dapn\bigl(\gcfg{i}{E_1 \cand E_2}(t, f, m), f, f\bigr), \\
  \ppp{i}(E_1 \cand E_2) &= \ffp{i}(E_1) + 2 \tfp{i}(E_1) \pp{i}(E_2) +
                            \ttp{i}(E_1) \ppp{i}(E_2) \\
                         &= \dapn\bigl(\gcfg{i}{E_1}(s_2, f, m), f, f\bigr) \\
                         & \quad \mathord{}
                            + \dapn\bigl(\gcfg{i}{E_1}(s_2, f, m), s_2, f\bigr)
                           \napn\bigl(\gcfg{i}{E_2}(t, t, m), t\bigr) \\
                         & \quad \mathord{}
                            + \dapn\bigl(\gcfg{i}{E_1}(s_2, f, m), f, s_2\bigr)
                           \napn\bigl(\gcfg{i}{E_2}(t, t, m), t\bigr) \\
                         & \quad \mathord{}
                           + \dapn\bigl(\gcfg{i}{E_1}(s_2, s_2, m), s_2, s_2\bigr)
                           \dapn\bigl(\gcfg{i}{E_2}(t, t, m), t, t\bigr) \\
                         &= \dapn\bigl(\gcfg{i}{E_1 \cand E_2}(t, t, m), t, t\bigr). \\
\end{align*}
\end{description}

\proofsec{Logical disjunction}
Dual to the case of logical conjunction.

\proofsec{Comma operator}
Let $\gcfg{i}{E_1, E_2}(t, f, m)$ be defined
as in~\eqref{eq:cfg_i-commaExpr}.
As the exit node for $\gcfg{i}{E_1}(s_2, s_2, m_1)$ and the entry node of $\gcfg{i}{E_2}(t, f, m)$ must be the same,
for each pair of paths to~$s_2$ in $\gcfg{i}{E_1}(s_2, s_2, m_1)$,
there are in $\gcfg{i}{E_1, E_2}(t, f, m)$:
\begin{itemize}
\item
  $\dapn\bigl(\gcfg{i}{E_2}(t, f, m), t, t\bigr)$ pairs of paths both to~$t$,
\item
  $\dapn\bigl(\gcfg{i}{E_2}(t, f, m), t, f\bigr)$ pairs of paths, the first to~$t$ and the second to~$f$,
\item
  $\dapn\bigl(\gcfg{i}{E_2}(t, f, m), f, t\bigr)$ pairs of paths, the first to~$f$ and the second to~$t$,
\item
  $\dapn\bigl(\gcfg{i}{E_2}(t, f, m), f, f\bigr)$ pairs of paths both to~$f$,
\end{itemize}
and in $\gcfg{i}{E_1, E_2}(t, t, m)$,
\begin{itemize}
\item
  $\dapn\bigl(\gcfg{i}{E_2}(t, t, m), t, t\bigr)$ pairs of paths
  both to an exit node.
\end{itemize}
Hence, by Definition~\ref{def:tt-ff-pp} and the inductive hypothesis:
\begin{align*}
  \ttp{i}(E_1, E_2)
    &= \ppp{i}(E_1) \ttp{i}(E_2) \\
    &= \dapn\bigl(\gcfg{i}{E_1}(s_2, s_2, m_1), s_2, s_2\bigr) \dapn\bigl(\gcfg{i}{E_2}(t, f, m), t, t\bigr) \\
    &= \dapn\bigl(\gcfg{i}{E_1, E_2}(t, f, m), t, t\bigr), \\
  \tfp{i}(E_1, E_2) &= \ppp{i}(E_1) \tfp{i}(E_2) \\
    &= \dapn\bigl(\gcfg{i}{E_1}(s_2, s_2, m_1), s_2, s_2\bigr) \dapn\bigl(\gcfg{i}{E_2}(t, f, m), t, f\bigr) \\
    &= \dapn\bigl(\gcfg{i}{E_1, E_2}(t, f, m), t, f\bigr), \\
  \tfp{i}(E_1, E_2) &= \ppp{i}(E_1) \tfp{i}(E_2) \\
    &= \dapn\bigl(\gcfg{i}{E_1}(s_2, s_2, m_1), s_2, s_2\bigr) \dapn\bigl(\gcfg{i}{E_2}(t, f, m), f, t\bigr) \\
    &= \dapn\bigl(\gcfg{i}{E_1, E_2}(t, f, m), f, t\bigr), \\
  \ffp{i}(E_1, E_2) &= \ppp{i}(E_1) \ffp{i}(E_2) \\
    &= \dapn\bigl(\gcfg{i}{E_1}(s_2, s_2, m_1), s_2, s_2\bigr) \dapn\bigl(\gcfg{i}{E_2}(t, f, m), f, f\bigr) \\
    &= \dapn\bigl(\gcfg{i}{E_1, E_2}(t, f, m), f, f\bigr), \\
  \ppp{i}(E_1, E_2) &= \ppp{i}(E_1) \ppp{i}(E_2) \\
    &= \dapn\bigl(\gcfg{i}{E_1}(s_2, s_2, m_1), s_2, s_2\bigr) \napn\bigl(\gcfg{i}{E_2}(t, t, m), t, t\bigr) \\
    &= \dapn\bigl(\gcfg{i}{E_1, E_2}(t, t, m), t, t\bigr).
\end{align*}

\proofsec{Binary conditional operator}
this case is identical to the case of logical disjunction,
as the CFGs are the same and functions $\tp{i}, \fp{i}$ and $\pp{i}$
are defined the same.

\proofsec{Other binary operators} let $\gcfg{i}{E_1 \bop E_2}(t, f,
m)$ be defined as in~\eqref{eq:cfg_i-unaryExpr}.  As the exit node
for $\gcfg{i}{E_1}(s_2, s_2, m_1)$ and the entry node for
$\gcfg{i}{E_2}(m, m, m+1)$ must be the same, for each pair of paths
both to~$s_2$, there are in $\gcfg{i}{E_1}(s_2, s_2, m_1)$,
\begin{itemize}
\item
$\dapn\bigl(\gcfg{i}{E_2}(m, m, m+1), m, m\bigr)$ pairs of paths both to~$m$.
\end{itemize}
Moreover, since $m$ is the second-last node through which each path
directed to an exit node passes in $\gcfg{i}{E_1 \bop E_2}(t, f, m)$
and the only arcs that exit from $m$ are $(m, t\bigr),
(m, f\bigr) \in A$, there are in $\gcfg{i}{E_1 \bop E_2}(t, f, m)$:
\begin{itemize}
\item
  $0$ pairs of paths both to~$t$,
\item
  $1$ pair of paths to~$t$ in the first traversal and $f$ in the second one,
\item
  $1$ pair of paths to~$f$ in the first traversal and $t$ in the second one,
\item
  $0$ pairs of paths both to~$f$,
\end{itemize}
and in $\gcfg{i}{E_1 \bop E_2}(t, t, m)$,
\begin{itemize}
\item
  $0$ pairs of paths both to the exit node~$t$.
\end{itemize}
Hence, by Definition~\ref{def:tt-ff-pp} and the inductive hypothesis:
\begin{align*}
  \ttp{i}(E_1 \bop E_2) &= 0 \\
    &= \dapn\bigl(\gcfg{i}{E_1 \bop E_2}(t, f, m), t, t\bigr), \\
  \tfp{i}(E_1 \bop E_2) &= \ppp{i}(E_1) \ppp{i}(E_2) \\
    &= \dapn\bigl(\gcfg{i}{E_1}(s_2, s_2, m_1), s_2, s_2\bigr)
       \dapn\bigl(\gcfg{i}{E_2}(m, m, m+1), m, m\bigr) \\
    &= \dapn\bigl(\gcfg{i}{E_1 \bop E_2}(t, f, m), t, f\bigr), \\
  \tfp{i}(E_1 \bop E_2) &= \ppp{i}(E_1) \ppp{i}(E_2) \\
    &= \dapn\bigl(\gcfg{i}{E_1}(s_2, s_2, m_1), s_2, s_2\bigr)
       \dapn\bigl(\gcfg{i}{E_2}(m, m, m+1), m, m\bigr) \\
    &= \dapn\bigl(\gcfg{i}{E_1 \bop E_2}(t, f, m), f, t\bigr), \\
  \ffp{i}(E_1 \bop E_2) &= 0 \\
    &= \dapn\bigl(\gcfg{i}{E_1 \bop E_2}(t, f, m), f, f\bigr), \\
  \ppp{i}(E_1 \bop E_2) &= 0 \\
    &= \dapn\bigl(\gcfg{i}{E_1 \bop E_2}(t, t, m), t, t\bigr).
\end{align*}

\proofsec{Conditional operator}
Then $E = E_1 \cqmk E_2 \ccol E_3$.
There are three cases:
\begin{description}
\item[$\tv{i}(E_1) = \true$:]
  then, by~\eqref{eq:cfg_i-condExpr} and~\eqref{eq:cfg_i-conjExpr},
  \[
  \gcfg{i}{E_1 \cqmk E_2 \ccol E_3}(t, f, m) = \gcfg{i}{E_2}(t, f, m)
  = \gcfg{i}{E_1 \cand E_2}(t, f, m).
  \]
  Also, by Definition~\ref{def:tt-ff-pp},
\begin{align*}
  \ttp{i}(E_1 \cqmk E_2 \ccol E_3)
  &= \ttp{i}(E_2)
  = \ttp{i}(E_1 \cand E_2), \\
  \tfp{i}(E_1 \cqmk E_2 \ccol E_3)
  &= \tfp{i}(E_2)
  = \tfp{i}(E_1 \cand E_2), \\
  \ffp{i}(E_1 \cqmk E_2 \ccol E_3)
  &= \ffp{i}(E_2)
  = \ffp{i}(E_1 \cand E_2).
\end{align*}
   Therefore this case is equivalent to the same case for logical conjunction.
\item [$\tv{i}(E_1) = 0$:] symmetric to the previous case.
\item [$\tv{i}(E_1) = \maybe$:]
  let $\gcfg{i}{E}(t,f,m) = \gcfg{i}{E_1 \cqmk E_2 \ccol E_3}(t, f, m)$
  be as defined in~\eqref{eq:cfg_i-condExpr}.
  Since $\gcfg{i}{E_1}(s_2, s_3, m_2)$ has the exit node for true evaluation
  equal to the entry node of $\gcfg{i}{E_2}(t, f, m)$,
  for each pair of paths both to~$s_2$ in $\gcfg{i}{E_1}(s_2, s_3, m_2)$
  there are in $\gcfg{i}{E}(t, f, m)$:
  \begin{itemize}
    \item
      $\dapn\bigl(\gcfg{i}{E_2}(t, f, m), t, t\bigr)$ pairs of paths both to~$t$,
    \item
      $\dapn\bigl(\gcfg{i}{E_2}(t, f, m), t, f\bigr)$ pairs of paths, the first to~$t$ and the second to~$f$,
    \item
      $\dapn\bigl(\gcfg{i}{E_2}(t, f, m), f, t\bigr)$ pairs of paths, the first to~$f$ and the second to~$t$,
    \item
      $\dapn\bigl(\gcfg{i}{E_2}(t, f, m), f, f\bigr)$ pairs of paths both to~$f$;
  \end{itemize}
  and in $\gcfg{i}{E}(t, t, m)$,
  \begin{itemize}
    \item
      $\dapn\bigl(\gcfg{i}{E_2}(t, f, m), t, t\bigr)$ pairs of paths both to an exit node.
  \end{itemize}

  Similarly, as $\gcfg{i}{E_1}(s_2, s_3, m_2)$ has the exit node for
  false evaluation equal to the entry node of
  $\gcfg{i}{E_3}(t, f, m)$, for each pair of paths both to~$s_3$ in
  $\gcfg{i}{E_1}(s_2, s_3, m_2)$ there are in $\gcfg{i}{E}(t, f, m)$:
  \begin{itemize}
    \item
      $\dapn\bigl(\gcfg{i}{E_3}(t, f, m), t, t\bigr)$ pairs of paths
      both to~$t$,
    \item
      $\dapn\bigl(\gcfg{i}{E_3}(t, f, m), t, f\bigr)$ pairs of paths,
      the first to~$t$ and the second to~$f$,
    \item
      $\dapn\bigl(\gcfg{i}{E_3}(t, f, m), f, t\bigr)$ pairs of paths,
      the first to~$f$ and the second to~$t$,
    \item
      $\dapn\bigl(\gcfg{i}{E_3}(t, f, m), f, f\bigr)$ pairs of paths
      both to~$f$;
  \end{itemize}
  and in $\gcfg{i}{E}(t, t, m)$,
  \begin{itemize}
     \item
       $\dapn\bigl(\gcfg{i}{E_3}(t, t, m), t, t\bigr)$ pairs of paths
       both to an exit node.
   \end{itemize}

  In addition, for each of the $2 \cdot \dapn\bigl(\gcfg{i}{E_1}(s_2,
  s_3, m_2), s_2, s_3\bigr)$ pairs of paths, where the two paths
  evaluate to different truth values, there are in $\cfg{i}{E}(t, f, m)$:
    \begin{itemize}
    \item
      $\napn\bigl(\gcfg{i}{E_2}(t, f, m), t\bigr) \napn\bigl(\gcfg{i}{E_3}(t, f, m), t\bigr)$
      pairs of paths both
      to~$t$,
    \item
      $\napn\bigl(\gcfg{i}{E_2}(t, f, m), t\bigr)
      \napn\bigl(\gcfg{i}{E_3}(t, f, m), f\bigr) +
      \napn\bigl(\gcfg{i}{E_2}(t, f, m), f\bigr)
      \napn\bigl(\gcfg{i}{E_3}(t, f, m), t\bigr)$ pairs of paths, the
      first to~$t$ and the second to~$f$,
    \item
      $\napn\bigl(\gcfg{i}{E_2}(t, f, m), f\bigr)
      \napn\bigl(\gcfg{i}{E_3}(t, f, m), t\bigr) +
      \napn\bigl(\gcfg{i}{E_2}(t, f, m), t\bigr)
      \napn\bigl(\gcfg{i}{E_3}(t, f, m), f\bigr)$ pairs of paths, the
      first to~$f$ and the second to~$t$,
    \item
      $\napn\bigl(\gcfg{i}{E_2}(t, f, m), f\bigr) \cdot
      \napn\bigl(\gcfg{i}{E_3}(t, f, m), f\bigr)$ both to~$f$;
  \end{itemize}
  and in $\gcfg{i}{E}(t, t, m)$,
  \begin{itemize}
    \item
      $\apn\bigl(\gcfg{i}{E_2}(t, f, m)\bigr)
      \apn\bigl(\gcfg{i}{E_3}(t, f, m)\bigr)$ pairs of paths both to
      exit node~$t$.
  \end{itemize}
  Hence, by Definition~\ref{def:tt-ff-pp} and the inductive hypothesis:
\begin{align*}
  \ttp{i}(E_1 \cqmk E_2 \ccol E_3)
    &= \ttp{i}(E_1) \ttp{i}(E_2)  \\
        & \quad \mathord{}
        + 2 \tfp{i}(E_1) \tp{i}(E_2) \tp{i}(E_3) + \ffp{i}(E_1) \ttp{i}(E_3) \\
    &= \dapn\bigl(\gcfg{i}{E_1}(s_2, s_3, m_2), s_2, s_2\bigr) \dapn\bigl(\gcfg{i}{E_2}(t, f, m), t, t\bigr) \\
    & \quad \mathord{}
      + \dapn\bigl(\gcfg{i}{E_1}(s_2,s_3, m_2), s_2, s_3\bigr) \\
          &\qquad\cdot
    \napn\bigl(\gcfg{i}{E_2}(t, f, m), t\bigr)
    \napn\bigl(\gcfg{i}{E_3}(t, f, m), t\bigr) \\
    & \quad \mathord{}
      + \dapn\bigl(\gcfg{i}{E_1}(s_2,s_3, m_2), s_3, s_2\bigr) \\
          &\qquad\cdot
    \napn\bigl(\gcfg{i}{E_2}(t, f, m), t\bigr)
    \napn\bigl(\gcfg{i}{E_3}(t, f, m), t\bigr) \\
    & \quad \mathord{}
      + \dapn\bigl(\gcfg{i}{E_1}(s_2, s_3, m_2), s_3, s_3\bigr) \dapn\bigl(\gcfg{i}{E_3}(t, f, m), t, t\bigr) \\
    &= \dapn\bigl(\gcfg{i}{E_1 \cqmk E_2 \ccol E_3}(t, f, m), t, t\bigr), \\
  \tfp{i}(E_1 \cqmk E_2 \ccol E_3)
    &= \ttp{i}(E_1) \tfp{i}(E_2)  \\
        & \quad \mathord{}
          + \tfp{i}(E_1)\bigl(\tp{i}(E_2) \fp{i}(E_3) + \fp{i}(E_2)+ \tp{i}(E_2)\bigr) \\
        & \quad \mathord{}  + \ffp{i}(E_1) \tfp{i}(E_3) \\
    &= \dapn\bigl(\gcfg{i}{E_1}(s_2, s_3, m_2), s_2, s_2\bigr) \dapn\bigl(\gcfg{i}{E_2}(t, f, m), t, f\bigr) \\
    & \quad \mathord{}
      + \dapn\bigl(\gcfg{i}{E_1}(s_2,s_3, m_2), s_2, s_3\bigr) \\
          &\qquad\cdot
    \napn\bigl(\gcfg{i}{E_2}(t, f, m), t\bigr)
    \napn\bigl(\gcfg{i}{E_3}(t, f, m), f\bigr) \\
    & \quad \mathord{}
      + \dapn\bigl(\gcfg{i}{E_1}(s_2,s_3, m_2), s_3, s_2\bigr) \\
          &\qquad\cdot
    \napn\bigl(\gcfg{i}{E_2}(t, f, m), f\bigr)
    \napn\bigl(\gcfg{i}{E_3}(t, f, m), t\bigr) \\
    & \quad \mathord{}
      + \dapn\bigl(\gcfg{i}{E_1}(s_2, s_3, m_2), s_3, s_3\bigr) \dapn\bigl(\gcfg{i}{E_3}(t, f, m), t, f\bigr) \\
    &= \dapn\bigl(\gcfg{i}{E_1 \cqmk E_2 \ccol E_3}(t, f, m), t, f\bigr), \\
  \tfp{i}(E_1 \cqmk E_2 \ccol E_3)
    &= \ttp{i}(E_1) \tfp{i}(E_2) \\
        & \quad \mathord{}
        + \tfp{i}(E_1)\bigl(\tp{i}(E_2) \fp{i}(E_3) + \fp{i}(E_2) + \tp{i}(E_2)\bigr)  \\
        & \quad \mathord{} + \ffp{i}(E_1) \tfp{i}(E_3) \\
    &= \dapn\bigl(\gcfg{i}{E_1}(s_2, s_3, m_2), s_2, s_2\bigr) \dapn\bigl(\gcfg{i}{E_2}(t, f, m), f, t\bigr) \\
    & \quad \mathord{}
      + \dapn\bigl(\gcfg{i}{E_1}(s_2,s_3, m_2), s_2, s_3\bigr) \\
          &\qquad\cdot
    \napn\bigl(\gcfg{i}{E_2}(t, f, m), f\bigr)
    \napn\bigl(\gcfg{i}{E_3}(t, f, m), t\bigr) \\
    & \quad \mathord{}
      + \dapn\bigl(\gcfg{i}{E_1}(s_2,s_3, m_2), s_3, s_2\bigr) \\
          &\qquad\cdot
    \napn\bigl(\gcfg{i}{E_2}(t, f, m), t\bigr)
    \napn\bigl(\gcfg{i}{E_3}(t, f, m), f\bigr) \\
    & \quad \mathord{}
      + \dapn\bigl(\gcfg{i}{E_1}(s_2, s_3, m_2), s_3, s_3\bigr) \dapn\bigl(\gcfg{i}{E_3}(t, f, m), f, t\bigr) \\
    &= \dapn\bigl(\gcfg{i}{E_1 \cqmk E_2 \ccol E_3}(t, f, m), f, t\bigr), \\
  \ffp{i}(E_1 \cqmk E_2 \ccol E_3)
    &= \ttp{i}(E_1) \ffp{i}(E_2)  \\
        & \quad \mathord{}
        + 2 \tfp{i}(E_1) \fp{i}(E_2) \fp{i}(E_3) + \ffp{i}(E_1) \ffp{i}(E_3) \\
    &= \dapn\bigl(\gcfg{i}{E_1}(s_2, s_3, m_2), s_2, s_2\bigr) \dapn\bigl(\gcfg{i}{E_2}(t, f, m), f, f\bigr) \\
    & \quad \mathord{}
      + \dapn\bigl(\gcfg{i}{E_1}(s_2,s_3, m_2), s_2, s_3\bigr) \\
          &\qquad\cdot
    \napn\bigl(\gcfg{i}{E_2}(t, f, m), f\bigr)
    \napn\bigl(\gcfg{i}{E_3}(t, f, m), f\bigr) \\
    & \quad \mathord{}
      + \dapn\bigl(\gcfg{i}{E_1}(s_2,s_3, m_2), s_3, s_2\bigr) \\
          &\qquad\cdot
    \napn\bigl(\gcfg{i}{E_2}(t, f, m), f\bigr)
    \napn\bigl(\gcfg{i}{E_3}(t, f, m), f\bigr) \\
    & \quad \mathord{}
      + \dapn\bigl(\gcfg{i}{E_1}(s_2, s_3, m_2), s_3, s_3\bigr) \dapn\bigl(\gcfg{i}{E_3}(t, f, m), t, t\bigr) \\
    &= \dapn\bigl(\gcfg{i}{E_1 \cqmk E_2 \ccol E_3}(t, f, m), f, f\bigr), \\
  \ppp{i}(E_1 \cqmk E_2 \ccol E_3)
    &= \ttp{i}(E_1) \ppp{i}(E_2)  \\
        & \quad \mathord{}
        + 2 \tfp{i}(E_1) \pp{i}(E_2) \pp{i}(E_3) + \ffp{i}(E_1) \pp{i}(E_3) \\
    &= \dapn\bigl(\gcfg{i}{E_1}(s_2, s_3, m_2), s_2, s_2\bigr) \dapn\bigl(\gcfg{i}{E_2}(t, t, m), t, t\bigr) \\
    & \quad \mathord{}
      + \dapn\bigl(\gcfg{i}{E_1}(s_2,s_3, m_2), s_2, s_3\bigr) \\
          &\qquad\cdot
    \apn\bigl(\gcfg{i}{E_2}(t, t, m)\bigr)
    \apn\bigl(\gcfg{i}{E_3}(t, t, m)\bigr) \\
    & \quad \mathord{}
      + \dapn\bigl(\gcfg{i}{E_1}(s_2,s_3, m_2), s_3, s_2\bigr) \\
          &\qquad\cdot
    \apn\bigl(\gcfg{i}{E_2}(t, t, m)\bigr)
    \apn\bigl(\gcfg{i}{E_3}(t, t, m)\bigr) \\
    & \quad \mathord{}
      + \dapn\bigl(\gcfg{i}{E_1}(s_2, s_3, m_2), s_3, s_3\bigr) \dapn\bigl(\gcfg{i}{E_3}(t, t, m), t, t\bigr) \\
    &= \dapn\bigl(\gcfg{i}{E_1 \cqmk E_2 \ccol E_3}(t, f, m), t, t\bigr). \qed
\end{align*}
\end{description}
\end{proof}

In the sequel, we will assume that any CFG $(N, A, s)$ constructed by
$\cfg{i}{}$ has $N \subset \Nset$; and that $B \in \Stm$ denotes a function
body and $\bodycfg{i}{B} = (N_B, A_B, s_B)$.
First we provide some additional notation.
\begin{itemize}
  \item
    $\orig(m) \in \Stm$ denotes a label or statement that generates
    the node $m$.  Note that, by Definition~\ref{def:reference-cfg},
    it follows that, if $m > n$, then $\orig(m)$ occurs before
    $\orig(n)$ in a function.
  \item
    For each $\stm \in \{\mathtt{switch}, \mathtt{while}, \mathtt{do-while}\}$
    in $B$, we insert two synthetic statements
    $\stm \; \mathtt{enter}$ (resp., $\stm \; \mathtt{exit}$) before (resp., after) the body of $\stm$.
    Therefore, if $n_1, n_2 \in N_B$ and $\orig(n_1) = \mathtt{\stm \; enter}$ and
    $\orig(n_2) = \mathtt{\stm \; exit}$, then $n_1 > n_2$.
  \item
    $\samebody(t_1, t_2, \stm)$,\footnote{$\samebody$ is a mnemonic for
    \emph{same body}.} where
    \(
    t_1, t_2 \in \Nset$ and $\stm \in \{\mathtt{switch}, \mathtt{while}, \mathtt{do-while}\},
    \)
    denotes a predicate that
    results in true if and only if $t_1$ and $t_2$ are the two nodes
    that enclose the body of a statement $\stm$, namely:
    \begin{multline}
      \label{eq:addnotate-sb}
      \samebody(t_1, t_2, \stm) \defeq
      \Bigl(
         \orig(t_1) = \stm \; \mathtt{enter}
         \land \orig(t_2) = \stm \; \mathtt{exit} \\
      \land \bigl|\bigl\{\, c \in [t_2,t_1] \bigm| \orig(c) = \stm \; \mathtt{enter} \,\bigr\}\bigr|
       = \bigl|\bigl\{\, c \in [t_2, t_1] \bigm| \orig(c) = \stm \; \mathtt{exit} \,\bigr\}\bigr|
      \Bigr).
    \end{multline}
\end{itemize}
Furthermore, for each label $L$ or statement $S$ in $B$,
where $\cfg{i}{L}(t, m) = (N, A, s)$
and $\cfg{i}{S}(t, \tb, \tc, m) = (N, A, s)$,
$t,\tb, \tc, m \in \Nset$ are such that $\tb \leq m \leq t \leq \tc$, and,
finally,
$N \sseq \{ t \} \union [m, m'-1]$:
\begin{itemize}
  \item
    $l(S)$ denotes the set of label identifiers contained in $S$:
    \begin{equation}
      \label{eq:addnotate-l}
      l(S) \defeq
      \bigl\{\, \id \in \Id \bigm| \exists m \in N \st \orig(m) = \id \,\bigr\}.
    \end{equation}
  \item
    $s(S)$ denotes the set of $\mathtt{case}$
    or $\mathtt{default}$ nodes contained in $S$
    where the controlling $\mathtt{switch}$ is external to $S$:
    \begin{equation}
      \label{eq:addnotate-s}
      s(S) \defeq
        \sset{
          n \in N
          }{
          \begin{aligned}
             &\orig(n) \in \{\mathtt{case}, \mathtt{default}\} \\
             & \land \nexists m_1, m_2 \in N \st \samebody(m_1, m_2, \mathrm{switch})
             \land n \in [m_1, m_2]
          \end{aligned}
          }.
    \end{equation}
  \item
    $t(S, B)$ denotes
    the set of $\mathtt{switch}$ nodes in $B$ that may pass the control to
    a $\mathtt{case}$ or $\mathtt{default}$ node in $S$:
    \begin{equation}
      \label{eq:addnotate-t}
      t(S, B)
      \defeq
        \sset{
          m \in N_B
        }{
          \begin{aligned}
            &\orig(m) = \mathtt{switch \; enter} \\
            &\land \exists m_2 \in N_B \st \samebody(m, m_2, \mathrm{switch}) \\
             &\land \exists n \in N \st n \in [m_2, m] \cap [t, s]
          \end{aligned}
          }.
    \end{equation}
  \item
    $b(s, t)$ denotes the set of $\mathtt{break}$ nodes between $s$
    and $t$ that are external to any inner \texttt{switch} or loop statements.
    Formally,
    if $\stm \in \{\mathtt{switch}, \mathtt{while}, \mathtt{do-while}\}$,
    \begin{equation}
      \label{eq:addnotate-b}
      b(s, t)
      \defeq
        \sset{
          m \in [s, t]
        }{
          \begin{aligned}
          &\orig(m) = \mathtt{break} \\
          &\land \nexists m_1, m_2 \in [s,t] \st
                  \bigl(\samebody(m_1, m_2, \stm)
          \land m \in [m_2, m_1] \bigr)
          \end{aligned}
        }.
    \end{equation}
  \item
    $c(s, t)$ denotes the set of $\mathtt{continue}$ nodes between $s$
    and $t$ that are external to any inner loop statements.
    Formally, if $\stm \in \{\mathtt{while}, \mathtt{do-while}\}$,
    \begin{equation}
      \label{eq:addnotate-c}
      c(s, t) \defeq
        \sset{
          m \in [s, t]
        }{
          \begin{aligned}
          &\orig(m) = \mathtt{continue} \\
          &\land \nexists m_1, m_2 \in [s,t] \st \bigl(\samebody(m_1, m_2, \stm)
          \land m \in [m_2, m_1] \bigr)
          \end{aligned}
        }.
    \end{equation}
  \item
    $r(s, t)$ denotes the set of $\mathtt{return}$ nodes between $s$ and $t$.
    Formally,
    \begin{equation}
      \label{eq:addnotate-r}
      r(s, t) \defeq \bigl\{\, m \in [s,t] \bigm| \orig(m) = \mathtt{return} \,\bigr\}.
    \end{equation}
  \item
    $\stapn(S, B)$ denotes the number of acyclic paths in $B$ to
    a node in $t(S, B)$:
    \begin{equation}
      \stapn(S, B)
      \defeq
      \sum\limits_{n \in t(S, B)}\napn\bigl(s, A, n\bigr).
    \end{equation}
  \item
    $g(\id, S, B)$ denotes the set of
    $\mathtt{goto} \; \id$ labeled statements in $S$:
    \begin{equation}
      \label{eq:addnotate-g}
      g(\id, S, B) \defeq \bigl\{\, g \in N \bigm| \orig(g) = \mathtt{goto} \; \id \,\bigr\}.
    \end{equation}
  \item
    $\pg(\id, S, B)$ denotes the set of
    $\mathtt{goto} \; \id$ labeled statements in $B$ and before $S$:
    \begin{equation}
      \pg(\id, S, B) \defeq \bigl\{\, g \in N_B, g > s \bigm| \orig(g) = \mathtt{goto} \; \id \,\bigr\}.
    \end{equation}
  \item
    $\gapn(\id, S, B)$ denotes the number of acyclic paths in $B$
    to nodes in $g(\id, S, B)$:
    \begin{equation}
      \gapn(\id, S, B)
      \defeq
      \sum\limits_{n \in g(\id, S, B)} \napn(s, A, n\bigr).
    \end{equation}
  \item
    $\pgapn(\id, S, B)$ denotes the number of acyclic paths in $B$
    to nodes in $\pg(\id, S, B)$:
    \begin{equation}
      \pgapn(\id, S, B)
      \defeq
      \sum\limits_{n \in \pg(\id, S, B)} \napn(s, A_B, n\bigr).
    \end{equation}
  \item
    $\ftapn(c, S, B)$, where $c \in N$, denotes the number of acyclic
    paths that ``fall through $c$ from above'':
    \begin{equation}
    \label{eq:ftapn}
      \ftapn(c, S, B)
      \defeq
      \begin{cases}
        \napn(s, A_1, c),
        &\text{if $\orig(c) = \id$,} \\
        \napn(s, A_2, c),
        &\text{if $\orig(c) \in \{\mathtt{case}, \mathtt{default}\}$,} \\
        \napn(s, A_f, c),
        &\text{otherwise,}
      \end{cases}
    \end{equation}
    where
    $A_1 = \bigl\{\, (c_1, c) \in A_B \bigm| c_1 \notin \pg(\id, S, B) \,\bigr\}$,
    $A_2 = \bigl\{\, (c_1, c) \in A_B \bigm| c_1 \notin t(S, B) \,\bigr\}$ and
    $A_f = \bigl\{\, (c_1, c) \in A_B \bigm| c_1 > c \,\bigr\}$.
  \item
    $\bapn(S, B)$ denotes the number of acyclic paths in $B$ to nodes in $b(s,t)$:
    \begin{equation}
      \bapn(S, B)
      \defeq
      \sum\limits_{n \in b(s, t)} \napn(s, A, n\bigr).
    \end{equation}
  \item
    $\capn(S)$ denotes the number of acyclic paths in $B$ to nodes in $c(s,t)$:
    \begin{equation}
      \capn(S)
      \defeq
      \sum\limits_{n \in c(s, t)} \napn(s, A, n\bigr).
    \end{equation}
  \item
    $\rapn(S)$ denotes the number of acyclic paths in $B$ to nodes in $r(s,t)$:
    \begin{equation}
      \rapn(S)
      \defeq
      \sum\limits_{n \in r(s, t)} \napn(s, A, n\bigr).
    \end{equation}
\end{itemize}

In Section~\ref{sec:ACPATH}, we gave an informal definition of a
controlled function body
(Definition~\ref{defn:controlled-function-body-informal}), we now
provide a formal version.

\begin{definition} \summary{(Controlled function body.)}
\label{def:controlled-function-body-formal}
Let $B \in \Stm$ be a full C function body.
We call $B$ a \emph{controlled function body} if it satisfies
the following properties:
\begin{enumerate}
\item
  \label{enum:cfb-backjump}
  for each $t_1, t_2 \in N$,
  if $\orig(t_1) = \id$ and $\orig(t_2) = \mathtt{goto} \; \id$,
  then $t_1 \leq t_2$;
\item
  \label{enum:cfb-intoloop}
  if $\stm \in \{\mathtt{while}, \mathtt{do-while}\}$,
  $w_1, w_2 \in N_B$ are such that
  $\orig(w_1) = \mathtt{\stm\; enter}$ and
  $\orig(w_2) = \mathtt{\stm \; exit}$,
  and one of the following holds:
  \begin{align*}
    \exists n \in N &\st n \in [w_2, w_1] \land \orig(n) \in \{\mathtt{break}, \mathtt{return}\} \land n\in b(w_1, w_2), \\
    \exists n_1,n_2 \in N &\st n_1 < w_2 \land n_2 \in [w_2, w_1] \land \orig(n_1) =\id \land \orig(n_2) = \mathtt{goto} \; \id,
  \end{align*}
  then neither of the following hold:
    \begin{align*}
      \exists t_1, t_2 \in N &\st t_1 \in [w_2, w_1] \\
      & \quad \land t_2 < w_2 \land \orig(t_2) \in \{\mathtt{case, default}\} \land \orig(t_1) = \mathtt{switch \; enter}, \\
      \exists t_1, t_2 \in N &\st t_1 \in [w_2 ,w_1] \\
      & \quad \land t_2 > w_1 \land \orig(t_1) =\id_w \land \orig(t_2) = \mathtt{goto} \; \id_w.
  \end{align*}
\end{enumerate}
\end{definition}

Note that it follows from condition~\ref{enum:cfb-backjump}
of Definition~\ref{def:controlled-function-body-formal}
that if $B \in \Stm$
is a controlled function body and $\bodycfg{i}{B} = (N_B, A_B, s_B)$,
for each $(m,n) \in A_B$ such that
$\orig(m) \notin \{\mathtt{while \; exit},  \mathtt{do-while \; exit}\}$,
the statement $\orig(m)$ will be before the statement $\orig(n)$ in $B$.

\begin{lemma}
\label{lm:controlled-function-body}
Let:
\begin{itemize}
  \item
  $i \in \{ 0, 1, 2 \}$ be an optimization level;
  \item
  $B \in \Stm$ a controlled function body;
  \item
  $\bodycfg{i}{B} = (N_B, A_B, s_B)$;
  \item
  $t \in N_B$ a target node in $\bodycfg{i}{B}$;
  \item
  $\stm \in \{\mathtt{while}, \mathtt{do-while} \}$;
  \item
  $w_1, w_2 \in N_B$ such that
  $\orig(w_1) = \stm \; \mathtt{enter}$ and
  $\orig(w_2) = \stm \; \mathtt{exit}$;
  \item
  $P = s_B, \ldots,  a, \ldots, t$
    an acylic path in $\bodycfg{i}{B}$ where $a \in [w_2, w_1]$.
\end{itemize}
  Then
  $P$ contains a subsequence $a_0, \ldots, a_j = a, \ldots a_k$
  where, for each $i \in [1,k-1]$, $a_i \in [w_2, w_1]$ and either
  $\orig(a_1) = \mathtt{\stm \; enter}$ or $\orig(a_{k-1}) = \mathtt{\stm \; exit}$.
\end{lemma}

\begin{proof}
  Suppose $a_0,a_1,\ldots, a_j=a, \ldots, a_k$ is the
  maximal subsequence of $P$ such that
  for each ${i \in [1,k-1]}$, $a_i \in [w_2,w_1]$.
  Suppose also that $\orig(a_1) \neq \mathtt{\stm \; enter}$ and
  $\orig(a_{k-1}) \neq \mathtt{\stm \; exit}$.
  Then we derive a contradiction.

  Since $\orig(a_1) \neq \mathtt{\stm \; enter}$,
  either $\orig(a_0) = \mathtt{goto \; \id_1}$ and $\orig(a_1) = \id_1$
  or $\orig(a_0) = \mathtt{switch \; enter}$
  and $\orig(a_1) \in \{\mathtt{case}, \mathtt{default}\}$;
  and also since $\orig(a_{k-1}) \neq \mathtt{\stm \; exit}$, either
  $\orig(a_k) = \id_k$ and $\orig(a_{k-1}) = \mathtt{goto} \; \id_k$
  or $\orig(a_k) \in \{\mathtt{return}, \mathtt{break}\}$.
  Therefore, condition~\ref{enum:cfb-intoloop} of
  Definition~\ref{def:controlled-function-body-formal}
  does not hold, contradicting
  the hypothesis that $B$ is a controlled function body.
  \qed
\end{proof}

\begin{lemma}
  \label{lm:label-paths}
  Let:
    $i \in \{ 0, 1, 2 \}$ be an optimization level;
    $B \in \Stm$ a controlled function body;
    $L \in \Lab$ a label in $B$;
    $t, m \in \Nset$ where $t < m$;
    $\bodycfg{i}{B} = (N_B, A_B, s_B)$;
    $\gcfg{i}{L}(t, m) = (N, A, s)$;
    $\Ft = \ftapn(s, L, B)$;
    $\St = \stapn(L, B)$
    and
    $\fund{\Gt}{\Id} {\Nset}$ be such that,
    for each $\id \in l(B)$, $\Gt(id) = \pgapn(\id, L, B)$.
  Then
  $\apc{i}{L}(\Ft, \St, \Gt) = \ftapn(t, L, B)$.
\end{lemma}

\begin{proof}
  We prove each kind of label separately.

\proofsec{Case label}
  if $L = \mathtt{case \; z}$ then, by~\eqref{eq:cfg_i-labels},
  $\gcfg{i}{L}(t, m) = \bigl(\{m, t\},\bigl\{(m, t)\bigr\}, m\bigr)$
  and, by~\eqref{eq:cfg_i-switchStat},
  $\exists c \in t(L, B) \st (c, m) \in A_B$.
  By hypothesis,
  \[
    \Ft = \napn\bigl(s_B, A_f \setminus \{(c, m)\}, m\bigr)
    \]
  and
  $\St = \napn\bigl(\bodycfg{i}{B}, c\bigr)$.
  The number of paths that reach the entry node of label $L$ without using an arc in $A_1$ is $\napn\bigl(s_B, A_f \setminus \bigl\{(c, m)\bigr\}, m\bigr)$.
  The number of paths passing through the arc $(c, m)$ is $\napn\bigl(\bodycfg{i}{B}, c, m\bigr)$.
  Each of these paths will reach $t$ via the arc $(m, t)$ and
  will not pass through any arcs in
  $\bigl\{\, (g, t) \bigm| g \in \pg(\id, L, B) \,\bigr\}$
  when $\orig(t) = \id$
  or through an arc $(c_1, t) \in t(L, B)$ if
  $\orig(t) \in \{\mathtt{case}, \mathtt{default}\}$.
  Concluding,
  \begin{align*}
     \apc{i}{L}(\Ft, \St, \Gt) &= \Ft + \St \\
           &= \napn\bigl(s_B, A_f \setminus \bigl\{(c, s)\bigr\} , s\bigr) + \napn\bigl(\bodycfg{i}{B}, c\bigr) \\
           &= \ftapn(t, L, B).
  \end{align*}
\proofsec{Default label}
   the proof is similar to the previous case.
\proofsec{Identifier label}
  by~\eqref{eq:cfg_i-labels},
  $\gcfg{i}{\id}(t, m) = \bigl(\{m, t\},\bigl\{(m, t)\bigr\}, m\bigr)$
  and, by hypothesis,
  \(
  \Ft
    =
    \napn\bigl(
      N_B,
      A_f \setminus \bigl\{\, (g, m) \bigm| g \in \pg(\id, L, B) \,\bigr\},
      m
    \bigl).
    \)
  The number of paths that reach the entry node of label $L$
  without using any of the arcs in
  \[
    A_g = \bigl\{\, (g, m) \bigm| g \in \pg(\id, L, B) \,\bigr\}\bigl)
  \]
  is
  \(
    \napn\bigl(
      s_B,
      A_f \setminus A_g,
      m
    \bigr)
  \).
  The number of paths passing through the arcs in $A_g$ is $\pgapn(\id, L, B)$.
  Each of these paths will reach $t$ via $(m, t)$ and
  will not pass through any arcs in
  \[
    \bigl\{\, (g, t) \bigm| g \in \pg(\id, L, B) \,\bigr\}
  \]
  if $\orig(t) = \id$ or through
  an arc $(c, t) \st c \in t(L, B)$
  if $\orig(t) \in \{\mathtt{case}, \mathtt{default}\}$.
  Hence, using the hypothesis:
  \begin{align*}
    \apc{i}{L}(\Ft, \St, \Gt)
      &= \Ft + \Gt(\id) \\
      &= \napn\bigl(N_B,
                    A_f
                      \setminus
                        \bigl\{\, (g, s) \bigm| g \in \pg(\id, L, B) \,\bigr\},
                    s\bigl) + \pgapn(\id, L, B)  \\
      &= \ftapn(t, L, B).
  \end{align*}
\qed
\end{proof}

\begin{lemma}
\label{lm:stat-paths}
Let:
$i \in \{ 0, 1, 2 \}$ be an optimization level;
$B \in \Stm$ be a controlled function body;
$S \in \Stm$ be a statement in $B$;
$t,\tb, \tc, c \in \Nset$ be such that $\tb \leq c < t \leq \tc$;
$\bodycfg{i}{B} = (N_B, A_B, s_B)$;
$\gcfg{i}{S}(t, \tb, \tc, m) = (N, A, s)$;
$\Ft = \ftapn(s, S, B)$;
$\St = \stapn(S, B)$
and
$\fund{\Gt}{\Id}{\Nset}$ be such that,
for each $\id \in l(B)$, $\Gt(id) = \pgapn(\id, S, B)$
Then
\[
  \apc{i}{S}(\Ft, \St, \Gt) = (\Ftout, \Bp, \Cp, \Rp, \Gtout),
\]
where
  \begin{align*}
    \Ftout &= \ftapn(t, S, B) \\
    \Bp &= \bapn(S, B), \\
    \Cp &= \capn(S, B), \\
    \Rp &= \rapn(S, B), \\
    \Gtout &= \Gt + \lambda \id \in l(B) \st \gapn(\id, S, B).
  \end{align*}
\end{lemma}

\begin{proof}
The proof is by induction on $S$ where each kind of statement
is considered separately.

\proofsec{Expression statement}
by~\eqref{eq:cfg_i-exprStat}, we have
$\gcfg{i}{E;}(t, \tb, \tc, m) = \gcfg{i}{E}(t, t, m)$,
and, by Lemma~\ref{lm:expr-count},
$\apn\bigl(\gcfg{i}{E}(t, t, m)\bigr) = \pp{i}(E)$.
Moreover, as there are no $\mathtt{break}$, $\mathtt{continue}$,
$\mathtt{return}$ or $\mathtt{goto}$ statements inside $S$,
$\bapn(S, B) = \capn(S, B) = \rapn(S, B) = 0$
and, for each $\id \in l(B)$, $\gapn(\id, S, B) = 0$.
Hence, by~\eqref{eq:apc_i-exprStat},
\begin{align*}
  \Ftout &= \Ft \pp{i}(E) \\
         &= \napn(s_B, A_f, s) \apn\bigl(\gcfg{i}{E}(t, t, m)\bigr) \\
         &= \ftapn(t, S, B), \\
     \Bp &= 0 \\
         &= \bapn(S, B), \\
     \Cp &= 0 \\
         &= \capn(S, B), \\
     \Rp &= 0 \\
         &= \rapn(S, B), \\
  \Gtout &= \Gt  \\
         &= \Gt + \lambda \id \in l(B) \st 0 \\
         &= \Gt + \lambda \id \in l(B) \st \gapn(\id, S, B).
  \end{align*}

\proofsec{Return statement}
by~\eqref{eq:cfg_i-returnStat}, $\gcfg{i}{\mathtt{return}}(t, \tb, \tc, m) = \bigl(\{m\}, \emptyset, m\bigl)$.
Hence, for each path in $(s_B, A_B, s)$ to $s = m$ there are $0$ paths to $t$ and $1$ path to the $\mathtt{return}$ node.
Moreover, as there are no $\mathtt{break}$, $\mathtt{continue}$ or $\mathtt{goto}$ statements inside $S$, $\bapn(S, B) = \capn(S, B) = 0$
and, for each $\id \in l(B)$, $\gapn(\id, S, B) = 0$.
Hence, by~\eqref{eq:apc_i-returnStat},
  \begin{align*}
    \Ftout &= 0 \\
             & = \napn\bigl((\{m\}, \emptyset, m), t\bigr) \\
             &= \ftapn(t, S, B), \\
    \Bp &= 0 \\
       &= \bapn(S), \\
    \Cp &= 0 \\
       &= \capn(S), \\
    \Rp &= \Ft \\
       &= \napn(s_B, A_B, s) \\
       &= \rapn(S), \\
    \Gtout &= \Gt  \\
              &= \Gt + \lambda \id \in l(B) \st 0 \\
              &= \Gt + \lambda \id \in l(B) \st \gapn(\id, S, B).
  \end{align*}

\proofsec{Return with expression statement}
  Let $\gcfg{i}{\mathtt{return} \; E}(t, \tb, \tc, m)$
  be defined as in~\eqref{eq:cfg_i-returnExprStat}.
  By Lemma~\ref{lm:expr-count},
  $\pp{i}(E) = \apn\bigl(\gcfg{i}{E}(m, m, m+1)\bigr)$.
  Then, for each path in $(s_B, A_B, s)$ to $s$, there are $0$ paths to $t$
  and $\pp{i}(E)$ paths to a return node.
  Moreover, as there are no $\mathtt{break}$, $\mathtt{continue}$ or
  $\mathtt{goto}$ statements inside $S$, $\bapn(S, B) = \capn(S, B) = 0$ and for each $\id \in l(B)$, $\gapn(\id, S, B) = 0$.
  Hence, by~\eqref{eq:apc_i-returnExprStat},
    Hence,
    \begin{align*}
      \Ftout &= 0 \\
               &= \ftapn(t, S, B), \\
      \Bp &= 0 \\
         &= \bapn(S), \\
      \Cp &= 0 \\
         &= \capn(S), \\
      \Rp &= \Ft \pp{i}(E) \\
         &= \napn(s_B, A_B, s) \napn\bigl(\gcfg{i}{E}(m, m, m+1), m\bigr) \\
         &= \rapn(S), \\
    \Gtout &= \Gt  \\
              &= \Gt + \lambda \id \in l(B) \st 0 \\
              &= \Gt + \lambda \id \in l(B) \st \gapn(\id, S, B).
    \end{align*}

\proofsec{Break statement} By~\eqref{eq:cfg_i-breakStat},
  $\gcfg{i}{\mathtt{break}}(t, \tb, \tc, m) = (\{m, \tb\},\bigl\{(m,\tb)\bigr\}, m)$.
  Then $s = m$ and, for each path in $(s_B, A_B, s)$ to $s$ there are $0$
  acyclic paths to $t$ and $1$ path to the $\mathtt{break}$ node.
  Moreover, as there are no $\mathtt{continue}$, $\mathtt{return}$ or
  $\mathtt{goto}$ statements inside $S$,
  $\capn(S, B) = \rapn(S, B) = 0$ and, for each $\id \in l(B)$,
  $\gapn(\id, S, B) = 0$.
  Hence, by~\eqref{eq:apc_i-breakStat},
  \begin{align*}
    \Ftout &= 0 \\
              &= \ftapn(t, S, B), \\
    \Bp &= \Ft \\
        &= \napn(s_B, A_B, s) \\
        &= \bapn(S), \\
    \Cp &= 0 \\
        &= \capn(S), \\
    \Rp &= 0 \\
        &= \rapn(S), \\
    \Gtout &= \Gt  \\
              &= \Gt + \lambda \id \in l(B) \st 0 \\
              &= \Gt + \lambda \id \in l(B) \st \gapn(\id, S, B).
  \end{align*}

\proofsec{Continue statement}
  By~\eqref{eq:cfg_i-continueStat},
  $\gcfg{i}{\mathtt{continue}}(t, \tb, \tc, m) = (\{m, \tb\},\bigl\{(m,\tc)\bigr\}, m)$.
  Then $s = m$ and, for each path in $(s_B, A_B, s)$ to $s$,
  there are $0$ paths to $t$ and $1$ path to the $\mathtt{continue}$ node.
  Moreover, as there are no $\mathtt{break}$, $\mathtt{return}$ or
  $\mathtt{goto}$ statements inside $S$,
  $\bapn(S, B) = \rapn(S, B) = 0$ and, for each $\id \in l(B)$,
  $\gapn(\id, S, B) = 0$.
  Hence, by~\eqref{eq:apc_i-continueStat},
  \begin{align*}
    \Ftout &= 0 \\
              &= \ftapn(t, S, B) \\
    \Bp &= 0 \\
        &= \bapn(S), \\
    \Cp &= \Ft \\
        &= \napn(s_B, A_B, s) \\
        &= \capn(S), \\
    \Rp &= 0 \\
        &= \rapn(S), \\
    \Gtout &= \Gt  \\
              &= \Gt + \lambda \id \in l(B) \st 0 \\
              &= \Gt + \lambda \id \in l(B) \st \gapn(\id, S, B).
  \end{align*}

\proofsec{Goto statement}
  Let $\gcfg{i}{\mathtt{goto} \; \id}(t, \tb, \tc, m)$
  be defined as in~\eqref{eq:cfg_i-gotoStat}.
  Then, for each path in $(s_B, A_B, s)$ to $s$, there are $0$ paths to $t$
  and $1$ path to $m$ where  $\orig(m) = \id$;
  hence $\gapn(\id, S, B) = \napn(s_B, A_B, s)$.
  Moreover, as there are no $\mathtt{break}$, $\mathtt{continue}$ or
  $\mathtt{return}$ statements inside $S$,
  $\bapn(S, B) = \capn(S, B) = \rapn(S, B) = 0$ and also
  for each $\id_1 \in l(B) \neq \id$, $\gapn(\id_1, S) = 0$.
  Hence, by~\eqref{eq:apc_i-gotoStat},
  \begin{align*}
    \Ftout &= 0 \\
             &= \ftapn(t, S, B), \\
    \Bp &= 0 \\
       &= \bapn(S) \Ft, \\
    \Cp &= 0 \\
       &= \capn(S) \Ft, \\
    \Rp &= 0 \\
       &= \rapn(S) \Ft, \\
    \Gtout &= \Gt\bigl[\bigl(\Gt(\id) + \Ft\bigr)/\id\bigr] \\
              &= \Gt + \lambda \id_1 \in l(B) \st
                    \begin{cases}
                      \Ft,    &\text{if $\id_1 = \id$,} \\
                      0, &\text{otherwise;}
                    \end{cases} \\
              &= \Gt + \lambda \id \in l(B) \st \gapn(\id, S, B).
  \end{align*}

\proofsec{Other statements}
By~\eqref{eq:cfg_i-otherStat},
$\gcfg{i}{\mathrm{\stm}}(t, \tb, \tc, m) = (\{t\},\emptyset, t)$ and $s = t$.
  Moreover, as there are no $\mathtt{break}$, $\mathtt{continue}$,
  $\mathtt{return}$ or
  $\mathtt{goto}$ statements inside $S$,
  $\bapn(S, B) = \capn(S, B) = \rapn(S, B) = 0$ and,
  for each $\id \in l(B)$, $\gapn(\id, S) = 0$.
  Hence, by~\eqref{eq:apc_i-otherStat},
  \begin{align*}
    \Ftout &= \Ft \\
              &= \ftapn(s, S, B) \\
              &= \ftapn(t, S, B), \\
    \Bp &= 0 \\
        &= \bapn(S), \\
    \Cp &= 0 \\
        &= \capn(S), \\
    \Rp &= 0 \\
        &= \rapn(S), \\
    \Gtout &= \Gt  \\
              &= \Gt + \lambda \id \in l(B) \st 0 \\
              &= \Gt + \lambda \id \in l(B) \st \gapn(\id, S, B).
  \end{align*}

\proofsec{Sequential composition}
  Let $\gcfg{i}{S_1 \; S_2}(t, \tb, \tc, m) = (N, A, s)$ be as defined in~\eqref{eq:cfg_i-seqStat}.
  Since $\Gt_1 = \Gt + \lambda \id \in l(B) \st \gapn(\id, S_1, B) = \lambda \id \in l(B) \st \pgapn(\id, S_2, B)$, using the inductive hypothesis on $S_1$, we can use the inductive hypothesis on $S_2$.
  Moreover, because we have $b(s_1, m_1) \cap b(s_2, m) = c(s_1, m_1) \cap c(s_2, m) = r(s_1, m_1) \cap r(s_2, m) = \emptyset$ the paths to $\mathtt{break}$, $\mathtt{continue}$ or $\mathtt{return}$ statements are, respectively,
  $\bapn(S_1 \; S_2, B) = \bapn(S_1) + \bapn(S_2)$,
  $\capn(S_1 \; S_2, B) = \capn(S_1) + \capn(S_2)$, $\rapn(S_1 \; S_2, B) = \rapn(S_1) + \rapn(S_2)$
  and, for each $\id \in l(B)$, since $g(\id, S_1, B) \cap g(\id, S_2, B) = \emptyset$, $\gapn(\id, S_1 \; S_2, B) = \gapn(\id, S_1, B) + \gapn(\id, S_2, B)$.
  Hence, by~\eqref{eq:apc_i-seqStat},
  \begin{align*}
    \Ftout &= \Ft_2 \\
              &= \ftapn(t, S_2, B) \\
              &= \ftapn(t, S_1 \; S_2, B), \\
    \Bp &= \Bp_1 + \Bp_2 \\
        &= \bapn(S_1) + \bapn(S_2) \\
        &= \bapn(S_1 \; S_2), \\
    \Cp &= \Cp_1 + \Cp_2 \\
        &= \capn(S_1) \Ft + \capn(S_2) \Ft_1 \\
        &= \capn(S_1 \; S_2), \\
    \Bp &= \Bp_1 + \Bp_2 \\
        &= \rapn(S_1) \Ft + \rapn(S_2) \Ft_1 \\
        &= \rapn(S_1 \; S_2), \\
    \Gtout &= \Gt_2 \\
              &= \Gt_1 + \lambda \id \in l(B) \st \gapn(\id, S_2, B) \\
              &= \Gt + \lambda \id \in l(B) \st \gapn(\id, S_1, B) + \lambda \id \in l(B) \st \gapn(\id, S_2, B) \\
              &= \Gt + \lambda \id \in l(B) \st \gapn(\id, S_1 \; S_2, B).
  \end{align*}

\proofsec{Conditional statement}
There are three cases:

\begin{description}
\item[$\tv{i}(E) = \true \land \Ms_2 = \emptyset$:]
  then $\tp{i}(E) = 1$ and $\fp{i}(E) = 0$.
  By~\eqref{eq:cfg_i-condStat},
  \[
    \gcfg{i}{\mathtt{if} \; (E) \; S_1 \; \mathtt{else} \; S_2}(t, \tb, \tc, m) = \gcfg{i}{S_1}(t, \tb, \tc, m).
  \]
  Hence $\fp{i}(E) \Ft = 0$ so that $\Ft_2 = 0$.
  Since $S_2$ does not contain any labeled statements,
  it is not possible to jump into $S_2$ so that $\Bp_2 = \Cp_2 = \Rp_2 = 0$
  and, also for each $\id \in l(B)$, $\gapn(\id, S_2, B) = 0$.
  Hence, by~\eqref{eq:apc_i-condStat} and the inductive hypothesis on $S_1$,
   \begin{align*}
    \Ftout &= \Ft_1 + \Ft_2 \\
         &= \Ft_1 \\
         &= \ftapn(t, S_1, B), \\
    \Bp &= \Bp_1 + \Bp_2 \\
       &= \Bp_1 \\
       &= \bapn(S_1) \\
       &= \bapn\bigl(\mathtt{if} \; (E) \; S_1 \; \mathtt{else} \; S_2\bigr), \\
    \Cp &= \Cp_1 + \Cp_2 \\
       &= \Cp_1 \\
       &= \capn(S_1) \tp{i}(E) \\
       &= \capn\bigl(\mathtt{if} \; (E) \; S_1 \; \mathtt{else} \; S_2\bigr), \\
    \Rp &= \Rp_1 + \Rp_2 \\
       &= \Rp_1 \\
       &= \rapn(S_1) \tp{i}(E) \\
       &= \rapn\bigl(\mathtt{if} \; (E) \; S_1 \; \mathtt{else} \; S_2\bigr), \\
    \Gtout &= \Gt_2 \\
          &= \Gt_1 + \lambda \id \in l(B) \st \gapn(\id, S_2, B) \\
          &= \Gt_1 + \lambda \id \in l(B) \st 0 \\
          &= \Gt_1 \\
          &= \Gt + \lambda \id \in l(B) \st \gapn(\id, S_1, B) \\
          &= \Gt + \lambda \id \in l(B) \st \gapn(\id, S, B).
  \end{align*}
\item[$\tv{i}(E) = \false \land \Ms_1 = \emptyset$:]
  the proof is similar to the previous case;
\item[otherwise:]
  Let $\gcfg{i}{\mathtt{if} \; (E) \; S_1 \; \mathtt{else} \; S_2}(t, \tb, \tc, m)$ be as defined in~\eqref{eq:cfg_i-condStat}.
  For each path to $m$ or $m_1$ there is $1$ path to $t$, using the arcs $(m, t)$ or $(m_1, t)$.
  Furthermore, since $\tp{i}(E) \Ft = \ftapn(s_1, S_1, B)$, we can use the inductive hypothesis on $S_1$, and since
  $\fp{i}(E) \Ft = \ftapn(s_2, S_2, B)$ and
  $\Gt_1 = \Gt + \lambda \id \in l(B) \st \gapn(\id, S_1, B) = \lambda \id \in l(B) \st \pgapn(\id, S_2, B)$ we can also use the inductive hypothesis on $S_2$.
  Moreover, as $b(s_1, m_1) \cap b(s_2, m) = c(s_1, m_1) \cap c(s_2, m) = r(s_1, m_1) \cap r(s_2, m) = \emptyset$ the paths to $\mathtt{break}$, $\mathtt{continue}$ or $\mathtt{return}$ statements are
  $\bapn(S_1 \; S_2, B) = \bapn(S_1) + \bapn(S_2)$,
  $\capn(S_1 \; S_2, B) = \capn(S_1) + \capn(S_2)$, $\rapn(S_1 \; S_2, B) = \rapn(S_1) + \rapn(S_2)$
  and for each $\id \in l(B)$, since $g(\id, S_1, B) \cap g(\id, S_2, B) = \emptyset$, $\gapn(\id, S_1 \; S_2) = \gapn(\id, S_1, B) + \gapn(\id, S_2, B)$.
  Hence, by~\eqref{eq:apc_i-condStat} and the inductive hypothesis on $S_1$,
  \begin{align*}
    \Ftout &= \Ft_1 + \Ft_2 \\
         &= \ftapn(m, S_1, B) + \ftapn(m_1, S_2, B) \\
         &= \napn\bigl(N_B, A_f, m\bigr) +  \napn\bigl(N_B, A_f, m_1\bigr) \\
         &= \ftapn(t, B), \\
    \Bp &= \Bp_1 + \Bp_2 \\
        &= \bapn(S_1) + \bapn(S_2) \\
        &= \bapn\bigl(\mathtt{if} \; (E) \; S_1 \; \mathtt{else} \; S_2\bigr), \\
    \Cp &= \Cp_1 + \Cp_2 \\
        &= \capn(S_1) + \capn(S_2) \\
        &= \capn\bigl(\mathtt{if} \; (E) \; S_1 \; \mathtt{else} \; S_2\bigr), \\
    \Bp &= \Bp_1 + \Bp_2 \\
        &= \rapn(S_1) + \rapn(S_2) \\
        &= \rapn\bigl(\mathtt{if} \; (E) \; S_1 \; \mathtt{else} \; S_2\bigr), \\
    \Gtout &= \Gt_2 \\
              &= \Gt_1 + \lambda \id \in l(B) \st \gapn(\id, S_2, B) \\
              &= \Gt + \lambda \id \in l(B) \st \gapn(\id, S_1, B) + \lambda \id \in l(B) \st \gapn(\id, S_2, B) \\
              &= \Gt + \lambda \id \in l(B) \st \gapn\bigl(\id, \mathtt{if} \; (E) \; S_1 \; \mathtt{else} \; S_2, B\bigr).
  \end{align*}
\end{description}

\proofsec{One-armed conditional statement}
There are three cases:

\begin{description}
\item[$\tv{i}(E) = \true$:]
  then $\tp{i}(E) = 1$ and $\fp{i}(E) = 0$.
  By~\eqref{eq:cfg_i-one-armedCondStat},
  \[
    \gcfg{i}{\mathtt{if} \; (E) \; S_1}(t, \tb, \tc, m) = \gcfg{i}{S_1}(t, \tb, \tc, m).
   \]
   For each path in $(N_B, A_B, m)$ to $m$ there is $1$ path to $t$ using the arc $(m, t)$.
  Hence, by~\eqref{eq:apc_i-one-armedCondStat} and the inductive hypothesis on $S_1$,
   \begin{align*}
    \Ftout &= \Ft_1 + \fp{i}(E) \Ft \\
         &= \Ft_1 \\
         &= \ftapn(t, S, B), \\
    \Bp &= \Bp_1 \\
       &= \bapn(S_1) \\
       &= \bapn\bigl(\mathtt{if} \; (E) \; S_1\bigr), \\
    \Cp &= \Cp_1 \\
       &= \capn(S_1) \\
       &= \capn\bigl(\mathtt{if} \; (E) \; S_1\bigr), \\
    \Rp &= \Rp_1 \\
       &= \rapn(S_1) \\
       &= \rapn\bigl(\mathtt{if} \; (E) \; S_1\bigr), \\
    \Gtout &= \Gt_1 \\
          &= \Gt + \lambda \id \in l(B) \st \gapn(\id, S_1, B) \\
          &= \Gt + \lambda \id \in l(B) \st \gapn(\id, S, B).
  \end{align*}
 \item[$\tv{i}(E) = \false \land \Ms_1 = \emptyset$:]
   Then $\tp{i}(E) = 0$ and $\fp{i}(E) = 1$.
   By~\eqref{eq:cfg_i-one-armedCondStat},
   \[
     \gcfg{i}{\mathtt{if} \; (E) \; S_1}(t, \tb, \tc, m) = (\{ t \}, \emptyset, t).
   \]
   Hence $\Ft_1 = 0$.
  Since $S_1$ does not contain any labeled statements,
  it is not possible to jump into $S_1$ so that $\Bp_1 = \Cp_1 = \Rp_1 = 0$
  and, also, for each $\id \in l(B)$, $\gapn(\id, S_1, B) = 0$.
  Hence, by~\eqref{eq:apc_i-one-armedCondStat},
  \begin{align*}
    \Ftout &= \Ft_1 + \fp{i}(E) \Ft \\
         &= \Ft \\
         &= \ftapn(t, S, B), \\
    \Bp &= \Bp_1 \\
       &= \bapn(S_1) \\
       &= 0 \\
       &= \bapn\bigl(\mathtt{if} \; (E) \; S_1\bigr), \\
    \Cp &= \Cp_1 \\
       &= \capn(S_1) \\
       &= 0 \\
       &= \capn\bigl(\mathtt{if} \; (E) \; S_1\bigr), \\
    \Rp &= \Rp_1 \\
       &= \rapn(S_1) \\
       &= 0 \\
       &= \rapn\bigl(\mathtt{if} \; (E) \; S_1\bigr), \\
    \Gtout &= \Gt_1 \\
          &= \Gt + \lambda \id \in l(B) \st \gapn(\id, S_1, B) \\
          &= \Gt + \lambda \id \in l(B) \st 0 \\
          &= \Gt + \lambda \id \in l(B) \st \gapn(\id, S, B).
  \end{align*}
\item[otherwise:]
  Let $\gcfg{i}{\mathtt{if} \; (E) \; S_1}(t, \tb, \tc, m)$ be as defined in~\eqref{eq:cfg_i-one-armedCondStat}.
  By Lemma~\ref{lm:expr-count},
  $\tp{i}(E) = \napn\bigl(\gcfg{i}{E}(s_1, t, m), s_1\bigr)$ and
  $\fp{i}(E) = \napn\bigl(\gcfg{i}{E}(s_1, t, m), t\bigr)$.
  Moreover, for each path from $s$ to $s_1$ in $(\gcfg{i}{E}(s_1, t, m)$,
  there are $\napn\bigl(\gcfg{i}{S_1}(m, \tb, \tc, m+1), m\bigr)$ paths to $t$.
  Since $\tp{i}(E) \Ft = \ftapn(s_1, B)$,
  we can use the inductive hypothesis on $S_1$.
  Hence, by~\eqref{eq:apc_i-one-armedCondStat},
  \begin{align*}
    \Ftout &= \Ft_1 + \fp{i}(E) \Ft \\
         &= \ftapn(m, S_1, B) + \fp{i}(E) \ftapn(s, B) \\
         &= \ftapn(t, S, B), \\
    \Bp &= \Bp_1 \\
        &= \bapn(S_1) \\
        &= \bapn\bigl(\mathtt{if} \; (E) \; S_1\bigr), \\
    \Cp &= \Cp_1 \\
        &= \capn(S_1) \\
        &= \capn\bigl(\mathtt{if} \; (E) \; S_1\bigr), \\
    \Rp &= \Rp_1 \\
        &= \rapn(S_1) \\
        &= \rapn\bigl(\mathtt{if} \; (E) \; S_1\bigr), \\
    \Gtout &= \Gt_1 \\
              &= \Gt + \lambda \id \in l(B) \st \gapn(\id, S_1, B) \\
              &= \Gt + \lambda \id \in l(B) \st \gapn\bigl(\id, \mathtt{if} \; (E) \; S_1, B\bigr).
  \end{align*}
\end{description}

\proofsec{Switch statement}
There are two cases:
\begin{description}
\item[$(\df, n) \in \Ms_1$:]
  Let $\gcfg{i}{\mathtt{switch} \; (E) \; S_1}(t, \tb, \tc, m)$ be as defined in~\eqref{eq:cfg_i-switchStat}.
  By Lemma~\ref{lm:expr-count}, there are $\pp{i}(E)$ paths from $s$ to $m_1$.
  Moreover, since arc $(m_1, s_S) \in A$ only if $\orig(s_S) \in \{\mathtt{case, default}\}$, $\ftapn(s_S, B) = 0$ so that $m_1 \in t(S, B)$.
  In addition to the $\ftapn(m, B)$ paths that fall through to $m$, there are $\bapn(S_1)$ paths that exit from $\mathtt{switch}$ via $\mathtt{break}$ nodes.
  Furthermore, by Definition~\eqref{eq:addnotate-b},
  $b(s, t) = \emptyset$.
  Also $\pp{i}(E) \Ft = \napn(N_B, A_f, m_1) = \stapn(S_1, B)$.
  Hence, by~\eqref{eq:apc_i-switchStat}
  and applying the inductive hypothesis to $S_1$,
  \begin{align*}
    \Ftout &= \Ft_S + \Bp_S \\
         &= \ftapn(m, S, B) + \bapn(S_1, B) \\
         &= \ftapn(t, S, B), \\
    \Bp &= 0 \\
       &= \bapn\bigl(\mathtt{switch} \; (E) \; S_1\bigr), \\
    \Cp &= \Cp_S \\
       &= \capn(S_1) \\
       &= \capn\bigl(\mathtt{switch} \; (E) \; S_1\bigr), \\
    \Rp &= \Rp_S \\
       &= \rapn(S_1) \\
       &= \rapn\bigl(\mathtt{switch} \; (E) \; S_1\bigr), \\
    \Gtout &= \Gt_S \\
          &= \Gt + \lambda \id \in l(B) \st \gapn(\id, S_1, B) \\
          &= \Gt + \lambda \id \in l(B) \st \gapn\bigl(\id, \mathtt{switch} \; (E) \; S_1, B\bigr).
  \end{align*}
\item[$(\df, n) \notin \Ms_1$:]
  Let $\gcfg{i}{\mathtt{switch} \; (E) \; S_1}(t, \tb, \tc, m)$ be as defined in~\eqref{eq:cfg_i-switchStat}.
  Then there is $1$ path from $m_1$ to $m$ using the arc $(m_1, m)$.
  By Lemma~\ref{lm:expr-count}, there are $\pp{i}(E)$ paths to $m_1$.
  Moreover, since arc $(m_1, s_S) \in A$ only if $\orig(s_S) \in \{\mathtt{case, default}\}$, $\ftapn(s_S, B) = 0$.
  In addition to the $\ftapn(m, B)$ paths that fall through to $m$, there are $\bapn(S_1)$ nodes that exit from $\mathtt{switch}$ via $\mathtt{break}$ nodes and $\pp{i}(E) \Ft$ paths that reach $m$ via $(m_1, m)$.
  Furthermore, by Definition~\eqref{eq:addnotate-b},
  $b(s, t) = \emptyset$.
  Also, since $0 = \ftapn(s_S, B)$, $\pp{i}(E) \Ft = \napn(N_B, A_f, m_1) = \stapn(S_1, B)$.
  Hence, by~\eqref{eq:apc_i-switchStat}
  and applying the inductive hypothesis to $S_1$,
  \begin{align*}
    \Ftout &= \Ft_S + \Bp_S + \pp{i}(E) \Ft \\
         &= \ftapn(m, S, B) + \bapn(S_1) + \pp{i}(E) \Ft \\
         &= \ftapn(t, S, B), \\
    \Bp &= 0 \\
       &= \bapn\bigl(\mathtt{switch} \; (E) \; S_1\bigr), \\
    \Cp &= \Cp_1 \\
       &= \capn(S_1) \\
       &= \capn\bigl(\mathtt{switch} \; (E) \; S_1\bigr), \\
    \Rp &= \Rp_1 \\
       &= \rapn(S_1) \\
       &= \rapn\bigl(\mathtt{switch} \; (E) \; S_1\bigr), \\
    \Gtout &= \Gt_1 \\
          &= \Gt + \lambda \id \in l(B) \st \gapn(\id, S_1, B) \\
          &= \Gt + \lambda \id \in l(B) \st \gapn\bigl(\id, \mathtt{switch} \; (E) \; S_1, B\bigr).
  \end{align*}
\end{description}

\proofsec{While statement}
  let $\gcfg{i}{\mathtt{while} \; (E) \; S_1}(t, \tb, \tc, m) = (N, A, s)$ be as defined in~\eqref{eq:cfg_i-whileStat}.
  Then $\orig(m_3) = \mathtt{while \; enter}$ and $orig(m) = \mathtt{while \; exit}$
  so that $b(s, t) = c(s, t) = \emptyset$.
  As $B$ is a controlled function body, by
  Lemma~\ref{lm:controlled-function-body},
  each acyclic path in $(N, A, s)$ to $t$ will
  include the node $m_3$ or $m$.
  Therefore the number of paths to $t$ is the sum of:
  \begin{itemize}
  \item
    paths that go to $t$ directly from $E$ evaluating~$\false$;
  \item
    paths that fall to an entry node for $\cfg(i)(E)(t,f,m)$,
    $E$ evaluating~$\true$, and then exit from $S_1$ via a $\mathtt{break}$ node;
  \item
    paths that fall to $S$ to an entry node for $\cfg(i)(E)(t,f,m)$,
    $E$ evaluating~$\true$ and exit via $m$,
    (possibly via a $\mathtt{continue}$) node) and then
    $E$ evaluating~$\false$;
  \item
    paths that jump into $S$ from a $\mathtt{goto \; \id}$
    or $\mathtt{switch}$ node, exit via $m$,
    (possibly via a $\mathtt{continue}$ node) and then
    $E$ evaluating~$\false$.
  \end{itemize}
  Since $\tp{i}(E) \Ft = \ftapn(m_1, B)$,
  we can use the inductive hypothesis on $S_1$.
  Hence, by~\eqref{eq:apc_i-whileStat},
  letting $T \defeq \tfp{i}(E) / \tp{i}(E)$ if
  $\tp{i}(E) \neq 0$ and $T = 0$ otherwise:
  \begin{align*}
    \Ftout &= \fp{i}(E) \Ft + \Bp_S\tp{i}(E) + (\Ft_S + \Cp_S) \tfp{i}(E) \\
         &= \fp{i}(E) \Ft + \bapn(S_1)\tp{i}(E) + \bigr(\ftapn(m_1, S, B) + \capn(S_1)\bigl)\tfp{i}(E) \\
         &= \fp{i}(E) \Ft + \bapn(S_1)\tp{i}(E) \\
         &\quad \mathord{}
         + \bigr(\ftapn(m_1, S_1, B) + \Ft \sum\limits_{n \in c(s, t)} \napn(m_1, A, n) \bigr) \tfp{i}(E) \\
         &= \fp{i}(E) \Ft + \bapn(S_1)\tp{i}(E) \\
         & \quad \mathord{}
          + \bigr(\ftapn(m_1, S_1, B)
          + \Ft \sum\limits_{n \in c(s, t)} \napn(m_1, A, n\bigr) \bigl) \tfp{i}(E) \\
         &= \ftapn(t, S, B), \\
    \Bp &= 0 \\
       &= \bapn\bigl(\mathtt{while} \; (E) \; S_1\bigr), \\
    \Cp &= 0 \\
       &= \capn\bigl(\mathtt{while} \; (E) \; S_1\bigr), \\
    \Rp &= \Rp_1 \\
       &= \rapn(S_1) \\
       &= \rapn\bigl(\mathtt{while} \; (E) \; S_1\bigr), \\
    \Gtout &= \Gt_1 \\
           &= \Gt + \lambda \id \in l(B) \st \gapn(\id, S_1, B) \\
           &= \Gt + \lambda \id \in l(B) \st \gapn\bigl(\id, \mathtt{while} \; (E) \; S_1\bigr).
  \end{align*}

\proofsec{Do while statement}
  let $\gcfg{i}{\mathtt{do} \; S_1 \; \mathtt{while} \; (E)}(t, \tb, \tc, m)$
  be as defined in~\eqref{eq:cfg_i-doWhileStat}.
  Then $\orig(m_1) = \mathtt{do-while \; enter}$ and
  $\orig(m) = \mathtt{do-while \; exit}$ so that $b(s, t) = c(s, t) = \emptyset$.
  As $B$ is a controlled function body, by
  Lemma~\ref{lm:controlled-function-body},
  each acyclic path in $(N, A, s)$ to $t$ will
  include the nodes $m$ or $m_1$.
  Therefore the only paths to $t$ through $S$ are:
  \begin{itemize}
  \item
    those that fall to $m_1$ and go to $t$ via a $\mathtt{break}$ node; and
  \item
    those that fall through to $m$,
    (possibly via a $\mathtt{continue}$ node) and then $E$ evaluating~$\false$.
  \end{itemize}
  Since $\Ft = \ftapn(m_1, S, B)$, we can use the inductive hypothesis on $S_1$.
  Hence, by~\eqref{eq:apc_i-doWhileStat},
  \begin{align*}
    \Ftout &= \fp{i}(E) \Ft_S + \Bp_S \\
         &= \ftapn(m, S, B) \fp{i}(E) + \bapn(S_1, B) \\
         &= \ftapn(t, S, B), \\
    \Bp &= 0 \\
       &= \bapn\bigl(\mathtt{do} \; S_1 \; \mathtt{while} \; (E)\bigr), \\
    \Cp &= 0 \\
       &= \capn\bigl(\mathtt{do} \; S_1 \; \mathtt{while} \; (E)\bigr), \\
    \Rp &= \Rp_1 \\
       &= \rapn(S_1) \\
       &= \rapn\bigl(\mathtt{do} \; S_1 \; \mathtt{while} \; (E)\bigr), \\
    \Gtout &= \Gt_1 \\
           &= \Gt + \lambda \id \in l(B) \st \gapn(\id, S_1, B) \\
           &= \Gt + \lambda \id \in l(B) \st \gapn\bigl(\id, \mathtt{do} \; S_1 \; \mathtt{while} \; (E)\bigr).
  \end{align*}
 
\proofsec{For statement}
the lemma is true by inductive hypothesis for sequential composition and while statement.

\proofsec{Labeled statement}
let $\gcfg{i}{L \ccol S_1}(t, \tb, \tc, m) = (N, A, s)$ be as defined in~\eqref{eq:cfg_i-labeledStat}.
By Lemma~\ref{lm:label-paths},
$\Ft_L = \ftapn(s_S, S, B)$ so that, by~\eqref{eq:apc_i-labelStat},
lemma is true by the inductive hypothesis on $S_1$.

\proofsec{Compound statement}
by~\eqref{eq:cfg_i-compundStat},
$\bcfg{i}{\{S_1\}}(t, \tb, \tc, m) = \cfg{i}{S_1}(t, \tb, \tc, m)$ so that,
by~\eqref{eq:apc_i-compStat},
the lemma is true by inductive hypothesis on $S_1$.
\qed
\end{proof}

\begin{corollary}
\label{cor:body-paths}
Let $B \in \Stm$ be a full C controlled function body,
if $\apc{i}{B}(1, 0, \lambda \id \in l(B) \st0) = (\Ftout, \Bp, \Cp, \Rp, \Gtout)$ then
$\apn(\bodycfg{i}{B}) = \Ftout + \Rp$
\end{corollary}

\begin{proof}
Let $\bodycfg{i}{B} = (N, A, s)$;
then $\ftapn(s, B, B) = \napn(s, A, s) = 1$,
$t(B, B) = \emptyset$, and, for all $\id \in l(B)$, $\pgapn(\id, B, B) = 0$.
Concluding, by Lemma~\ref{lm:stat-paths},
$\Ftout = \ftapn(0, B, B) = \apn(N, A, 0)$
as $\orig(0) \notin \{\mathtt{case, default, goto \; \id}\}$
and $\Rp = \rapn(B,B)$.
Then the number of paths leading to an exit node are $\Ftout + \Rp$.
\qed
\end{proof}

\apniscorrect*
\begin{proof}
  By Definition~\ref{def:apc_i-Fullbody} and Corollary~\ref{cor:body-paths}:
  \begin{align*}
    \bodyapc{i}{B}
    &= \Ftout + \Rp \\
    &= \apn\bigl(\bodycfg{i}{B}\bigr).
  \end{align*}
  \qed
\end{proof}

\section{Example Reference CFGs}
\label{app:example-reference-cfgs}

We now illustrate the CFGs built according to
Definition~\ref{def:reference-cfg} by means of examples.
In the examples, all expressions represented by `$\cdot$'
are assumed to be non-constants and to have a trivial control flow.
Similarly, all statements represented by `$\mathord{-}$'
are assumed to result into a basic block, i.e., a single node in the CFG.
Moreover, the CFGs have been simplified by removing all nodes and arcs
that are unreachable from the entry node, which is represented by
a diamond-shaped box.  Exit nodes are emphasized by being enclosed
into double circles.
All the drawings have been obtained automatically from an executable
version of Definition~\ref{def:reference-cfg}.
\begin{figure}
  \centering

\begin{tikzpicture}[>=latex,line join=bevel,]
\node (11) at (163.0bp,242.0bp) [draw,ellipse] {11};
  \node (10) at (171.0bp,170.0bp) [draw,ellipse] {10};
  \node (13) at (97.0bp,386.0bp) [draw,diamond] {13};
  \node (12) at (125.0bp,314.0bp) [draw,ellipse] {12};
  \node (3) at (27.0bp,98.0bp) [draw,ellipse] {3};
  \node (2) at (63.0bp,22.0bp) [draw,circle, double] {2};
  \node (5) at (99.0bp,170.0bp) [draw,ellipse] {5};
  \node (4) at (99.0bp,98.0bp) [draw,ellipse] {4};
  \node (7) at (207.0bp,22.0bp) [draw,circle, double] {7};
  \node (9) at (243.0bp,98.0bp) [draw,ellipse] {9};
  \node (8) at (171.0bp,98.0bp) [draw,ellipse] {8};
  \draw [->] (11) ..controls (165.86bp,215.98bp) and (166.92bp,206.71bp)  .. (10);
  \draw [->] (10) ..controls (195.75bp,144.94bp) and (209.52bp,131.55bp)  .. (9);
  \draw [->] (10) ..controls (171.0bp,143.98bp) and (171.0bp,134.71bp)  .. (8);
  \draw [->] (13) ..controls (105.84bp,362.89bp) and (110.42bp,351.45bp)  .. (12);
  \draw [->] (11) ..controls (140.69bp,216.6bp) and (129.17bp,203.99bp)  .. (5);
  \draw [->] (5) ..controls (99.0bp,143.98bp) and (99.0bp,134.71bp)  .. (4);
  \draw [->] (12) ..controls (117.43bp,271.67bp) and (109.29bp,227.21bp)  .. (5);
  \draw [->] (13) ..controls (92.211bp,359.21bp) and (89.961bp,344.81bp)  .. (89.0bp,332.0bp) .. controls (85.46bp,284.83bp) and (90.945bp,229.82bp)  .. (5);
  \draw [->] (12) ..controls (138.42bp,288.28bp) and (144.16bp,277.71bp)  .. (11);
  \draw [->] (3) ..controls (39.197bp,71.929bp) and (44.473bp,61.083bp)  .. (2);
  \draw [->] (9) ..controls (230.8bp,71.929bp) and (225.53bp,61.083bp)  .. (7);
  \draw [->] (5) ..controls (74.25bp,144.94bp) and (60.476bp,131.55bp)  .. (3);
  \draw [->] (8) ..controls (183.2bp,71.929bp) and (188.47bp,61.083bp)  .. (7);
  \draw [->] (4) ..controls (86.803bp,71.929bp) and (81.527bp,61.083bp)  .. (2);
\end{tikzpicture}
\caption{CFG for $\bigl\{\, \mathtt{if} \; (\cdot \cand \cdot \cand \cdot) \; \mathtt{return} \; (\cdot \cqmk 0 \ccol 1); \; \mathtt{else} \; \mathtt{return} \; (\cdot \cqmk 0 \ccol 1); \,\bigr\} $: $8$~acyclic paths}
\label{fig:cfg4}

\end{figure}
Figure~\ref{fig:cfg4} shows the CFG generated by a command containing
\texttt{return} statements and branching expressions.
\begin{figure}
  \centering

\begin{tikzpicture}[>=latex,line join=bevel,]
\node (0) at (22.0bp,22.0bp) [draw,circle, double] {0};
  \node (4) at (77.0bp,386.0bp) [draw,ellipse] {4};
  \node (7) at (50.0bp,170.0bp) [draw,ellipse] {7};
  \node (6) at (50.0bp,98.0bp) [draw,ellipse] {6};
  \node (9) at (50.0bp,314.0bp) [draw,diamond] {9};
  \node (8) at (50.0bp,242.0bp) [draw,ellipse] {8};
  \draw [->] (7) ..controls (50.0bp,143.98bp) and (50.0bp,134.71bp)  .. (6);
  \draw [->] (6) ..controls (40.533bp,71.98bp) and (36.691bp,61.827bp)  .. (0);
  \draw [->] (9) ..controls (31.04bp,292.23bp) and (19.035bp,276.37bp)  .. (14.0bp,260.0bp) .. controls (-7.9472bp,188.63bp) and (5.2197bp,99.426bp)  .. (0);
  \draw [->] (9) ..controls (50.0bp,287.98bp) and (50.0bp,278.71bp)  .. (8);
  \draw [->] (4) ..controls (67.129bp,359.41bp) and (62.781bp,348.14bp)  .. (9);
  \draw [->] (8) ..controls (50.0bp,215.98bp) and (50.0bp,206.71bp)  .. (7);
  \draw [->] (7) ..controls (71.961bp,195.9bp) and (81.619bp,209.75bp)  .. (86.0bp,224.0bp) .. controls (99.972bp,269.44bp) and (91.449bp,325.43bp)  .. (4);
\end{tikzpicture}
\caption{CFG for $\bigl\{\, \mathtt{while} \; (\cdot) \; \mathtt{if} \; (\cdot) \; \mathtt{break}; \; \mathtt{else} \; \mathtt{continue}; \,\bigr\} $: $3$~acyclic paths}
\label{fig:cfg15}

\end{figure}
Figure~\ref{fig:cfg15} shows the effect of $\mathtt{break}$
and $\mathtt{continue}$ in a $\mathtt{while}$ statement.
\begin{figure}
  \centering

\begin{tikzpicture}[>=latex,line join=bevel,]
\node (1) at (116.0bp,386.0bp) [draw,ellipse] {1};
  \node (0) at (22.0bp,22.0bp) [draw,circle, double] {0};
  \node (3) at (110.0bp,22.0bp) [draw,ellipse] {3};
  \node (5) at (74.0bp,170.0bp) [draw,ellipse] {5};
  \node (4) at (55.0bp,98.0bp) [draw,ellipse] {4};
  \node (7) at (84.0bp,314.0bp) [draw,diamond] {7};
  \node (6) at (29.0bp,242.0bp) [draw,ellipse] {6};
  \draw [->] (7) ..controls (81.17bp,272.81bp) and (78.009bp,227.92bp)  .. (5);
  \draw [->] (7) ..controls (67.735bp,292.3bp) and (56.867bp,278.47bp)  .. (6);
  \draw [->] (3) ..controls (128.37bp,66.468bp) and (148.0bp,120.81bp)  .. (148.0bp,169.0bp) .. controls (148.0bp,243.0bp) and (148.0bp,243.0bp)  .. (148.0bp,243.0bp) .. controls (148.0bp,284.19bp) and (135.03bp,330.49bp)  .. (1);
  \draw [->] (5) ..controls (67.254bp,144.14bp) and (64.647bp,134.54bp)  .. (4);
  \draw [->] (6) ..controls (24.729bp,199.52bp) and (20.538bp,154.5bp)  .. (19.0bp,116.0bp) .. controls (18.362bp,100.01bp) and (18.579bp,95.994bp)  .. (19.0bp,80.0bp) .. controls (19.22bp,71.642bp) and (19.621bp,62.601bp)  .. (0);
  \draw [->] (6) ..controls (44.85bp,216.34bp) and (51.973bp,205.26bp)  .. (5);
  \draw [->] (1) ..controls (104.2bp,359.19bp) and (98.9bp,347.59bp)  .. (7);
  \draw [->] (4) ..controls (43.888bp,72.082bp) and (39.154bp,61.466bp)  .. (0);
  \draw [->] (5) ..controls (83.485bp,141.63bp) and (87.82bp,128.15bp)  .. (91.0bp,116.0bp) .. controls (96.781bp,93.908bp) and (101.87bp,68.483bp)  .. (3);
  \draw [->] (4) ..controls (73.868bp,71.614bp) and (83.789bp,58.266bp)  .. (3);
\end{tikzpicture}
\caption{CFG for $\bigl\{\, \mathtt{while} \; ((\cdot \cor \cdot) \cand (\cdot \cor \cdot)) \; \mathord{-} \,\bigr\} $: $7$~acyclic paths}
\label{fig:cfg18}

\end{figure}
\begin{figure}
  \centering

\begin{tikzpicture}[>=latex,line join=bevel,]
\node (0) at (158.0bp,22.0bp) [draw,circle, double] {0};
  \node (3) at (57.0bp,22.0bp) [draw,ellipse] {3};
  \node (2) at (74.0bp,386.0bp) [draw,ellipse] {2};
  \node (5) at (93.0bp,170.0bp) [draw,ellipse] {5};
  \node (4) at (131.0bp,98.0bp) [draw,ellipse] {4};
  \node (7) at (74.0bp,314.0bp) [draw,diamond] {7};
  \node (6) at (93.0bp,242.0bp) [draw,ellipse] {6};
  \draw [->] (3) ..controls (51.471bp,79.282bp) and (44.248bp,178.49bp)  .. (57.0bp,260.0bp) .. controls (58.563bp,269.99bp) and (61.688bp,280.67bp)  .. (7);
  \draw [->] (7) ..controls (80.196bp,290.17bp) and (83.102bp,279.46bp)  .. (6);
  \draw [->] (5) ..controls (106.42bp,144.28bp) and (112.16bp,133.71bp)  .. (4);
  \draw [->] (2) ..controls (124.95bp,351.46bp) and (186.0bp,302.34bp)  .. (186.0bp,243.0bp) .. controls (186.0bp,243.0bp) and (186.0bp,243.0bp)  .. (186.0bp,169.0bp) .. controls (186.0bp,128.14bp) and (175.29bp,81.971bp)  .. (0);
  \draw [->] (6) ..controls (93.0bp,215.98bp) and (93.0bp,206.71bp)  .. (5);
  \draw [->] (3) ..controls (29.767bp,64.832bp) and (0.0000bp,118.66bp)  .. (0.0bp,169.0bp) .. controls (0.0bp,243.0bp) and (0.0bp,243.0bp)  .. (0.0bp,243.0bp) .. controls (0.0bp,289.41bp) and (31.466bp,335.67bp)  .. (2);
  \draw [->] (2) ..controls (74.0bp,359.98bp) and (74.0bp,350.71bp)  .. (7);
  \draw [->] (5) ..controls (82.595bp,126.8bp) and (70.921bp,79.458bp)  .. (3);
  \draw [->] (4) ..controls (140.13bp,71.98bp) and (143.83bp,61.827bp)  .. (0);
  \draw [->] (4) ..controls (105.57bp,71.567bp) and (90.669bp,56.669bp)  .. (3);
\end{tikzpicture}
\caption{CFG for $\bigl\{\, \mathtt{do} \;\mathord{-}\; \mathtt{while} \; ((\cdot \cor \cdot) \cand (\cdot \cor \cdot)) \,\bigr\} $: $3$~acyclic paths}
\label{fig:cfg28}
\end{figure}
Figures~\ref{fig:cfg18} and~\ref{fig:cfg28} show the difference between
$\mathtt{while}$ and $\mathtt{do-while}$ statements:
for $\mathtt{do-while}$ the backward arc that generates the loop cannot be crossed
in an acyclic path, whereas this is allowed in $\mathtt{while}$ statements as the guard
expression can be evaluated twice also in acyclic paths.
\begin{figure}
  \centering

\begin{tikzpicture}[>=latex,line join=bevel,]
\node (11) at (53.029bp,530.0bp) [draw,diamond] {11};
  \node (10) at (53.029bp,458.0bp) [draw,ellipse] {10};
  \node (1) at (40.029bp,98.0bp) [draw,ellipse] {1};
  \node (0) at (103.03bp,22.0bp) [draw,circle, double] {0};
  \node (2) at (31.029bp,170.0bp) [draw,ellipse] {2};
  \node (5) at (48.029bp,242.0bp) [draw,ellipse] {5};
  \node (4) at (103.03bp,170.0bp) [draw,ellipse] {4};
  \node (7) at (53.029bp,386.0bp) [draw,ellipse] {7};
  \node (6) at (53.029bp,314.0bp) [draw,ellipse] {6};
  \node (9) at (144.03bp,386.0bp) [draw,ellipse] {9};
  \node (8) at (158.03bp,242.0bp) [draw,ellipse] {8};
  \draw [->] (1) ..controls (61.366bp,71.937bp) and (72.879bp,58.414bp)  .. (0);
  \draw [->] (9) ..controls (148.1bp,343.67bp) and (152.49bp,299.21bp)  .. (8);
  \draw [->] (10) ..controls (83.872bp,433.27bp) and (103.34bp,418.3bp)  .. (9);
  \draw [->] (7) ..controls (53.029bp,359.98bp) and (53.029bp,350.71bp)  .. (6);
  \draw [->] (8) ..controls (151.13bp,205.89bp) and (145.14bp,176.85bp)  .. (139.03bp,152.0bp) .. controls (130.64bp,117.92bp) and (119.63bp,79.172bp)  .. (0);
  \draw [->] (11) ..controls (53.029bp,503.98bp) and (53.029bp,494.71bp)  .. (10);
  \draw [->] (2) ..controls (34.199bp,144.35bp) and (35.397bp,135.03bp)  .. (1);
  \draw [->] (10) ..controls (31.303bp,432.03bp) and (21.656bp,418.17bp)  .. (17.029bp,404.0bp) .. controls (-7.8182bp,327.92bp) and (-1.6774bp,302.85bp)  .. (12.029bp,224.0bp) .. controls (13.617bp,214.86bp) and (16.638bp,205.22bp)  .. (2);
  \draw [->] (6) ..controls (51.242bp,287.98bp) and (50.58bp,278.71bp)  .. (5);
  \draw [->] (6) ..controls (70.615bp,287.22bp) and (78.866bp,273.32bp)  .. (84.029bp,260.0bp) .. controls (91.773bp,240.02bp) and (96.644bp,216.23bp)  .. (4);
  \draw [->] (5) ..controls (42.021bp,216.26bp) and (39.728bp,206.82bp)  .. (2);
  \draw [->] (4) ..controls (103.03bp,128.12bp) and (103.03bp,84.348bp)  .. (0);
  \draw [->] (10) ..controls (53.029bp,431.98bp) and (53.029bp,422.71bp)  .. (7);
\end{tikzpicture}
\caption{CFG for $\bigl\{\, \mathtt{switch} \; (\cdot) \; \bigl\{\, \mathtt{case} \, 1\mathord{:} \, \bigl\{\, \mathord{-} \; \mathtt{break}; \,\bigr\}  \; \mathtt{case} \, 2\mathord{:} \, \mathtt{if} \; (\cdot) \; \mathord{-} \; \mathtt{else} \; \bigl\{\, \mathord{-} \; \mathtt{break}; \,\bigr\}  \; \mathtt{default}\mathord{:} \; \mathord{-} \,\bigr\}  \,\bigr\} $: $4$~acyclic paths}
\label{fig:cfg29}
\end{figure}
\begin{figure}
  \centering

\begin{tikzpicture}[>=latex,line join=bevel,]
\node (11) at (110.0bp,530.0bp) [draw,diamond] {11};
  \node (10) at (110.0bp,458.0bp) [draw,ellipse] {10};
  \node (1) at (165.0bp,170.0bp) [draw,ellipse] {1};
  \node (0) at (144.0bp,94.0bp) [draw,circle, double] {0};
  \node (3) at (165.0bp,242.0bp) [draw,ellipse] {3};
  \node (5) at (165.0bp,386.0bp) [draw,ellipse] {5};
  \node (4) at (165.0bp,314.0bp) [draw,ellipse] {4};
  \node (7) at (55.0bp,94.0bp) [draw,ellipse] {7};
  \node (6) at (220.0bp,18.0bp) [draw,ellipse] {6};
  \node (9) at (55.0bp,314.0bp) [draw,ellipse] {9};
  \node (8) at (55.0bp,170.0bp) [draw,ellipse] {8};
  \draw [->] (1) ..controls (157.96bp,144.19bp) and (155.17bp,134.35bp)  .. (0);
  \draw [->] (9) ..controls (55.0bp,271.67bp) and (55.0bp,227.21bp)  .. (8);
  \draw [->] (10) ..controls (94.154bp,416.09bp) and (76.459bp,370.4bp)  .. (9);
  \draw [->] (7) ..controls (105.59bp,70.311bp) and (156.04bp,47.685bp)  .. (6);
  \draw [->] (3) ..controls (128.33bp,217.67bp) and (102.24bp,201.06bp)  .. (8);
  \draw [->] (10) ..controls (173.88bp,429.53bp) and (258.0bp,383.72bp)  .. (258.0bp,315.0bp) .. controls (258.0bp,315.0bp) and (258.0bp,315.0bp)  .. (258.0bp,169.0bp) .. controls (258.0bp,124.54bp) and (242.12bp,74.881bp)  .. (6);
  \draw [->] (5) ..controls (165.0bp,359.98bp) and (165.0bp,350.71bp)  .. (4);
  \draw [->] (8) ..controls (55.0bp,143.06bp) and (55.0bp,132.16bp)  .. (7);
  \draw [->] (10) ..controls (110.0bp,413.29bp) and (110.0bp,360.11bp)  .. (110.0bp,315.0bp) .. controls (110.0bp,315.0bp) and (110.0bp,315.0bp)  .. (110.0bp,241.0bp) .. controls (110.0bp,199.79bp) and (123.01bp,153.71bp)  .. (0);
  \draw [->] (6) ..controls (220.0bp,63.942bp) and (220.0bp,120.82bp)  .. (220.0bp,169.0bp) .. controls (220.0bp,243.0bp) and (220.0bp,243.0bp)  .. (220.0bp,243.0bp) .. controls (220.0bp,283.45bp) and (216.88bp,294.8bp)  .. (201.0bp,332.0bp) .. controls (196.61bp,342.28bp) and (190.12bp,352.66bp)  .. (5);
  \draw [->] (3) ..controls (165.0bp,215.98bp) and (165.0bp,206.71bp)  .. (1);
  \draw [->] (11) ..controls (110.0bp,503.98bp) and (110.0bp,494.71bp)  .. (10);
  \draw [->] (10) ..controls (59.324bp,422.73bp) and (0.0bp,373.63bp)  .. (0.0bp,315.0bp) .. controls (0.0bp,315.0bp) and (0.0bp,315.0bp)  .. (0.0bp,241.0bp) .. controls (0.0bp,200.55bp) and (3.7416bp,189.46bp)  .. (19.0bp,152.0bp) .. controls (23.665bp,140.55bp) and (30.612bp,128.84bp)  .. (7);
  \draw [->] (4) ..controls (165.0bp,287.98bp) and (165.0bp,278.71bp)  .. (3);
  \draw [->] (10) ..controls (129.21bp,432.55bp) and (138.66bp,420.52bp)  .. (5);
\end{tikzpicture}
\caption{CFG for $\bigl\{\, \mathtt{switch} \; (\cdot) \; \mathtt{case} \, 0\mathord{:} \, \mathtt{do} \;\bigl\{\, \mathord{-} \; \mathtt{case} \, 1\mathord{:} \, \mathord{-} \; \mathtt{case} \, 2\mathord{:} \, \mathord{-} \; \mathtt{case} \, 3\mathord{:} \, \mathord{-} \,\bigr\} \; \mathtt{while} \; (\cdot) \,\bigr\} $: $5$~acyclic paths}
\label{fig:cfg27}
\end{figure}
Figures~\ref{fig:cfg29} and~\ref{fig:cfg27} illustrates CFGs generated from
\texttt{switch} statements: the former shows the effect of
\texttt{break} statements
in switches, the latter is a reduced version of Duff's device.
\begin{figure}
  \centering

\begin{tikzpicture}[>=latex,line join=bevel,]
\node (11) at (27.575bp,386.0bp) [draw,diamond] {11};
  \node (10) at (55.575bp,314.0bp) [draw,ellipse] {10};
  \node (1) at (27.575bp,98.0bp) [draw,ellipse] {1};
  \node (0) at (27.575bp,22.0bp) [draw,circle, double] {0};
  \node (7) at (27.575bp,170.0bp) [draw,ellipse] {7};
  \node (8) at (27.575bp,242.0bp) [draw,ellipse] {8};
  \draw [->] (1) ..controls (27.575bp,72.165bp) and (27.575bp,62.878bp)  .. (0);
  \draw [->] (11) ..controls (36.418bp,362.89bp) and (40.996bp,351.45bp)  .. (10);
  \draw [->] (10) ..controls (45.661bp,288.22bp) and (41.631bp,278.14bp)  .. (8);
  \draw [->] (11) ..controls (22.785bp,359.21bp) and (20.536bp,344.81bp)  .. (19.575bp,332.0bp) .. controls (18.377bp,316.04bp) and (18.377bp,311.96bp)  .. (19.575bp,296.0bp) .. controls (20.211bp,287.52bp) and (21.413bp,278.34bp)  .. (8);
  \draw [->] (7) ..controls (27.575bp,143.98bp) and (27.575bp,134.71bp)  .. (1);
  \draw [->] (8) ..controls (27.575bp,215.98bp) and (27.575bp,206.71bp)  .. (7);
\end{tikzpicture}
\caption{CFG for $\bigl\{\, \mathtt{if} \; (\cdot) \; \mathtt{goto} \; \mathtt{l1}; \; \mathtt{else} \; \mathtt{l1}\mathord{:} \;\mathtt{goto} \; \mathtt{l2}; \; \mathtt{while} \; (\cdot) \; \mathord{-} \; \mathtt{l2}\mathord{:} \;\mathord{-} \,\bigr\} $: $2$~acyclic paths}
\label{fig:cfg24}
\end{figure}
Figure~\ref{fig:cfg24} shows the CFG genereted for a program containing
a nasty use of \texttt{goto} statements, which is perfectly legal in C:
jumping from one of the branches to the other in if-then-else statements.

\end{document}